%% file: main.tex
\documentclass{lmcs}
\pdfoutput=1
\usepackage[utf8]{inputenc}

\usepackage{lastpage}
\lmcsdoi{21}{2}{23}
\lmcsheading{}{\pageref{LastPage}}{}{}%
{Feb.~15,~2024}{Jun.~16,~2025}{}


\usepackage{hyperref}
\theoremstyle{plain} %\crefname{satz}{Satz}{S\"atze}
\input{myMacros.tex}

\begin{document}

\title[Conjunctive Queries with Free Access Patterns under Updates]{Conjunctive Queries with Free Access Patterns \\ under Updates}

%\titlecomment{{\lsuper*} This work is an extended version of a paper with the same title presented at ICDT 2023~\cite{ICDT-23-version}.
%This work is an extended version of the paper OPTIONAL comment concerning the title, \eg,
 % if a variant or an extended abstract of the paper has appeared elsewhere.
% }
%\thanks{thanks, optional.}	%optional

% affiliations are numbered automatically with a, b, c (see below)
% use the optional argument to indicate the affiliation(s) of each author
% omit the argument if there is only one author, or only one affiliation
\author[A.~Kara]{Ahmet Kara\lmcsorcid{0000-0001-8155-8070}}[a]
\author[M.~Nikolic]{Milos Nikolic\lmcsorcid{0000-0002-1548-6803}}[b]
\author[D.~Olteanu]{Dan Olteanu\lmcsorcid{0000-0002-4682-7068}}[c]
\author[H.~Zhang]{Haozhe Zhang\lmcsorcid{0000-0002-0930-1980}}[c]

% affiliation 1 (automatically numbered a)
\address{OTH Regensburg}	%optional
% write emails for all authors having that affiliation
\email{ahmet.kara@oth-regensburg.de}  %optional

% affiliation 2 (automatically numbered b)
\address{University of Edinburgh}	%optional
\email{milos.nikolic@ed.ac.uk}  %optional

% affiliation 3 (automatically numbered c)
\address{University of Zurich}	%optional
% write emails for all authors having that affiliation
\email{dan.olteanu@uzh.ch, haozhe.zhang@uzh.ch}  %optional

%% etc.

%% required for running head on odd and even pages, use suitable
%% abbreviations in case of long titles and many authors:

%%%%%%%%%%%%%%%%%%%%%%%%%%%%%%%%%%%%%%%%%%%%%%%%%%%%%%%%%%%%%%%%%%%%%%%%%%%

%% the abstract has to PRECEDE the command \maketitle:
%% be sure not to issue the \maketitle command twice!

\begin{abstract}
  \noindent 
  We study the problem of answering conjunctive queries with free access patterns (CQAPs) under updates. 
  A free access pattern is a partition of the free variables of the query into input and output.
  The query returns tuples over the output variables given a tuple of values over the input variables.

  We introduce a fully dynamic evaluation approach that works for all CQAPs and is optimal for two classes of CQAPs. This approach recovers prior work on the dynamic evaluation of conjunctive queries without access patterns.

  We first give a syntactic characterisation of all CQAPs that admit constant time per single-tuple update and whose output tuples can be enumerated with constant delay given a tuple of values over the input variables.

  We further chart the complexity trade-off between the preprocessing time, update time and  enumeration delay for a class of CQAPs.
  For some of these CQAPs, our approach achieves optimal, albeit non-constant, update time and delay. This optimality is predicated on the Online Matrix-Vector Multiplication conjecture.
  
  We finally adapt our approach to the dynamic evaluation of tractable CQAPs over probabilistic databases under updates.
\end{abstract}

\maketitle

\input{introduction}
\input{preliminaries}
\input{queries}
\input{results}
\input{preprocessing}
\input{enumeration}
\input{updates}
\input{discussion}
\input{dichotomy}

\input{tradeoff}
\input{probsemantics}
\input{probabilistic}

\input{related}

\input{conclusion}

\section*{Acknowledgment}
\noindent 
The authors would like to thank the anonymous reviewers for their valuable suggestions, which have significantly improved this article.
The authors would also like to acknowledge the UZH Global Strategy and Partnerships Funding Scheme and EPSRC grant EP/T022124/1.
This project has received funding from the European Union’s Horizon 2020 research and innovation programme under grant agreement No 682588.

\bibliographystyle{alphaurl}
\bibliography{bibliography}

\end{document}

%% file: myMacros.tex
\usepackage{pgfplots}
\pgfplotsset{compat=1.14}
\usepackage{tikz-3dplot}
\usepackage{xspace}
\usepackage{tikz}
\usepackage{multirow}
\usepackage{marvosym}
\usepackage{enumitem}
\usepackage{amsthm}
\usepackage{amssymb}
\usepackage{mathtools}
\usepackage{amsfonts}
\usepackage{booktabs}
\usepackage{stmaryrd}

\pgfplotsset{compat=1.16}
\usepackage{tikz-3dplot}
\usepackage{xspace}
\usepackage{tikz}
\pgfdeclarelayer{background}
\pgfsetlayers{background,main}
\usetikzlibrary{patterns}
\usetikzlibrary{decorations.pathmorphing, decorations.markings}

\usepackage{xcolor,colortbl}

\definecolor{light-gray}{gray}{0.7.2}
\definecolor{goodgreen}{rgb}{0.1, 0.5, 0.1}
\definecolor{burntorange}{rgb}{0.8, 0.33, 0.0}
\newcommand{\vars}{\mathit{vars}}
\newcommand{\free}{\mathit{free}}

\newcommand{\atoms}{\mathit{atoms}}
\newcommand{\dep}{\textit{dep}}

\newcommand{\gyo}{\text{GYO}}

\newcommand{\bigO}[1]{\mathcal{O}(#1)}

\newcommand{\VO}{\mathsf{VO}}
\newcommand{\canonicalVO}{\mathsf{canonVO}}
\newcommand{\freeTopVO}{\mathsf{freeTopVO}}

\newcommand{\fw}{\mathsf{w}}
\newcommand{\dfw}{\delta}
\newcommand{\freetop}{\mathsf{free}\text{-}\mathsf{top}}
\newcommand{\inputtop}{\mathsf{input}\text{-}\mathsf{top}}
\newcommand{\acceff}{\mathsf{acc}\text{-}\mathsf{top}}

\newcommand{\linenumber}{\makebox[2ex][r]{\rownumber\TAB}}

\newcommand{\eps}{\epsilon}

\newcommand{\Dom}{\mathsf{Dom}}
\newcommand{\inst}[1]{\mathbf{#1}}
\newcommand{\tup}[1]{\mathbf{#1}}

\newcommand{\OMv}{\textsf{OMv}\xspace}

\newcommand{\anc}{\mathsf{anc}}

\newcommand{\signs}{\textit{sig}}

\newcommand{\schema}[1]{\mathsf{Sch}(#1)}

\newcommand{\veryshortarrow}[1][3pt]{\mathrel{%
   \hbox{\rule[\dimexpr\fontdimen22\textfont2-.2pt\relax]{#1}{.4pt}}%
   \mkern-4mu\hbox{\usefont{U}{lasy}{m}{n}\symbol{41}}}}

\newcommand{\shortarrow}[1][5pt]{\mathrel{%
  \hspace{1pt}  \hbox{\rule[\dimexpr\fontdimen22\textfont2-.2pt\relax]{#1}{.4pt}}%
   \mkern-4mu\hbox{\usefont{U}{lasy}{m}{n}\symbol{41}}}}

\newcommand{\floor}[1]{\left\lfloor #1 \right\rfloor}

\newcommand{\vt}{\tau}

\newcommand{\vos}{\Omega}
\newcommand{\ftvo}{\textsc{Access-Top}}
\newcommand{\delete}{\Delta}

\newcommand{\calH}{\mathcal{H}}
\newcommand{\calT}{\mathcal{T}}
\newcommand{\calR}{\mathcal{R}}
\newcommand{\calB}{\mathcal{B}}
\newcommand{\calA}{\mathcal{A}}
\newcommand{\calF}{\mathcal{F}}

\newcommand{\calD}{\mathcal{D}}
\newcommand{\calE}{\mathcal{E}}
\newcommand{\calS}{\mathcal{S}}
\newcommand{\calX}{\mathcal{X}}
\newcommand{\calY}{\mathcal{Y}}
\newcommand{\calV}{\mathcal{V}}
\newcommand{\calI}{\mathcal{I}}

\newcommand{\calZ}{\mathcal{Z}}

\newcommand{\calO}{\mathcal{O}}

\newcommand{\TAB}{\makebox[2.5ex][r]{}}%
\newcommand{\LET}{\textbf{let}\xspace}%
\newcommand{\IF}{\textbf{if}\xspace}%
\newcommand{\ELSE}{\textbf{else}\xspace}%
\newcommand{\FOREACH}{\textbf{foreach}\xspace}%
\newcommand{\AND}{\textbf{and}\xspace}%
\newcommand{\OR}{\textbf{or}\xspace}%
\newcommand{\RETURN}{\textbf{return}\xspace}%
\newcommand{\MATCH}{\textbf{switch}\xspace}%
\newcommand{\EOF}{\textbf{EOF}\xspace}%

\newcommand{\nop}[1]{}

\usepackage{todonotes}

\nop{

}

\newcounter{magicrownumbers}
\newcommand\rownumber{\footnotesize\stepcounter{magicrownumbers}\arabic{magicrownumbers}}

\nop{
\theoremstyle{plain}                  
\newtheorem{theorem}[thm]{Theorem}

\newtheorem{proposition}[theorem]{Proposition}

}

\newtheoremstyle{cited}%
{.5\baselineskip\@plus.2\baselineskip
    \@minus.2\baselineskip}% (space above)
{.5\baselineskip\@plus.2\baselineskip
    \@minus.2\baselineskip}% (space below)
{\itshape}% (body font)
{\parindent}% (indent amount)
{}% {theorem head font}
{.}% {punctuation after theorem head}
{.5em}% {space after theorem head}
{\textsc{\thmname{#1}} \thmnote{\normalfont#3}}% {theorem head spec}

%% file: introduction.tex
\section{Introduction}
\label{sec:introduction}

We consider the problem of dynamic evaluation for conjunctive queries with access restrictions. Restricted access to data is commonplace~\cite{Nash:FOAccess:04, Nash:UCQAccess:04, Li:ICDT:2001}: For instance, the flight information behind a user-interface query can only be accessed by providing values for specific input fields such as the departure and destination airports in a flight booking database.

We formalise such queries as \textbf{C}onjunctive \textbf{Q}ueries with free \textbf{A}ccess \textbf{P}atterns (CQAPs for short): The free variables of a CQAP are partitioned into {\em input} and {\em output}. The query yields tuples of values over the output variables {\em given} a tuple of values over the input variables.

\begin{exa}
\label{ex:flight}
Assume that a flight booking company has a (simplified) database consisting of the two relations 
$\texttt{Flight}$ and $\texttt{Airport}$.
The relation $\texttt{Flight}$ contains information about flights, including flight numbers, the departure and arrival airports, and the date of the flights. The relation $\texttt{Airport}$ contains the names of airports and the cities in which they are located.
Assume that the company 
provides a user-interface where users can search for flights by 
specifying 
a departure city \texttt{depCity}, an arrival city \texttt{arrCity}, and a departure date \texttt{date}. 
Given such a triple of inputs, the user-interface queries the database 
and lists all flight numbers \texttt{flightNo} together with the corresponding departure and arrival airports, \texttt{depAirport} and \texttt{arrAirport}, that match the 
user request. 
We formalise this data access using the following CQAP: 
\begin{align*}
\texttt{FlightSearch}(\texttt{flightNo}, &\ \texttt{depAirport}, \texttt{arrAirport} \mid 
\texttt{depCity}, \texttt{arrCity}, \texttt{date}) = \\
&  \texttt{Flight}(\texttt{flightNo}, \texttt{depAirport}, \texttt{arrAirport}, \texttt{date}), \\
&  \texttt{Airport}(\texttt{depAirport}, \texttt{depCity}), \
\texttt{Airport}(\texttt{arrAirport},  \texttt{arrCity})
\end{align*}
The variables in the query head that appear after the symbol $|$, i.e.,   
\texttt{depCity}, \texttt{arrCity}, and \texttt{date}, are input variables.
The other variables in the head, i.e.,  
\texttt{flightNo}, \texttt{depAirport}, and \texttt{arrAirport},
are output variables. 
If a user is interested in all flights from Edinburgh to Zurich on the 1st of January 2024, 
the user interface runs the query $\texttt{FlightSearch}$ after setting the input variables to the 
values \texttt{"Edinburgh"}, \texttt{"Zurich"}, and \texttt{"2024-01-01"}.
\qed
\end{exa}

In database systems, CQAPs formalise the notion of parameterized queries (or prepared statements)~\cite{AbiteboulHV95}. In probabilistic graphical models, they correspond to conditional queries ~\cite{PGM:Book:2009}: 
Such inference queries ask for (the probability of) each possible value of a tuple of random variables (corresponding to CQAP output variables) given specific values for a tuple of random variables (corresponding to CQAP input variables). 

Prior work on queries with access patterns considered a more general setting than CQAP: There, each relation in the query body may have input and output variables such that values for the latter can only be obtained if values for the former are supplied~\cite{Florescu:SIGMOD:99, Yerneni:ICDT:99, Deutsch:TCS:2007, Benedikt:VLDB:15, Benedikt:PODS:14}. In this more general setting, and in sharp contrast to our simpler setting, a fundamental question is whether the query can even be answered for a given access pattern to each relation~\cite{Nash:FOAccess:04, Nash:UCQAccess:04, Li:ICDT:2001}.

We introduce a fully dynamic evaluation approach for CQAPs. It is fully dynamic in the sense that it supports both inserts and deletes of tuples to the input database. It computes a data structure that supports the enumeration of the distinct output tuples for any values of the input variables and maintains this data structure under  updates to the input database.

Our analysis of the overall computation time is refined into three components.
The {\em preprocessing time} is the time to compute the data structure before receiving any updates.
Given a tuple over the input variables, the {\em enumeration delay} is the time between the start of the enumeration process and the output of the first tuple, the time between outputting any two consecutive tuples, and the time between outputting the last tuple and the end of the enumeration process~\cite{DurandFO07}. 
The {\em update time} is the time used to update the data structure\footnote{We do not allow updates during the enumeration; this functionality is orthogonal to our contributions and can be supported using a versioned data structure.} for one single-tuple update.

There are simple, albeit more expensive alternatives to our approach. For instance, on an update request we may only update the input database, and on an enumeration request we may use an existing enumeration algorithm for the residual query obtained by setting the input variables to constants in the original query. However, such an approach needs time-consuming and independent preparation for each enumeration request, e.g., to remove dangling tuples and possibly create a data structure to support enumeration. In contrast, the data structure constructed by our approach shares this preparation across the enumeration requests and can readily serve enumeration requests for any values of the input variables.

The contributions of this paper are as follows.

Section~\ref{sec:cqap} introduces the language of CQAPs. Two new notions account for the nature of free access patterns:  {\em access-top variable orders} and {\em query fractures}. 

An access-top variable order is a  decomposition of the query into a forest of (rooted) trees with one node per variable, where: the input variables are above all other variables; and the free (input and output) variables are above the bound variables. This variable order is compiled into a forest of view trees, which is a data structure that represents compactly the query output.

Since access to the query output requires fixing values for the input variables, the query can be fractured by breaking its joins on the input variables and replacing each of their occurrences with fresh variables within each connected component of the query hypergraph. This does not violate the access pattern, since each fresh input variable is set to the corresponding given input value. Yet this may lead to structurally simpler queries whose fully dynamic evaluation admits lower complexity.

Section~\ref{sec:var_orders_width_measures} introduces the {\em static} and {\em dynamic} widths that capture the complexities of the preprocessing and respectively update steps. For a given CQAP, these widths are defined over the possible access-top variable orders of the fracture of the query.

Section~\ref{sec:results} overviews our main results on the complexity of dynamic CQAP evaluation.

Sections~\ref{sec:preprocessing}-\ref{sec:updates} introduce our approach for dynamic CQAP evaluation. Computing and maintaining each view in the view tree accounts for preprocessing and respectively updates, while the view tree as a whole allows for the enumeration of the output tuples with constant delay.
Section~\ref{sec:discussion} discusses key decisions behind our approach.

Section~\ref{sec:dichotomy} gives a syntactic characterisation of those CQAPs that admit linear-time preprocessing and constant-time update and enumeration delay. We call this class of well-behaved queries $\text{CQAP}_0$. All queries outside $\text{CQAP}_0$ 
and without repeating relation symbols 
 do not admit constant-time update and delay regardless of the preprocessing time, unless the widely held Online Matrix-Vector Multiplication conjecture~\cite{Henzinger:OMv:2015} fails. This dichotomy generalises a prior dichotomy for $q$-hierarchical queries {\em without access patterns}~\cite{BerkholzKS17}. The $q$-hierarchical queries are in $\text{CQAP}_0$ and have no input variables. The class $\text{CQAP}_0$ further contains  cyclic CQAPs with input variables. 
\begin{exa}
\label{ex:triangle_detection} 
The following triangle detection problem is in $\text{CQAP}_0$:
 Given three nodes in a graph, we ask whether the nodes form a triangle. 
 This problem can be expressed by the following query in 
$\text{CQAP}_0$: 
 \begin{align*}
 Q(\cdot | A,B,C) = E(A,B), E(B,C), E(C,A), 
 \end{align*}
 where $E$ is the edge relation of the graph.
 All variables of the query are free and input. The dot ($\cdot$) in the query head signalises that the query does not have output variables.   
\qed
\end{exa}
The  smallest query patterns not in $\text{CQAP}_0$ strictly include the non-$q$-hierarchical ones and also others that are sensitive to the interplay of the output and input variables. 

Section~\ref{sec:trade-off} charts the preprocessing time - update time - enumeration delay trade-off for the dynamic evaluation of CQAPs whose fractures are hierarchical. It shows that by increasing the preprocessing and update times, we can decrease the enumeration delay. Our trade-off reveals the optimality for a particular class of CQAPs with hierarchical fractures, called $\text{CQAP}_1$, which lies outside $\text{CQAP}_0$. 
\begin{exa}
\label{ex:edge_triangle_listing}  
The following edge triangle listing problem is in $\text{CQAP}_1$:
Given an edge in a graph, the task is to list all triangles containing this edge. 
This problem can be expressed by the following query in $\text{CQAP}_1$: 
\begin{align*}
Q(C | A,B) = E(A,B), E(B,C), E(C,A). \tag*{\qed}
\end{align*}
\end{exa}
The complexity of $\text{CQAP}_1$ for both the update time and the enumeration delay matches the (conditional) lower bound $\Omega(N^{\frac{1}{2}})$ for queries outside $\text{CQAP}_0$, where $N$ is the size of the input database. 
This is weakly Pareto optimal, as there can be no tighter upper bounds for {\em both} the update time and the enumeration delay (though it does not rule out the possibility that one of the two times can be lowered).
Our approach for $\text{CQAP}_1$ exhibits a continuum of trade-offs: $\bigO{N^{1+\eps}}$ preprocessing time, $\bigO{N^{\eps}}$ amortized update time and $\bigO{N^{1-\eps}}$ enumeration delay, for every $\eps\in[0,1]$. By tweaking the parameter $\eps$, one can optimise the overall time for a sequence of enumeration and update tasks and achieve an asymptotically lower compute time than prior work (Section~\ref{sec:comparison}).
Our approach recovers the complexity of the well-studied dynamic set intersection problem~\cite{DBLP:conf/wads/KopelowitzPP15}:
\begin{exa}
\label{ex:set_intersection}  
The dynamic set intersection problem is defined as follows. We are given sets $S_1, \ldots, S_m$ that are subject to element insertions and deletions. For each access request $(i, j)$ with $i,j \in [m]$, we need to decide whether the intersection of the sets $S_i$ and $S_j$ is empty. 
Consider a relation $S$ that consists of the tuples $\{(i,x) \mid x \in S_i \text{ and } i \in \{1, \ldots , m\}\}$. 
The dynamic set intersection problem is expressed by the following query in 
$\text{CQAP}_1$:
\begin{align*}
Q(\cdot | B,C) = S(B,A), S(C,A).
\end{align*}
The variables $B$ and $C$ are free and input. The query has no output variables.  
Prior work designs a randomised algorithm for this problem that uses  expected $\bigO{N^{\frac{1}{2}}}$ update time and enumeration delay, where $N$ is the size of the sets~\cite{DBLP:conf/wads/KopelowitzPP15}. 
Our approach recovers these complexities using our deterministic algorithm and $\eps=\frac{1}{2}$.
\qed
\end{exa}

Our dynamic evaluation approach for CQAPs can be applied to further domains.
Section~\ref{sec:probabilistic-semantics} discusses
three possible semantics for updates in probabilistic databases: set, bag, and expectation-variance.
In these probabilistic settings, an update can be an insertion or a deletion of a tuple with an arbitrary probability.
Section~\ref{sec:probabilistic} shows how to maintain hierarchical conjunctive queries without repeating relation symbols 
with constant update time and enumeration delay under the set semantics and the expectation-variance semantics for updates to the underlying probabilistic database.

\begin{exa}
    The CQAP language naturally expresses conditional queries over probabilistic databases, asking for the probability of a certain outcome {\em given} specific values for the input variables. Consider the flight search query in Example~\ref{ex:flight} and a probabilistic version of the relation \texttt{Flight}, which encodes the probability of each flight taking place based on historical evidence. The query returns tuples of flight number, departure and arrival airports, given date and departure and arrival cities, {\em together} with the probability for the flight to happen.

    Consider now a probabilistic graph, where each edge has a probability for being in the graph. The edge relation of the graph has one tuple per probabilistic edge. Then, the CQAP in Example~\ref{ex:triangle_detection} returns the probability that three given vertices $(a,b,c)$ form a triangle in the graph. The CQAP in Example~\ref{ex:edge_triangle_listing} gives, for each input edge $(a,b)$, each $C$-node $c$ with which it forms a triangle $(a,b,c)$ and the probability of that triangle. \qed
\end{exa}

A prior version of this work appeared in ICDT 2023~\cite{ICDT-23-version}. This article extends the prior version as follows.
We illustrate CQAPs using Examples \ref{ex:flight}-\ref{ex:set_intersection} and variable orders using Examples \ref{ex:widths-four-cycle} and \ref{ex:triangle_three_tails}. 
We include the proof of Theorem~\ref{thm:general}, which states the complexity of  our approach for the dynamic evaluation of arbitrary CQAPs.
Theorem~\ref{thm:general} is now implied by the new Propositions~\ref{prop:preprocessing}, 
\ref{prop:enumeration}, and \ref{prop:updates}.
We also include the proof of Theorem~\ref{thm:dichotomy}, which characterises 
the class of CQAPs that admit linear preprocessing time, constant update time, and constant enumeration delay.
Furthermore, we explain in more detail the adaptive evaluation strategy that achieves the preprocessing - update - enumeration trade-offs for CQAPs with hierarchical fractures (Sections~\ref{sec:trade-off-preprocessing} -- \ref{sec:trade-off-updates}). 
Finally, we introduce three update semantics for probabilistic databases (Section~\ref{sec:probabilistic-semantics}) and show that our approach maintains 
queries in CQAP$_0$ over probabilistic databases with constant update time and enumeration delay under two of these update semantics (Section~\ref{sec:probabilistic}). 
Due to lack of space, some proofs and technical details are deferred to the technical report~\cite{access_pattern_arxiv}. 

%% file: preliminaries.tex
\section{Preliminaries}
\label{sec:preliminaries}

\paragraph{Data Model}
A schema $\calX = (X_1, \ldots, X_n)$ is a tuple of distinct variables.
Each variable $X_i$ has a discrete domain $\Dom(X_i)$.
We treat schemas and sets of variables interchangeably, assuming a fixed ordering of variables.
A tuple $\tup{x}$ of values has schema $\calX=\schema{\tup{x}}$ and is an element from 
$\Dom(\calX) = \Dom(X_1) \times \dots \times \Dom(X_n)$.
A relation $R$  over schema $\mathcal{X}$ is a function
$R: \Dom(\mathcal{X}) \to \mathbb{Z}$  such that
the multiplicity $R(\inst{x})$ is non-zero for finitely many tuples $\inst{x}$.
A tuple $\inst{x}$ is in $R$, denoted by $\inst{x} \in R$, if $R(\inst{x}) \neq 0$.
The size $|R|$ of $R$ is the size of the set $\{ \inst{x} \mid \inst{x} \in R \}$.
A database is a set of relations and has size given by the sum of the sizes of its relations.
Given a tuple $\inst{x}$ over schema  $\mathcal{X}$ and  $\mathcal{S}\subseteq\mathcal{X}$,
$\inst{x}[\mathcal{S}]$ is the restriction of
$\inst{x}$ onto $\calS$.
For a relation $R$ over schema  $\mathcal{X}$, schema $\mathcal{S}\subseteq\mathcal{X}$,
and tuple $\inst{t} \in \Dom(\mathcal{S})$: 
$\sigma_{\mathcal{S} = \inst{t}} R =
	\{\, \inst{x} \,\mid\, \inst{x} \in R \land \inst{x}[\mathcal{S}] = \inst{t} \,\}$ is the set of tuples in $R$
that agree with $\inst{t}$ on the variables in $\calS$;
$\pi_{\mathcal{S}}R = \{\, \inst{x}[\mathcal{S}] \,\mid\, \inst{x} \in R \,\}$ 
stands for the set of tuples in $R$ projected onto $\calS$, i.e., the set of distinct $\calS$-values from the tuples in $R$ with non-zero multiplicities.
For a relation $R$ over schema $\calX$ and $\calY\subseteq\calX$, 
the {\em indicator projection} $I_{\calY}R$ is a relation
over $\calY$ such that~\cite{FAQ:PODS:2016}:
\begin{align*}
\text{ for all } \inst{y} \in \Dom(\calY)  : I_{\calY} R(\inst{y}) = 
\begin{cases}
1 & \text{if there is } \inst{t} \in R \text{ such that } \inst{y} = \inst{t}[\calY] \\
0 & \text{otherwise }
\end{cases}
\end{align*}  
That is, the indicator projection $I_{\calY} R$ is a relation mapping the tuples from $\pi_{\calY}R$ to 1.

\paragraph{Updates.}
An update is a relation,  where tuples with positive multiplicities represent inserts and tuples with negative multiplicities represent deletes. 
Consider a relation $R$ and an update $\delta R$ over the same schema $\calX$.
To apply the update $\delta R$ to $R$ means to compute their {\em union}  
$R\uplus \delta{R}$ defined as:
\begin{align*}
(R\uplus \delta{R})(\tup{x}) = R(\tup{x}) + \delta R(\tup{x}), \text{ for } \tup{x} \in \Dom(\calX).
\end{align*}
A single-tuple update to relation $R$ is a singleton relation $\delta R = \{\tup{x} \rightarrow m\}$, where the multiplicity $m=\delta R(t)$ of the tuple $t$ in $\delta R$ is non-zero.

Updates to input relations may cause changes to indicator projections. 
Applying the single-tuple update $\delta R$ to $R$ triggers a single-tuple update $\delta I_{\calY}R = \{\inst{x}[\calY] \rightarrow k\}$ to $I_{\calY}R$ in the following two cases.
If $\delta R$ is an insertion and $\inst{x}[\calY]$ is a value not already in $\pi_{\calY}R$, then the new update $\delta I_{\calY}R$ is triggered with $k=1$. If $\delta R$ is a deletion and $\pi_{\calY}R$ does not contain $\inst{x}[\calY]$ after applying the update to $R$, then the new update $\delta I_{\calY}R$ is triggered with $k=-1$.

\paragraph{Computational Model}
We consider the RAM model of computation where 
schemas and 
data values are stored in registers 
of logarithmic size and operations on them can be done in constant time\footnote{In this article, we use data complexity: The complexity factors that only depend on the query, such as the number of variables in a query atom and the number of query atoms, are considered constant.}.
We assume that each relation 
$R$  over schema $\mathcal{X}$ is implemented by a data structure that stores key-value entries $(\inst{x},R(\inst{x}))$ for each tuple $\inst{x}$ with $R(\inst{x}) \neq 0$ and needs $O(|R|)$ space. 
This data structure can:
(1) look up, insert, and delete entries in constant time,
(2) enumerate all stored entries in $R$ with constant delay, and
(3) report $|R|$ in constant time.
For a schema $\mathcal{S} \subset \mathcal{X}$, 
we use an index data structure that for any $\tup{t} \in \Dom(\mathcal{S})$ can: 
(4) enumerate all tuples in $\sigma_{\mathcal{S}=\tup{t}}R$ with constant delay,
(5) check $\tup{t} \in \pi_{\mathcal{S}}R$ in constant time; 
(6) return $|\sigma_{\mathcal{S}=\tup{t}}R|$ in constant time; and
(7) insert and delete index entries in constant time.

%% file: queries.tex
\section{Conjunctive Queries with Free Access Patterns}
\label{sec:cqap}
We introduce the queries investigated in this paper along with several of their properties.
A {\em conjunctive query with free access patterns} (CQAP for short) has the form 
\begin{align}
	Q(\calO | \calI) = R_1(\calX_1), \ldots, R_n(\calX_n). \label{eq:CQAP}
\end{align}
We denote by:
$(R_i)_{i\in[n]}$ the relation symbols;
$(R_i(\calX_i))_{i\in[n]}$ the atoms;
$\vars(Q) = \bigcup_{i\in[n]}\calX_i$ the set of variables;
$\atoms(X)$ the set of the atoms containing the variable $X$;
$\atoms(Q)=\{R_i(\calX_i) \mid i\in[n]\}$ the
set of all atoms;
and $\free(Q)=\calO\cup\calI\subseteq\vars(Q)$ the set of {\em free} variables, which are partitioned into {\em input} variables $\calI$ and {\em output} variables $\calO$. An empty set of input or output variables is denoted by a dot ($\cdot$).
All variables in $\vars(Q) \setminus \free(Q)$ are called {\em bound}.
We call $R_1(\calX_1), \ldots, R_n(\calX_n)$ the {\em body} of $Q$.

The hypergraph of a query $Q$ is $\calH=(\calV=\vars(Q),\calE = \{ \calX_i \mid i\in[n]\})$, whose  vertices are the variables and hyperedges are the schemas of the atoms in $Q$.
The {\em fracture} of a CQAP $Q$ is a CQAP $Q_\dagger$ constructed as follows. We start with $Q_\dagger$ as a copy of $Q$. We replace each occurrence of an input variable by a fresh variable. Then, we compute the connected components of the hypergraph of the modified query. Finally, we replace in each connected component of the modified query all new variables originating from the same input variable by one fresh input variable.

 We next define  the notion of dominance for variables in a CQAP $Q$.
 For variables $A$ and $B$, we say that $B$ {\em dominates} $A$ if $\atoms(A) \subset \atoms(B)$.
 The query $Q$ is {\em free-dominant} ({\em input-dominant}) if for any two variables $A$ and $B$, it holds: if $A$ is free (input) and $B$ dominates $A$, then $B$ is free (input).
 The query $Q$ is {\em almost free-dominant} ({\em almost input-dominant}) if: (1) For any variable $B$ that is not  free (input) and for any atom $R(\calX)\in\atoms(B)$, there is an atom $S(\calY)\in\atoms(B)$, possibly different from 
 $R(\calX)$, such that $\calX \cup \calY$ cover all  free (input) variables dominated by $B$; (2) $Q$ is not already free-dominant (input-dominant).
 A query $Q$ is {\em hierarchical} if for any $A, B \in \vars(Q)$, either $\atoms(A) \subseteq \atoms(B)$, $\atoms(B) \subseteq \atoms(A)$, or $\atoms(B) \cap \atoms(A) = \emptyset$. 
 The class of hypergraphs of hierarchical queries is strictly contained in the class of $\gamma$-acyclic (hence, $\alpha$- and $\beta$-acyclic) hypergraphs. 
The class of Berge-acyclic hypergraphs and the class of hypergraphs of hierarchical queries are incomparable: there are Berge-acyclic queries that are not hierarchical, e.g., $Q() = R(A),S(A,B),T(B)$, and hierarchical queries that are not Berge-acyclic, e.g., $Q() = R(A,B),S(A,B)$.   
For the precise definitions of these acyclicity notions, 
we refer to a recent overview~\cite{Brault-Baron16}.
 A query is $q$-hierarchical if it is hierarchical and free-dominant.

\begin{defi}
\label{def:CQAP_zero}\label{def:CQAP_one}
A query is in $\text{CQAP}_0$ if its fracture is hierarchical, free-dominant, and input-dominant.
A query is in $\text{CQAP}_1$ if its fracture is hierarchical and is almost free-dominant, or almost input-dominant, or both.
\end{defi}

The subset of $\text{CQAP}_0$ without input variables is the class of $q$-hierarchical queries~\cite{BerkholzKS17}.

\begin{exa}\label{ex:illustrate_cqap_0}
The query $Q_0(B,C \hspace{-1pt} \mid\hspace{-2.5pt} \cdot ) = R(A,B), S(A,C)$ is hierarchical and  input-dominant.
It is not free-dominant: The bound variable $A$ dominates the free variables $B$ and $C$.

	The query $Q_1(A,C \mid B,D) = R(A,B), S(B,C), T(C,D), U(A,D)$ is input-dominant, free-dominant, but not hierarchical. Its fracture $Q_\dagger(A,C \mid B_1,B_2,D_1,D_2)$ 
$=$ $R(A,B_1),$ $S(B_2,C),$ $T(C,D_1),$ $U(A,D_2)$ is hierarchical but not input-dominant: $C$ dominates both $B_2$ and $D_1$ and $A$ dominates both $B_1$ and $D_2$, yet $A$ and $C$ are not input variables. It is however almost input-dominant: $A$ is not input and for any of its atoms $R(A,B_1)$ and $U(A,D_2)$, there is another atom $U(A,D_2)$ and respectively $R(A,B_1)$ such that both $R(A,B_1)$ and $U(A,D_2)$ cover the variables $B_1$ and $D_2$ dominated by $A$; a similar reasoning applies to $C$. This means that $Q_1$ is in $\text{CQAP}_1$.

	The query $Q_2(A\mid B) = S(A,B),T(B)$ is in $\text{CQAP}_0$, since its fracture $Q_\dagger(A\mid B_1,B_2) = S(A,B_1),T(B_2)$ is hierarchical, free-dominant, and input-dominant.
	
	The query $Q_3(B\mid A) = S(A,B),T(B)$ is in $\text{CQAP}_1$. Its fracture is the query itself. It is hierarchical, yet not input-dominant, since $B$ dominates $A$ and is not input. It is, however, almost input-dominant: for each atom of $B$, there is one other atom such that together they cover $A$. Indeed, atom $S(A,B)$ already covers $A$, and it also does so together with $T(B)$; atom $T(B)$ does not cover $A$, but it does so together with $S(A,B)$.

The following are the smallest hierarchical queries that are not in $\text{CQAP}_0$ but in $\text{CQAP}_1$: $Q(A \mid\cdot) = R(A,B), S(B)$; $Q(B \mid A) = R(A,B), S(B)$; and $Q(\cdot \mid A) = R(A,B), S(B)$.
\qed
\end{exa}

\paragraph{Query Semantics.} 
We give the semantics of CQAPs using the function $\llbracket\cdot\rrbracket$ defined on the structure of CQAPs:
\begin{align}
	\llbracket (Q (\calO|\calI) = \text{body} \rrbracket &= \{ (t\mapsto m) \ |\  J = \llbracket \text{body}\rrbracket, (t_2\mapsto m_2)\in J, t = t_2[\calO], \textsf{in} = t_2[\calI] \nonumber\\
	&\hspace*{6.25em} m = \sum_{(t_1\mapsto m_1) \in J, t = t_1[\calO], \text{in} = t_1[\calI]} m_1
	\} \label{eq:semantics:head}\\
	\llbracket Q_1(\calX_1), Q_2(\calX_2) \rrbracket &= \{ (t\mapsto m) \ |\  (t_1\mapsto m_1)  \in \llbracket Q_1(\calX_1) \rrbracket, (t_2\mapsto m_2) \in \llbracket Q_2(\calX_2) \rrbracket, \nonumber\\
	&\hspace*{4.25em} t\in \Dom(\calX_1 \cup \calX_2), t_1 = t[\calX_1], t_2 = t[\calX_2], m = m_1 \cdot m_2 \}\label{eq:semantics:join}\\
	\llbracket R(\calX) \rrbracket &= \{ (t\mapsto m) \ | \  (t\mapsto m)\in R\} \label{eq:semantics:atom}
\end{align}
Eq.~\eqref{eq:semantics:head} computes the set of mappings of the tuples over the output variables $\calO$ to their multiplicities under a specific tuple $\textsf{in}$ of constants assigned to the input variables $\calI$. It recursively invokes the semantics function applied to the body of the query. Eq.~\eqref{eq:semantics:join} computes the set of mappings $(t\mapsto m)$ defining the join of two subqueries $Q_1$ and $Q_2$. The tuple $t$ is the result of joining the tuple $t_1$ in the output of $Q_1$ and the tuple $t_2$ in the output of $Q_2$, while its multiplicity $m$ is the product of the multiplicities of $t_1$ and $t_2$. Eq.~\eqref{eq:semantics:atom} is the base case of one relation atom. Its semantics is the set of key-value mappings represented by the corresponding relation.

%%%%%%%%%%%%%%%%%%%%%%%%%%%%
\paragraph{Delta Queries.}
Updates to input relations can change the query output. 
A delta query captures this change for updates to one input relation. The derivation of delta queries follows the standard delta rules~\cite{Chirkova:Views:2012:FTD}.
Consider a CQAP as in Eq.~\eqref{eq:CQAP} and an update $\delta R_i$ to a relation $R_i$. If there is a single atom using the relation symbol $R_i$ in $Q$, then the delta query expressing the change in the query output is:
\begin{align*}
\delta Q(\calO|\calI) = R_1(\calX_1), \ldots, R_{i-1}(\calX_{i-1}), \delta R_i(\calX_i), R_{i+1}(\calX_{i+1}), \ldots, R_n(\calX_n)
\end{align*}
If there are several atoms using $R_i$ in $Q$, then we issue one delta query for each such atom.

\nop{
\begin{align*}
\delta({R(\calX)}) &= \delta{R(\calX)}	\\
\delta({S(\calX)}) &= \emptyset \\
\delta{(Q_1(\calX_1), Q_2(\calX_2))} &= 
	(\delta{(Q_1(\calX_1))}, Q_2(\calX_2)) \uplus
	((Q_1(\calX_1)), \delta{Q_2(\calX_2)}) \uplus
	(\delta{(Q_1(\calX_1))}, \delta{Q_2(\calX_2)})  
\end{align*}

The last rule refers to the case of a join of two queries.
}

%%%%%%%%%%%%%%%%%%%%%%%%%%%%
%%%%%%%%%%%%%%%%%%%%%%%%%%%%
\section{Variable Orders and Width Measures}
\label{sec:var_orders_width_measures}
In this section, we introduce the notions of variable orders and width measures for CQAPs.

\subsection{Variable Orders.}
\label{sec:vo}
Variable orders are used as logical plans for the evaluation of conjunctive queries~\cite{OlteanuZ15}. We next adapt them to CQAPs.
Given a query, two variables  {\em depend} on each other if they occur in the same query atom.
A {\em variable order}, or VO for short, $\omega$ for a CQAP $Q$ is a pair
$(T_\omega, \dep_\omega)$, where:

\begin{itemize}
	\item $T_\omega$ is a forest of (rooted) trees with one node per variable.
	For each atom $R(\calX)$ in $Q$, $\calX$ is a subset of the set of variables on a root-to-leaf path in $T_{\omega}$.
	\item The function $\dep_\omega$ maps each variable
			$X$ to the subset of its ancestor variables in $T_{\omega}$
			on which the variables in the subtree rooted at $X$
			depend.
\end{itemize}

For convenience, we sometimes omit the index $\omega$ in $(T_{\omega},\dep_{\omega})$
 when $\omega$ is clear from the context.
A VO always exists for a query, e.g., by having all variables on a single path.
In the remainder of this paper, we consider VOs in which atoms corresponding to relations and their indicator projections
are added as new leaves. 
Each atom in the query is added as a child of its variable placed lowest in the VO. We explain next how the indicator projections are added to a VO $\omega$. 
Indicator projections can reduce the asymptotic complexity of cyclic queries~\cite{FAQ:PODS:2016}.

\begin{figure}[t]
	\centering
	\setlength{\tabcolsep}{3pt}
	\renewcommand{\linenumber}{\makebox[2ex][r]{\rownumber\TAB}}
	\setcounter{magicrownumbers}{0}
	\begin{tabular}[t]{@{}c@{}c@{}l@{}}
		\toprule
		\multicolumn{3}{l}{\textsf{indicators}(CQAP $Q$, \text{VO} $\nu$) : VO}   \\
		\midrule
		\multicolumn{3}{l}{\MATCH $\nu$:}                       \\
		\midrule
		\phantom{ab} & $R(\calY)$ \hspace*{2.5em} & 

		 \linenumber \RETURN $R(\calY)$ \\
		\cmidrule{2-3} \\[-6pt]
		             &
		\begin{minipage}[t]{1.5cm}
			\vspace{-1.2em}
			\hspace*{-0.75cm}
			\begin{tikzpicture}[xscale=0.5, yscale=1]
				\node at (0,-2)  (n4) {$X$};
				\node at (-1.2,-3)  (n1) {$\nu_1$} edge[-] (n4);
				\node at (0,-3)  (n2) {$\ldots$};
				\node at (1.2,-3)  (n3) {$\nu_k$} edge[-] (n4);
			\end{tikzpicture}
		\end{minipage}
		             &
		\begin{minipage}[t]{11.5cm}
			\vspace{-0.4cm}
			\linenumber \LET $\hat{\nu}_i = \textsf{indicators}(Q, \nu_i)$ \ $\forall i\in[k]$ \\[0.5ex]
		         \linenumber \LET $\calS = \{X\} \cup \dep_{\omega}(X)$; \LET $\calR $ be the set of atoms in $\nu$ \\[0.5ex]
		          \linenumber \LET $\calI = \{\, I_{\calY\cap \calS}R(\calY\cap \calS) \mid R(\calY) \in (\atoms(Q)\setminus \calR) \land (\calY\cap\calS) \neq \emptyset \,\}$ \\[0.5ex]
			\linenumber \LET $\{I_1, {.}{.}{.}, I_\ell\} = \gyo^*(\calI, \calR)$  \\[0.5ex] 
		 \linenumber \RETURN $
				\left\{
				\begin{array}{@{~~}c@{~~}}
					\tikz {
						\node at (3.6,-1)  (n4) {$X$};
					        \node at (2.2,-1.75)  (n1) {$\hat{\nu}_1$} edge[-] (n4);
						\node at (2.65,-1.75)  (n2) {$\ldots$};
						\node at (3.2,-1.75)  (n3) {$\hat{\nu}_k$} edge[-] (n4);
						\node at (4.0,-1.75)  (n3) {$I_1$} edge[-] (n4);
						\node at (4.4,-1.75)  (n2) {$\ldots$};
						\node at (4.9,-1.75)  (n3) {$I_\ell$} edge[-] (n4);
					}
				\end{array}  \right.$
		\end{minipage}                                              \\[2.75ex]
		\bottomrule
	\end{tabular}
	\caption{Extending a VO $\omega$ of a CQAP $Q$ with indicator projections by calling $\textsf{indicators}(Q,\nu=\omega)$. The function \textsf{indicators} is defined using pattern matching (left column under $\MATCH$) on the structure of $\omega$, which can be a leaf (relation atom) or an inner node (query variable).
	Each variable $X$ in $\omega$ gets as new children the indicator projections of relations that do not occur in the subtree rooted at $X$ but form a cyclic query with those that occur. 
	$\gyo^*$ (defined in Section~\ref{sec:vo}) is based on the GYO reduction~\cite{BeeriFMY83}.}
	\label{fig:extended_variable_order}
\end{figure}

Given a CQAP $Q$ and a VO $\omega$, where the atoms of $Q$ have been already added, the function $\textsf{indicators}$ in Figure~\ref{fig:extended_variable_order} extends $\omega$ with indicator projections. 
It processes $\omega$ recursively in a bottom-up manner (Lines 1-2).
At each variable $X$ in $\omega$,
we compute the set $\calI$ of indicator projections (Line 4). 
Such indicator projections $I_{\calY\cap\calS}R$ are for relations $R$ whose atoms $R(\calY)$ are not included in the subtree rooted at $X$ but have schema $\calY$ that shares a non-empty set of variables with $\calS = \{X\} \cup \dep_{\omega}(X)$.
We choose from this set those indicators that 
together with the atoms in the subtree rooted at $X$ form a cyclic query  (Line 5). 
We achieve this using a variant of the GYO reduction~\cite{BeeriFMY83}.
Given the hypergraph formed by the hyperedges representing these indicators $\calI$ and the atoms  $\calR$, GYO repeatedly applies two rules until it reaches a fixpoint: (1) Remove a node that only appears in one hyperedge; (2) Remove a hyperedge that is included in another hyperedge. If the result of GYO is a hypergraph with no nodes and one empty hyperedge, then the input hypergraph is ($\alpha$-)acyclic. Otherwise, the input hypergraph is cyclic and the GYO's 
output is a cyclic hypergraph. Our GYO variant, dubbed $\gyo^*$ in Figure~\ref{fig:extended_variable_order}, returns the hyperedges that originated from the indicator projections in $\calI$ and contribute to this non-empty output hypergraph.
This set of chosen indicator projections, 
which is empty if the input hypergraph is ($\alpha$-)acyclic, are added as children of $X$ (Line 6).

The next proposition states that joining a query with the indicator projections constructed by the function 
\textsf{indicators} in Figure~\ref{fig:extended_variable_order} does not change the result of the query.   
The proof of the proposition is in Appendix~B.1 of the technical report~\cite{access_pattern_arxiv}.
\begin{prop}
\label{prop:vo-equiv-ext-vo}
For any CQAP $Q(\calO | \calI)$ and VO $\omega$ for $Q$, $Q(\calO | \calI)$ is equivalent to a CQAP $Q'(\calO | \calI)$, whose body is the conjunction of the atoms of $Q$ and the indicator projections at the leaves of the VO returned by \textsf{indicators}(Q, $\omega$).
\end{prop}

 The following example illustrates the construction of indicator projections  
as described by the function \textsf{indicators} in Figure~\ref{fig:extended_variable_order}.
Examples~\ref{exa:effect_of_indicators} and \ref{ex:triangle_three_tails} show that indicator projections can reduce the preprocessing and update time for CQAPs.

	\begin{exa}\label{ex:CQAP-triangle}
	Consider the triangle CQAP
	$$Q(B,C|A) = R(A,B), S(B,C), T(C,A).$$
	The fracture $Q_\dagger$ of $Q$ is the query itself. Figure~\ref{fig:general_triangle-prelim} depicts a VO $\omega$ for $Q$. 
	The input variable $A$ is on top of the output variables $B$ and $C$.  
	The atoms $S(B,C)$ and $T(C,A)$ are included 
	in the subtree of $\omega$ rooted at $C$ but  the atom $R(A,B)$ is not.
	We apply $\gyo^*$ to the atoms $S(B,C)$ and $T(C,A)$ and the indicator projection $I_{A,B}R(A,B)$ 
	and obtain $\gyo^*( \{I_{A,B}R(A,B)\}, \{S(B,C), T(C,A)\}) = \{I_{A,B}R(A,B)\}$, which means that
	the indicator projection $I_{A,B}R(A,B)$ and the atoms $S(B,C)$ and $T(C,A)$ form a cyclic query. For this reason,
	$I_{A,B}R(A,B)$ is added as a new child of $C$ in $\omega$.	  
\qed	
	\end{exa}

\begin{figure}[t]
	\centering
  \begin{minipage}[b]{0.54\linewidth}
	\centering
	  \begin{tikzpicture}[xscale=0.7, yscale=0.7]
		  \node[anchor=west] at (-6.5, 0.0) (A1) {\small  $\dep(A)=\emptyset$};	
		  \node[anchor=west] at (-6.5, -0.8) (B1) {\small  $\dep(B)=\{A\}$};	
		  \node[anchor=west] at (-6.5, -1.6) (C1) {\small  $\dep(C)=\{A,B\}$};	
				
		  \node at (-1, -0.0) (A) {\small  $A$};
		  \node at (-1.0, -1.0) (B) {\small $B$} edge[-] (A);
		  \node at (-1.0, -3.0) (C) {\small $C$} edge[-] (B);
		  \node at (-3.5, -4.0) (S) {\small $S(B,C)$} edge[-] (C);
		  \node at (-1., -4.0) (T) {\small $T(C,A)$} edge[-] (C);
		  \node at (1.5, -2.0) (R) {\small  $R(A,B)$} edge[-] (B);
		  \node at (1.75, -4.0) (IR) {\small $I_{A,B}R(A,B)$} edge[-] (C);
		%   \node at (0,  -4) (invisible) {};
	  \end{tikzpicture}
  \end{minipage}
  \begin{minipage}[b]{0.44\linewidth}
	\centering
	  \begin{tikzpicture}[xscale=0.7, yscale=0.7]
		  \node at (-1, -0.0) (A) {\small  $V_A(A)$};
		  \node at (-1.0, -1.0) (B) {\small $V_B(A,B)$} edge[-] (A);
		  \node at (-1.0, -2.0) (C') {\small $V'_C(A,B)$} edge[-] (B);
		  \node at (-1.0, -3.0) (C) {\small $V_C(A,B,C)$} edge[-] (C');
		  \node at (-3.5, -4.0) (S) {\small $S(B,C)$} edge[-] (C);
		  \node at (-1., -4.0) (T) {\small $T(C,A)$} edge[-] (C);
		  \node at (1.5, -2.0) (R) {\small  $R(A,B)$} edge[-] (B);
		  \node at (1.75, -4.0) (IR) {\small\textbf{$I_{A,B}R(A,B)$}} edge[-] (C);
		%   \node at (0,  -4) (invisible) {};
	  \end{tikzpicture}
  \end{minipage}
  \caption{Left: (Access-top) VO for the query $Q(B,C|A)= R(A,B), S(B,C), T(C,A)$. Right: The view tree constructed from this VO. Note the indicator $I_{A,B}R(A,B)$ added below the variable $C$ (left) and below the view $V_C$ (right). }
  \label{fig:general_triangle-prelim}
  \end{figure}

For the following development, we need additional notation.
Given a VO $\omega$, its subtree rooted at $X$ is denoted by $\omega_X$.
The sets $\vars(\omega)$ and $\anc_\omega(X)$ consist of
all variables of $\omega$ and respectively
the variables on the path from $X$ to the root excluding $X$.
 We denote by $\atoms(\omega)$ all  atoms and indicators 
 at the leaves of $\omega$ and by $Q_X$ the query that is the join of all atoms $\atoms(\omega_X)$ and where all variables are free. 

We next introduce classes of VOs for CQAPs. 
A VO $\omega$ is  {\em canonical} if the variables of the leaf atom of each root-to-leaf path are 
{\em exactly} the inner nodes of the path. Hierarchical queries are precisely those conjunctive queries that admit canonical variable orders.
A VO $\omega$ is {\em free-top} if no bound variable is an ancestor of a free variable.
It is  {\em input-top} if no output variable is an ancestor of an input variable.
The sets of free-top and input-top VOs for $Q$ are denoted as 
$\freetop(Q)$ and $\inputtop(Q)$, respectively.
A VO is called {\em access-top} if it is free-top and 
input-top\footnote{Although our approach in this work uses variable orders, it could also be phrased in terms of hypertree decompositions~\cite{Gottlob99}, while preserving the same complexities. Every variable order $\omega$ can be translated into a hypertree decomposition by replacing each node $X$ by a bag consisting of the variables $\{X\} \cup \dep_\omega(X)$; conversely, every hypertree decomposition can be transformed into a variable order by replacing each bag $B$ by a path 
that consists of all variables in $B$ that do not appear in bags that are ancestor of $B$~\cite{OlteanuZ15}. 
\nop{A  free-top (input-top) variable order corresponds to a hypertree decomposition that contains a connected subtree consisting of the free (input) variables.}
An access-top variable order corresponds to a hypertree decomposition that contains a connected subtree consisting of all free variables and also a connected subtree consisting of the input variables. Since all input variables are free, the subtree consisting of the input variables must be subsumed by the subtree consisting 
of the free variables.  For a discussion on the usefulness of variable orders for our approach, see Section~\ref{sec:discussion}.
}:
$\acceff(Q) = \freetop(Q) \cap \inputtop(Q).$

\nop{The sets $\freeTopVO(Q)$, $\canonicalVO(Q)$, and $\VO(Q)$ consist of free-top, canonical, and all variable orders of $Q$.}

\begin{exa}\label{ex:queries-var-orders}
The query $Q(B | A) = R(A,B), S(B)$ admits the VO $B-\{A-R(A,B),S(B)\}$ 
(notation-wise, ``-"
represents the parent-child relationship), where the variable $B$ has two children: the variable $A$ and the atom $S(B)$; and the variable $A$ has one child: the atom $R(A,B)$. 
The dependency sets are $\dep(B)=\emptyset$ and $\dep(A)=\{B\}$. This VO is free-top, since both variables are free; it is  not input-top, since the output variable $B$ is on top of the input variable $A$. 
By swapping $A$ and $B$, the VO becomes $A-B-\{R(A,B),S(B)\}$ with the dependency sets  $\dep(A)=\emptyset$ and $\dep(B)=\{A\}$.

The triangle query $Q(A,B|\cdot) = R(A,B), S(B,C),T(A,C)$ admits the VO 
$C-A-\{T(A,C),$ $B-\{R(A,B), S(B,C),$ $I_{AC}T(A,C)\}\}$, where one child of $B$ is the indicator projection $I_{AC}T$ of $T$ on $\{A,C\}$.
The dependency sets are $\dep(C) = \emptyset$, $\dep(A) = \{C\}$, and $\dep(B) = \{A,C\}$.
The VO is trivially input-top, since the query has no input variables; it is not free-top, since the bound variable $C$ is on top of the free variables $A$ and $B$.

The fracture of the 4-cycle query $Q_1$ in Example~\ref{ex:illustrate_cqap_0} admits the access-top VO consisting of the following two disconnected paths: $B_1-D_2-A-\{R(A,B_1),U(A,D_2)\}$ and $B_2-D_1-C-\{S(B_2,C),T(C,D_1)\}$, where the dependency sets are: $\dep(A) = \{B_1,D_2\}$, $\dep(D_2)=\{B_1\}$, $\dep(B_1)=\dep(B_2)=\emptyset$, $\dep(C) = \{B_2,D_1\}$, and $\dep(D_1)=\{B_2\}$.
\qed
\end{exa}

\subsection{Width Measures}
\label{sec:widths}
Given a conjunctive query $Q$ and $\calF \subseteq \vars(Q)$,    
a {\em fractional edge cover}
of $\calF$ is a solution 
$\boldsymbol{\lambda} = (\lambda_{R(\calX)})_{R(\calX) \in \atoms(Q)}$ to the following 
linear program \cite{AtseriasGM13}: 
\begin{align*}
\text{minimize} & \TAB\sum_{R(\calX) \in\, \atoms(Q)} \lambda_{R(\calX)} && \\[3pt]
\text{subject to} &\TAB \sum_{R(\calX): \, X \in \calX} \lambda_{R(\calX)} \geq 1 && \text{ for all } X \in \calF \text{ and } \\[3pt]
& \TAB\lambda_{R(\calX)} \in [0,1] && \text{ for all } R(\calX) \in \atoms(Q)
\end{align*}
The optimal objective value of the above program 
is called the {\em fractional edge cover number} of $\calF$ 
in $Q$ and is denoted as $\rho_{Q}^{\ast}(\calF)$.  
An {\it integral edge cover} of $\calF$ is a feasible solution 
to the variant of the above program with 
$\lambda_{R(\calX)}\in\{0,1\}$ for each $R(\calX) \in \atoms(Q)$.
The optimal objective value of this program 
is called the {\em integral edge cover number} of $\calF$, denoted as $\rho_{Q}(\calF)$.
If $Q$ is clear from the context, we omit 
the subscript $Q$ in $\rho_{Q}^{\ast}(\calF)$
and $\rho_{Q}(\calF)$.

For hierarchical queries, the integral and fractional edge cover numbers are the same. 

\begin{lem}[Lemma D.1 in~\cite{KaraNOZ2020}]
\label{lem:rho_rhostar}
For any hierarchical query $Q$ and $\calF \subseteq \vars(Q)$, it holds
$\rho^{\ast}(\calF) = \rho(\calF)$.
\end{lem}

We next introduce two width measures for a VO $\omega$ and CQAP $Q$. They capture the complexity of computing and maintaining the output of $Q$.

\begin{defi}
	\label{def:fac_width_VO}
	The static width $\fw(\omega)$ and dynamic width $\dfw(\omega)$  of a VO $\omega$ are:
	\begin{align*} 	
		\fw(\omega) & = \max_{X \in \vars(\omega)} \rho^*_{Q_X}(\{X\} \cup \dep_{\omega}(X))\\
		\dfw(\omega) & = \max_{X \in \vars(\omega)} \
			\max_{R(\calY) \in \atoms(\omega_X)}
			\rho^*_{Q_X}((\{X\} \cup \dep_{\omega}(X)) \setminus \calY)
	\end{align*}
\end{defi}

$Q_X$ is the join of all atoms under $X$ in the VO $\omega$.
For a query $Q_X$, the set of variables $\calX=\{X\} \cup \dep_\omega(X)$,
and a database of size $N$, $N^{e}$  is an upper bound on the worst-case output size of the query $Q_X(\calX)$,
where $e = \rho_{Q_X}^*(\calX)$ is  the fractional edge cover number of $\calX$.
%Using a worst-case optimal join algorithm~\cite{LFTJ:ICDT:2014,Ngo:JACM:18}, we can compute the output of 
%$Q_X$ in time $\bigO{N^{e}}$.
%
The static width $\fw$ of a VO $\omega$ is defined by the maximum over the fractional edge cover numbers of the queries $Q_X$ for the variables $X$ in $\omega$.
 The dynamic width is defined similarly, with one simplification: We consider every case of a relation (or indicator projection) $R$ being replaced by a single-tuple update, so its variables $\calY$ are all set to constants and can be ignored in the computation of the fractional edge cover number.

We consider the standard lexicographic ordering $\leq$ on pairs of dynamic and static widths: $(\delta_1,\fw_1) \leq (\delta_2,\fw_2)$ if $\delta_1 < \delta_2$ or  $\delta_1=\delta_2$ and $\fw_1 \leq \fw_2$. Given a set $\calS$ of VOs, we define $\min_{\omega\in \calS} (\delta(\omega),\fw(\omega)) = (\dfw,\fw)$ such that  $\forall \omega \in \calS: (\dfw,\fw) \leq (\dfw(\omega),\fw(\omega))$.

\begin{defi}
	\label{def:fac_width_CQAP}
	The dynamic width $\dfw(Q)$ and static width $\fw(Q)$ of a CQAP $Q$ are:  
		$$(\dfw(Q), \fw(Q))       = \min_{\omega\in\acceff(Q_{\dagger})} (\dfw(\omega), \fw(\omega))$$ 
\end{defi}

Since we are interested in dynamic evaluation, Definition~\ref{def:fac_width_CQAP} first minimises for the dynamic width and then for the static width. 
To determine the dynamic and the static width of a CQAP $Q$,  we first search for the VOs of the fracture $Q_{\dagger}$ with minimal dynamic width and choose among them one with the smallest static width. 
%The extended technical report~\cite{DBLP:journals/corr/abs-2206-09032} further expands on the width measures with examples and properties.
\nop{
In particular, it shows:
\begin{proposition}
\label{prop:CQAP_0_dynamic_width_01}
A query is $\text{CQAP}_0$ if and only if it has dynamic width $0$.
A query is $\text{CQAP}_1$ if and only if it has dynamic width $1$.
\end{proposition}
%a query is in $CQAP_i$ if and only if its dynamic width is $i$ (for $i\in\{0,1\}$).
}
%The following example shows that the widths of a query depend on the access pattern of the query. 

\begin{exa}
\label{ex:widths-four-cycle}
We show how to compute the widths for the variable order of the fractured 4-cycle query in Example~\ref{ex:queries-var-orders}: For the bag at variable $A$, we have $\rho^*(\{A\}\cup\dep(A)) = \rho^*(\{A,D_2,B_1\})=2$, which is the largest fractional edge cover number for any variable in the variable order. Further access-top variable orders are possible by swapping $B_1$ with $D_2$ and $B_2$ with $D_1$, yielding the same overall cost. The static width of the fractured 4-cycle query is thus 2. To compute the dynamic width of the same variable order, we consider for each atom, the fractional edge cover number of each bag without the variables in this atom. For the bag $\{A\}\cup\dep(A)=\{A,D_2,B_1\}$, we get $\rho^*(\{A,D_2,B_1\}\setminus\{A,B_1\})=1$ for the atom $R(A,B_1)$ and $\rho^*(\{A,D_2,B_1\}\setminus\{A,D_2\})=1$ for the atom $U(A,D_2)$. Overall, the dynamic width of this variable order is 1.
\qed
\end{exa}

\begin{exa}\label{ex:intro-access-patterns}
Consider the query $Q(\calO\mid\calI) = R(A,B,C), S(A,B,D), T(A,E).$ The static width $\fw$ and the dynamic width $\delta$ of $Q$ vary depending on the access pattern:
\begin{itemize}
	\item $\fw=1$ and $\dfw=0$ for $Q(C,D,E \mid A,B)$, $Q(A,B,C,D,E |\cdot)$, $Q(\cdot | A,B,C,D,E)$ and $Q(B,C,D,E | A)$;
	\item $\fw=1$ and $\dfw=1$ for $Q(A,C,D,E \mid B)$;
	\item $\fw=2$ and $\dfw=1$ for $Q(A,C,D \mid B,E)$;
	\item $\fw=2$ and $\dfw=2$ for $Q(A,E \mid B,C,D)$;
	\item $\fw=3$ and $\dfw=2$ for $Q(A,B \mid C,D,E)$.
\end{itemize} 
\end{exa}

  The next example illustrates that the indicator projections constructed by the function 
  \textsf{indicators} in Figure~\ref{fig:extended_variable_order} can lower the static width of the VO
  of a query. Lower static width implies lower preprocessing time as stated in Theorem~\ref{thm:general}.  
%  does not change the result of the query.  
 % the effect of indicator projections obtained using the procedure 
%$\gyo^*$ introduced in Section~\ref{sec:vo} on the static width of VOs. 

 \begin{exa}
 \label{exa:effect_of_indicators} 
Recall the triangle CQAP
	$Q(B,C|A) = R(A,B), S(B,C), T(C,A)$ from Example~\ref{ex:CQAP-triangle} and its access-top VO $\omega$ in Figure~\ref{fig:general_triangle-prelim}.
	The indicator projection $I_{A,B}R(A,B)$ is below $C$ in $\omega$, since the output of $\gyo^*(\{I_{A,B}R(A,B)\},$ $\{ R(A,B), T(C,A)\})$ is $\{I_{A,B}R(A,B)\}$.

	Assume first that $I_{A,B}R(A,B)$ is not included in $\omega$. In this case, the query $Q_C$ is defined as the join of $S(B,C)$ and $T(C,A)$, which means
	$\rho_{Q_C}^*(\{C\} \cup dep(C)) = \rho_{Q_C}^*(\{A,B,C\}) = 2$.
	
	Assume now that  $I_{A,B}R(A,B)$ is included in $\omega$. In this case, 
	the query $Q_C$ is defined as the join of $S(B,C)$, $T(C,A)$, and $I_{A,B}R(A,B)$, which means 
	$\rho_{Q_C}^*(\{C\} \cup dep(C)) = \rho_{Q_C}^*(\{A,B,C\}) = \frac{3}{2}$.
	Hence, the fractional edge cover number reduces from $2$ to $\frac{3}{2}$. 
\nop{	By adding the indicator $I_{A,B}R$ below $C$, the fractional edge cover number $\rho^*(\{C\} \cup dep(C)) = \rho^*(\{A,B,C\})$ of the query $Q_C(A,B,C) = S(B,C), T(C,A), I_{A,B}R(A,B)$ reduces from $2$ to $\frac{3}{2}$.} 
	This fractional edge cover number dominates the static width of $\omega$, so the static width of $\omega$ is $\frac{3}{2}$.

	The dynamic width of $\omega$ (including $I_{A,B}R$) is dominated by the fractional edge cover number $\rho^*_{Q_C}(\{C\} \cup dep(C)) - \calS) = \rho^*_{Q_C}(\{A,B,C\} - \calS)$, where $\calS$ is the schema of $S$, $T$, or $I_{A,B}R$. In each of these three cases, $\{A,B,C\} - \calS$ consists of a single variable. Hence, the fractional edge cover number is 1 and, therefore, the dynamic width of $\omega$ is 1.
\qed		
	\end{exa}

\begin{figure}[t]
\small
\centering

\begin{minipage}[b]{0.22\linewidth}
	\begin{tikzpicture}[xscale=0.96, yscale=0.8]
		\node at (-0.25, 0.0) (A) {\small  $A$};
		\node at (-1.0, -1.0) (B) {\small $B$};
		\node at (0.5, -1.0) (C) {\small $C$};
		\node at (-0.25, 1.0) (D) {\small $D$};
		\node at (-0.25, 2.0) (E) {\small $E$};
		\node at (-1., -2.0) (F) {\small $F$};
		\node at (-1., -3.0) (G) {\small $G$};
		\node at (.5, -2.0) (H) {\small $H$};
		\node at (.5, -3.0) (J) {\small $J$};
		
		\node at (-1.1, -0.25) (R1) {\small $R_1$};
		\node at (0.65, -0.25) (R2) {\small $R_2$};
		\node at (-.25, -1.5) (R3) {\small $R_3$};
		\node at (0.25, 0.5) (R4) {\small $R_4$};
		\node at (0.25, 1.5) (R5) {\small $R_5$};
		\node at (-1.5, -1.5) (R6) {\small $R_6$};
		\node at (-1.5, -2.5) (R7) {\small $R_7$};
		\node at (1., -1.5) (R8) {\small $R_8$};
		\node at (1., -2.5) (R9) {\small $R_9$};

		\node at (0,  -3) (invisible) {};

		\begin{pgfonlayer}{background}
			\draw[opacity=.3,fill opacity=.3,line cap=round, line join=round, line width=14pt,color=teal] (-0.25,0.0) -- (-1,-1); 
			\draw[opacity=.3,fill opacity=.3,line cap=round, line join=round, line width=14pt,color=teal] (0.5,-1.0) -- (-1,-1); 
			\draw[opacity=.3,fill opacity=.3,line cap=round, line join=round, line width=14pt,color=teal] (-0.25,0.0) -- (0.5, -1.0); 

			\draw[opacity=.5,fill opacity=.5,line cap=round, line join=round, line width=13pt,color=orange] (-0.25,0.0) -- (-0.25, 1.0);
			\draw[opacity=.5,fill opacity=.5,line cap=round, line join=round, line width=15pt,color=orange] (-0.25,1.0) -- (-0.25, 2.0);

			\draw[opacity=.5,fill opacity=.5,line cap=round, line join=round, line width=13pt,color=orange] (-1.0,-1.0) -- (-1.0, -2.0);
			\draw[opacity=.5,fill opacity=.5,line cap=round, line join=round, line width=15pt,color=orange] (-1.0, -3.0) -- (-1.0, -2.0);

			\draw[opacity=.5,fill opacity=.5,line cap=round, line join=round, line width=13pt,color=orange] (0.5,-1.0) -- (0.5, -2.0);
			\draw[opacity=.5,fill opacity=.5,line cap=round, line join=round, line width=15pt,color=orange] (0.5, -3.0) -- (0.5, -2.0);

		\end{pgfonlayer}
	\end{tikzpicture}
\end{minipage}
\begin{minipage}[b]{0.18\linewidth}
	\scalebox{0.9}{
	\centering
	\begin{tikzpicture}[xscale=0.925, yscale=0.7]
\begin{scope}[yshift=2.4cm, xshift=0.2cm]
\node at (0.2, -1.2) (A) {\small  $\dep(A)=\emptyset$};	
\node at (0.2, -1.8) (B) {\small  $\dep(B)=\{A\}$};	
\node at (0.2, -2.4) (C) {\small  $\dep(C)=\{A,B\}$};	
\node at (0.2, -3.0) (D) {\small  $\dep(D)=\{A\}$};
\node at (0.2, -3.6) (E) {\small  $\dep(E)=\{D\}$};
\node at (0.2, -4.2) (F) {\small  $\dep(F)=\{B\}$};
\node at (0.2, -4.8) (G) {\small  $\dep(G)=\{F\}$};
\node at (0.2, -5.4) (H) {\small  $\dep(H)=\{C\}$};
\node at (0.2, -6.0) (H) {\small  $\dep(J)=\{H\}$};
\end{scope}
\node at (0,  -4) (invisible) {};
	\end{tikzpicture}
	}
\end{minipage}
\begin{minipage}[b]{0.58\linewidth}
	\scalebox{0.85}{
	\centering
	\begin{tikzpicture}[xscale=1.3, yscale=0.9]
		\node at (0, 0.0) (A) {\small  $A$};
		\node at (0.0, -1.0) (B) {\small $B$} edge[-] (A);
		\node at (-0.0, -2.0) (C) {\small $C$} edge[-] (B);
		\node at (1.0, -3.0) (D) {\small $D$} edge[-] (C);
		\node at (1.0, -4.0) (E) {\small $E$} edge[-] (D);
		\node at (2.0, -2.0) (F) {\small $F$} edge[-] (B);
		\node at (2.0, -3.0) (G) {\small $G$} edge[-] (F);
		\node at (0.0, -3.0) (H) {\small $H$} edge[-] (C);
		\node at (0.0, -4.0) (J) {\small $J$} edge[-] (H);
		
		\node at (-1.5, -2.0) (R1) {\small $R_1(A,B)$} edge[-] (B);
		\node at (-1.25, -3.6) (R2) {\small $R_2(B,C)$} edge[-] (C);
		\node at (-2.5, -3.6) (R3) {\small $R_3(C,A)$} edge[-] (C);
		\node at (-3.0, -3.0) (IR1) {\small $I_{A,B}{R_1}(A,B)$} edge[-] (C);
		\node at (2.0, -4.0) (R4) {\small $R_4(A,D)$} edge[-] (D);
		\node at (1.3, -5.0) (R5) {\small $R_5(D,E)$} edge[-] (E);
		\node at (3.2, -3.0) (R6) {\small $R_6(B,F)$} edge[-] (F);
		\node at (3.2, -4.0) (R7) {\small $R_7(F,G)$} edge[-] (G);
		\node at (-1.0, -4.5) (R8) {\small $R_8(C,H)$} edge[-] (H);
		\node at (0.0, -5.0) (R9) {\small $R_9(H,J)$} edge[-] (J);
		\node at (0,  -4) (invisible) {};
	\end{tikzpicture}
	}
\end{minipage}

\begin{minipage}[b]{0.495\linewidth}
	\scalebox{0.845}{
	\centering
	\begin{tikzpicture}[xscale=1.3, yscale=0.9]

		\node at (0, 1.0) (D) {\small  $D$};
		\node at (0, 0.0) (A) {\small  $A$} edge[-] (D);
		\node at (0.0, -1.0) (B) {\small $B$} edge[-] (A);
		\node at (-0.0, -2.0) (C) {\small $C$} edge[-] (B);
		% \node at (-0.0, -3.0) (D) {\small $D$} edge[-] (C);
		\node at (1.0, -.0) (E) {\small $E$} edge[-] (D);
		\node at (1.0, -2.0) (F) {\small $F$} edge[-] (B);
		\node at (1.0, -3.0) (G) {\small $G$} edge[-] (F);
		\node at (0.0, -3.0) (H) {\small $H$} edge[-] (C);
		\node at (0.0, -4.0) (J) {\small $J$} edge[-] (H);
		
		\node at (-1.5, -2.0) (R1) {\small $R_1(A,B)$} edge[-] (B);
		\node at (-0.75, -3.5) (R2) {\small $R_2(B,C)$} edge[-] (C);
		\node at (-2, -3.5) (R3) {\small $R_3(C,A)$} edge[-] (C);
		\node at (-3.0, -3.0) (IR1) {\small $I_{A,B}{R_1}(A,B)$} edge[-] (C);
		\node at (-1.5, -1.0) (R4) {\small $R_4(A,D)$} edge[-] (A);
		\node at (1.0, -1.0) (R5) {\small $R_5(D,E)$} edge[-] (E);
		\node at (2.15, -3.0) (R6) {\small $R_6(B,F)$} edge[-] (F);
		\node at (2.15, -4.0) (R7) {\small $R_7(F,G)$} edge[-] (G);
		\node at (1.0, -4.0) (R8) {\small $R_8(C,H)$} edge[-] (H);
		\node at (0.0, -5.0) (R9) {\small $R_9(H,J)$} edge[-] (J);
		\node at (0,  -4) (invisible) {};
	\end{tikzpicture}
	}
\end{minipage}
\begin{minipage}[b]{0.495\linewidth}
	\scalebox{0.845}{
	\centering
	\begin{tikzpicture}[xscale=1.3, yscale=0.9]

		\node at (0.0, 2.0) (E) {\small $E$};
		\node at (0, 1.0) (D) {\small  $D$} edge[-] (E);
		\node at (0, 0.0) (A) {\small  $A$} edge[-] (D);
		\node at (0.0, -1.0) (B) {\small $B$} edge[-] (A);
		\node at (-0.0, -2.0) (C) {\small $C$} edge[-] (B);
		% \node at (-0.0, -3.0) (D) {\small $D$} edge[-] (C);
		\node at (1.0, -2.0) (F) {\small $F$} edge[-] (B);
		\node at (1.0, -3.0) (G) {\small $G$} edge[-] (F);
		\node at (0.0, -3.0) (H) {\small $H$} edge[-] (C);
		\node at (0.0, -4.0) (J) {\small $J$} edge[-] (H);
		
		\node at (-1.5, -2.0) (R1) {\small $R_1(A,B)$} edge[-] (B);
		\node at (-0.75, -3.5) (R2) {\small $R_2(B,C)$} edge[-] (C);
		\node at (-2, -3.5) (R3) {\small $R_3(C,A)$} edge[-] (C);
		\node at (-3.0, -3.0) (IR1) {\small $I_{A,B}{R_1}(A,B)$} edge[-] (C);
		\node at (-1.5, -1.0) (R4) {\small $R_4(A,D)$} edge[-] (A);
		\node at (1.0, 0.0) (R5) {\small $R_5(D,E)$} edge[-] (D);
		\node at (2.15, -3.0) (R6) {\small $R_6(B,F)$} edge[-] (F);
		\node at (2.15, -4.0) (R7) {\small $R_7(F,G)$} edge[-] (G);
		\node at (1.0, -4.0) (R8) {\small $R_8(C,H)$} edge[-] (H);
		\node at (0.0, -5.0) (R9) {\small $R_9(H,J)$} edge[-] (J);
		\node at (0,  -4) (invisible) {};
	\end{tikzpicture}
	}
\end{minipage}
\caption{Top left: The hypergraph of the query $Q$ in Example~\ref{ex:triangle_three_tails}. Remaining three: the optimal access-top VOs of the query $Q$ with the roots $A$, $D$ and $E$, respectively. All other access-top VOs are analogous to these three VOs. The dependent sets of the two VOs in the second row are omitted.}
\label{fig:general_large}
\end{figure}

The next example demonstrates that indicator projections are inevitable
when we want to construct VOs with minimal dynamic width
for cyclic CQAPs.
As stated in Theorem~\ref{thm:general}, low dynamic width implies low update time for CQAPs.

\begin{exa}
\label{ex:triangle_three_tails}
Consider the following query:
\begin{align*}
Q(A,B,C,D,E,F,G,H,J \mid \cdot) = & R_1(A,B), R_2(B,C), R_3(C,A),R_4(A,D), R_5(D,E),\\ 
&  R_6(B,F), R_7(F,G), R_8(C,H), R_9(H,J)
\end{align*}
It is a triangle query with three tails. The  fracture of the query is the same as the query.
Figure~\ref{fig:general_large} shows the hypergraph (top-left) of the query and three access-top VOs for the query. 
No other VO for the query has better static or dynamic width than these VOs.
%They are the optimal VOs that are rooted at variables $A$, $D$ and $E$. That is, other VOs rooted at the corresponding variable do not admit smaller static or dynamic width.
%Since the query is symmetric, the optimal VOs rooted at other variables are analogous to these three VOs. 

Consider the VO in the top right of Figure~\ref{fig:general_large}.
The subtree of the variable order rooted at $C$ has the atoms 
$\calR = \{R_3(C,A), R_2(B,C), R_8(C,H), R_9(H,J), R_5(D,E), R_4(A,D)\}$.
The atoms that are not in the subtree but whose 
schemas intersect with $\{C\} \cup dep(C) = \{A,B,C\}$ are 
$R_1(A,B)$ and $R_6(B,F)$. Hence, we consider the indicator projections  
$\calI = \{I_{A,B}R_1(A,B), I_{B}R_6(B)\}$. We have  $\gyo^*(\calI,\calR) = \{I_{A,B}R_1(A,B)\}$.
Therefore, $I_{A,B}R_1(A,B)$ becomes a child of $C$.
This indicator projection 
reduces the dynamic width of the variable order from $2$ to $\frac{3}{2}$, as explained next.

%The induced query $Q_C$ at $C$ contains the variables $\{C\} \cup dep(C) = \{A,B,C\}$. 
It holds $\{C\} \cup dep(C) = \{A,B,C\}$. We compute the maximal 
$\rho^*(\{A,B,C\} \setminus \calS)$, where $\calS$ is the schema of any atom in the
subtree of the variable order rooted at $C$.  
If we choose $\calS$ to be the schema of $R_9(H,J)$, we obtain 
$\{A,B,C\} \setminus \calS = \{A,B,C\}$. 
In case $I_{A,B}{R_1}$ is not included in the subtree rooted at $C$,
we have $\rho^*(\{A,B,C\} \setminus \calS) =2$. Otherwise, 
we have $\rho^*(\{A,B,C\} \setminus \calS) =\frac{3}{2}$
(by assigning a weight of $\frac{1}{2}$ to the indicator projection $I_{A,B}{R_1}$
and to each of the atoms $R_3$ and $R_2$).
For any other variable $X$ and atom $R(\calS)$ below $X$, 
 the fractional edge cover number $\rho^*((\{X\} \cup \dep(X))  \setminus \calS) $
is not greater than $1$.
Hence, we conclude that the dynamic width of the VO is $\frac{3}{2}$.
\nop{ 
dynamic width of the subtree $\omega_C$ rooted at $C$ is defined as the maximal fractional edge cover number of 
$\{C\} \cup dep(C) = \{A,B,C\}$ minus the schema of an atom below $C$. 
If we choose the atom to be $R_9(H,J)$, the remaining variables are still $\{A,B,C\}$.
With the indicator projection $I_{A,B}{R_1}$, the fractional edge cover number is $\rho^*(A,B,C) = \frac{3}{2}$ (by assigning a weight of $\frac{1}{2}$ to each atom $I_{A,B}{R_1}$, $R_3$ and $R_2$).
Without $I_{A,B}{R_1}$, the fractional edge cover number is $\rho^*(A,B,C) = 2$.
Hence, the indicator projection $I_{A,B}{R_1}$ reduces the dynamic width of $\omega_C$ from $2$ to $\frac{3}{2}$.
Since $\omega_C$ is the only subtree that has a dynamic width greater than 1, the dynamic width of the query $Q$ is $\frac{3}{2}$.
}

The two VOs in the second row of Figure~\ref{fig:general_large} are similar to the aforementioned VO: Each of them has 
the variables $A$, $B$, and $C$ on one root-to-leaf path, followed by the atom $R_9$, which has no intersection with 
$\{A, B, C\}$. 
The indicator projection $I_{A,B}{R_1}$ created under variable $C$ reduces the dynamic width from $2$ to $\frac{3}{2}$ in the same way. 
%Hence, the indicator projections can reduce the dynamic width of the query $Q$ for all VOs.
\qed
\end{exa}
	
Prior work defined the static and the dynamic width of conjunctive queries 
without access patterns~\cite{KaraNOZ2020}. It was shown that for any hierarchical conjunctive 
query with static width $\fw$ and dynamic width $\dfw$, it holds 
$\dfw=\fw$ or $\dfw=\fw-1$ 
(Proposition 3.7 in~\cite{KaraNOZ2020}).
The proof can easily be adapted to the width measures of CQAPs.
The only change is that we argue over access-top variable orders for the fractures 
of CQAPs instead of free-top variable orders for conjunctive queries. 

\begin{prop}[Corollary of Proposition 3.7 in~\cite{KaraNOZ2020}]
	\label{prop:width_delta_inequal}
	For any CQAP with hierarchical fracture, static width
	$\fw$ and dynamic width $\dfw$, it holds either
	$\dfw=\fw$ or $\dfw=\fw-1$.
\end{prop}

\subsection{From Canonical to Access-Top VOs}
\label{sec:canonical-free-top}
\begin{figure}[t]
	\centering
	\setlength{\tabcolsep}{3pt}
	\renewcommand{\arraystretch}{1.05}
	\renewcommand{\linenumber}{\makebox[2ex][r]{\rownumber\TAB}}
	\setcounter{magicrownumbers}{0}
	\begin{tabular}[t]{@{}c@{}c@{}l@{}}
	\toprule
  \multicolumn{3}{l}{$\ftvo$(VO $\omega$, access pattern $(\calO|\calI)$) : VO} \\
  \midrule
  \multicolumn{3}{l}{\MATCH $\omega$:}                                                                         \\
  \midrule
  \phantom{ab} & $R(\calY)$\hspace*{2.5em} & \linenumber \RETURN $R(\calY)$  \\[2pt]
  \cmidrule{2-3} \\[-6pt]
                &
  \begin{minipage}[t]{2cm}
    \vspace{-1.8em}
    \hspace*{-0.55cm}
    \begin{tikzpicture}[xscale=0.5, yscale=1]
      \node at (0,-2)  (n4) {$X$};
      \node at (-1.2,-3)  (X) {$\omega_1$} edge[-] (n4);
      \node at (0,-3)  (n2) {$\ldots$};
      \node at (1.2,-3)  (X) {$\omega_k$} edge[-] (n4);
    \end{tikzpicture}
  \end{minipage}
                &
  \begin{minipage}[t]{10cm}
    \vspace{-0.4cm}
    \linenumber \LET $\omega_i' = \ftvo(\omega_i, (\calO|\calI)), \forall i\in[k]$ \\[0.5ex]
	\linenumber \LET $\calD=
      \begin{cases}
        \emptyset & \text{if } X \in \calI \\[0.2ex]
        vars(\omega) \cap \calI, & \text{else if } X \in \calO \\[0.2ex]
        vars(\omega) \cap (\calI \cup \calO) & \text{otherwise}
      \end{cases}$\\[0.5ex]
    \linenumber\LET $\{\hat{\omega}^{i}_1, \dots, \hat{\omega}^{i}_{m_i}\} = \delete(\omega_i', \calD),\ \forall i \in [k]$ \\[0.5ex]
    \linenumber \LET $(X_1, {.}{.}{.}, X_{\ell})$  be 
    an ordering of the variables in $\calD$ \\
    \phantom{---} such that it  is compatible with the partial order of $\omega$ and\\
    \phantom{---} the input variables come before the output variables \\[0.5ex]

    \linenumber \RETURN $
      \begin{array}{@{~~}c@{~~}}
        \tikz {
        \node at (0,-1)  (first) {$X_1$};
        \node at (0,-1.3)  (point) {$\cdot$};
        \node at (0,-1.45)  (point) {$\cdot$};
        \node at (0,-1.6)  (point) {$\cdot$};
        \node at (0,-2)  (last) {$X_{\ell}$};
        \node at (0,-2.75)  (X) {$X$} edge[-] (last);
        \node at (-2.0,-3.5)  (n1) {$\hat{\omega}^{1}_1$} edge[-] (X);
        \node at (-1.5,-3.5)  (n2) {$\ldots$};
        \node at (-0.75,-3.5)  (n3) {$\hat{\omega}^1_{m_1}$} edge[-] (X);
        \node at (0,-3.5)  (n4) {$\ldots$};
        \node at (0.75,-3.5)  (n5) {$\hat{\omega}^{k}_1$} edge[-] (X);
        \node at (1.25,-3.5)  (n6) {$\ldots$};
        \node at (2.0,-3.5)  (n7) {$\hat{\omega}^k_{m_k}$} edge[-] (X);
        }
      \end{array}$ \\[0.5ex]
  \end{minipage}                                                                                   \\[2.75ex]
  \bottomrule
	\end{tabular}\vspace{-0.0em}
	\caption{Construction of an access-top VO from a canonical
		VO $\omega$ of a hierarchical CQAP with access pattern $(\calO|\calI)$.
		The function $\delete(\omega', \calD)$, defined in Figure~\ref{fig:node_deletion}, deletes the variables in $\calD$ from the VO $\omega'$.}
	\label{fig:canonical-to-free-top}
	\vspace*{-0.0em}
\end{figure}

\begin{figure}[t]
	\centering
	\setlength{\tabcolsep}{3pt}
	\renewcommand{\arraystretch}{1.05}
	\renewcommand{\linenumber}{\makebox[2ex][r]{\rownumber\TAB}}
	\setcounter{magicrownumbers}{0}
	\begin{tabular}[t]{@{}c@{}c@{}l@{}}
		\toprule
		\multicolumn{3}{l}{$\delete$(VO $\omega$, variables $\calD$) : set of VOs} \\
		\midrule
		\multicolumn{3}{l}{\MATCH $\omega$:}                                                               \\
		\midrule
		\phantom{ab} & $R(\calY)$\hspace*{2.5em} & \linenumber \RETURN $\{R(\calY)\}$                      \\[2pt]
		\cmidrule{2-3}                                                                                     \\[-6pt]
		             &
		\begin{minipage}[t]{2cm}
			\vspace{-1.8em}
			\hspace*{-0.55cm}
			\begin{tikzpicture}[xscale=0.5, yscale=1]
				\node at (0,-2)  (n4) {$Y$};
				\node at (-1.2,-3)  (n1) {$\omega_1$} edge[-] (n4);
				\node at (0,-3)  (n2) {$\ldots$};
				\node at (1.2,-3)  (n3) {$\omega_k$} edge[-] (n4);
			\end{tikzpicture}
		\end{minipage}
		             &
		\begin{minipage}[t]{10cm}
			\vspace{-0.4cm}
			\linenumber \LET $\{\omega^i_1, {.}{.}{.}, \omega^i_{m_i}\} = \delete(\omega_i, \calD), \ \forall i\in[k]$ \\[0.5ex]
			\linenumber \IF $Y \notin \calD$ \\[0.5ex]
			\linenumber \TAB \RETURN
			$\left\{
				\begin{array}{@{~~}c@{~~}}
					\tikz {
					\node at (0,-2.75)  (X) {$Y$};
					\node at (-1.8,-3.5)  (n1) {$\omega_1^{1}$} edge[-] (X);
					\node at (-1.3,-3.5)  (n2) {$\ldots$};
					\node at (-0.65,-3.5)  (n3) {$\omega^1_{m_1}$} edge[-] (X);
					\node at (0,-3.5)  (n4) {$\ldots$};
					\node at (0.65,-3.5)  (n5) {$\omega^{k}_1$} edge[-] (X);
					\node at (1.15,-3.5)  (n6) {$\ldots$};
					\node at (1.8,-3.5)  (n7) {$\omega^k_{m_k}$} edge[-] (X);
					}
				\end{array}
				\right\}$
			\\[0.5ex]
			\nop{\linenumber \ELSE\IF $Y$ has parent $Z$ \\[0.5ex]
			\linenumber \TAB \RETURN
			$\left\{
				\begin{array}{@{~~}c@{~~}}
					\tikz {
					\node at (0,-2.75)  (X) {$Z$};
					\node at (-1.8,-3.5)  (n1) {$\omega_1^{1}$} edge[-] (X);
					\node at (-1.3,-3.5)  (n2) {$\ldots$};
					\node at (-0.65,-3.5)  (n3) {$\omega^1_{m_1}$} edge[-] (X);
					\node at (0,-3.5)  (n4) {$\ldots$};
					\node at (0.65,-3.5)  (n5) {$\omega^{k}_1$} edge[-] (X);
					\node at (1.15,-3.5)  (n6) {$\ldots$};
					\node at (1.8,-3.5)  (n7) {$\omega^k_{m_k}$} edge[-] (X);
					}
				\end{array}
				\right\}$ \\[0.5ex]
			}
			\linenumber \RETURN $\left\{\omega^1_1, {.}{.}{.}, \omega^1_{m_1}, {.}{.}{.}, \omega^k_1, {.}{.}{.}, \omega^k_{m_k} \right\}$
		\end{minipage} \\[2.75ex]
		\bottomrule
	\end{tabular}
	\caption{Deletion of  a set $\calD$ of variables from a  VO $\omega$.
		In case $\omega$ has a root variable $Y$, the variables in $\calD$ are first deleted from the
		 child trees of $Y$. If $Y$ is included in $\calD$,   the child trees of $Y$ become a forest of trees without any common root.}
		\nop{If $Y \in \calD$ and $Y$ has a parent $Z$, the child trees  of $Y$ are appended to  $Z$. If $Y \in \calD$ and $Y$ has no parent, the child trees of $Y$ become independent.}
	\label{fig:node_deletion}
\end{figure}

Given a canonical VO $\omega$ of a hierarchical CQAP $Q$ with input variables $\calI$ and output variables $\calO$, the function $\ftvo(\omega, (\calO|\calI))$ in Figure~\ref{fig:canonical-to-free-top} returns an access-top VO for $Q$
whose static and dynamic widths equal the corresponding widths of $Q$.

First, we give the high-level idea of the construction.  
At each variable $X$, the function {\em pulls up} some variables from the subtree rooted at $X$, which means that 
it deletes these variables from the subtree  and puts them on a path
on top of $X$. 
\nop{If $X$ is an input variable, no variable needs to be pulled up.} 
If $X$ is an output variable, all input variables in the subtree are pulled up. 
If it is a bound variable, all free variables in the subtree are pulled up. 
In the newly constructed path on top of $X$, the input variables are placed on top of the output variables.

We explain the function $\ftvo(\omega, (\calO|\calI))$ in more detail. 
It proceeds recursively on the structure of $\omega$.
Consider a variable $X$ in $\omega$ and assume that the child trees of $X$ are already access-top.  
The function  defines a set $\calD$ of variables (Line 3) that are going to be deleted from the subtree $\omega_X$ rooted at $X$
and put on a path on top of $X$. 
If $X$ is an input variable (Case 1 in Line 3), then $\calD$ is empty, which means that we do not need to pull up 
any variable from $\omega_X$.
If $X$ is an output variable  (Case 2 in Line 3), then $\calD$ consists of all input variables in $\omega_X$.
If $X$ is a bound variable (Case 3 in Line 3), then $\calD$ consists of all free  variables in $\omega_X$.
The deletion of the variables in $\calD$ from $\omega_X$ (Line 4) is implemented 
by the function
$\delete$ in Figure \ref{fig:node_deletion}, which we explain in more detail further below.
The top-down ordering of the new path  constructed from the variables in $\calD$ 
respects the partial ordering defined by $\omega_X$ and has  the input variables on top of the output variables (Line 5). 
Observe that this is possible, since the child trees of $X$ are already access-top.

Given a variable order $\omega'$ and a set $\calD$ of variables  to be deleted from $\omega'$,
the function $\delete(\omega,'\calD)$ in Figure \ref{fig:node_deletion}
traverses recursively over each variable $Y$ in $\omega'$ with child trees 
$\omega_1, \ldots, \omega_k$.
First, the function  deletes the variables in $\calD$ from the child trees of $Y$
and obtains the trees $\calT = \{\delete(\omega_1, \calD), \ldots, \delete(\omega_k, \calD)\}$ (Line~2). 
If $Y$ is not included in $\calD$, the function returns the tree with root $Y$ and child trees 
$\calT$  (Lines 3-4). 
Otherwise, it returns the forest $\calT$ (Line 5).

\begin{exa}
\label{ex:access-top-large}
\rm
Consider the query 
$$Q(C,D \mid E) = R(A,B,C), S(A,B,D), T(A,E),$$
which is hierarchical but not free-dominant. 
Figure~\ref{fig:preprocessing_large_hypergraphs}  shows the hypergraph 
and the canonical VO $\omega$ of the query (top row).
We explain an intermediate and the final step of the function 	
	 $\ftvo(\omega(\{C,D\} | \{E\}))$ in Figure~\ref{fig:canonical-to-free-top} that transforms $\omega$
	 into an access-top VO. 

At variable $B$ in $\omega$, the function determines that $B$ is bound and its two children are free. Hence, the function 
 moves 
$C$ and $D$ on a path above $B$. 
%Since these variables come from different branches below $B$, their ordering does not play a role. 
Figure~\ref{fig:canonical-to-free-top} (bottom row, left)
shows the VO $\omega'$ that arises from this step. At variable $A$ in $\omega'$, the function determines that $A$
is bound and the children $C$, $D$, and $E$  are free. Thus, it puts the latter variables on 
a path on top of $A$ such that  the input variable $E$ sits on top of the output variables $C$ and $D$. 
Figure~\ref{fig:canonical-to-free-top} (bottom row, right) shows the resulting access-top 
VO, which is the final VO returned by the function $\ftvo(\omega(\{C,D\} | \{E\}))$.
\qed
\end{exa}

\begin{figure}[t]
\begin{minipage}[b]{0.33\linewidth}
  \centering
  \begin{tikzpicture}[xscale=1.1, yscale=1.0]
    \node at (-0.25, 0.0) (A) {\small  $A$};
    \node at (-1.0, -1.0) (B) {\small $B$};
    \node at (0.5, -1.0) (E) {\small $E$};
    \node at (-1.75, -2.0) (C) {\small $C$};
    \node at (-0.5, -2.0) (D) {\small $D$};
    \node at (0.5, -1.75) (T) {\small $T(A,E)$};
    \node at (-1.95, -2.75) (R) {\small  $R(A,B,C)$};
    \node at (-0.3, -2.75) (S) {\small  $S(A,B,D)$};
    \begin{pgfonlayer}{background}
      \draw[opacity=.5,fill opacity=.5,line cap=round, line join=round, line width=18pt,color=teal] (-0.25,0.0) -- (-1,-1) -- (-1.75, -2.0);
      \draw[opacity=.5,fill opacity=.5,line cap=round, line join=round, line width=15pt,color=yellow] (-0.25,0.0) -- (-1,-1) -- (-0.5, -2.0);
      \draw[opacity=.5,fill opacity=.5,line cap=round, line join=round, line width=15pt,color=orange] (-0.25,0.0) -- (0.5,-1.0);
    \end{pgfonlayer}
  \end{tikzpicture}
\end{minipage}
\hspace{3em}
\begin{minipage}[b]{0.33\linewidth}
  % \vspace*{0.5cm}
  \centering
  \begin{tikzpicture}[xscale=1.1, yscale=1.0]
    \node at (0.0, 0.0) (A) {\small  $A$};
    \node at (-1.0, -1.0) (B) {\small  $B$} edge[-] (A);
    \node at (-2.0, -2.0) (C) {\small  $C$} edge[-] (B);
    \node at (-0.0, -2.0) (D) {\small  $D$} edge[-] (B);
    \node at (-2.0, -3.0) (R) {\small  $R(A,B,C)$} edge[-] (C);
    \node at (0.0, -3.0) (S) {\small  $S(A,B,D)$} edge[-] (D);

    \node at (1.0, -1.0) (E) {\small  $E$} edge[-] (A);
    \node at (1.0, -2.0) (T) {\small  $T(A,E)$} edge[-] (E);
  \end{tikzpicture}
\end{minipage}

\vspace{1em}
\begin{minipage}[b]{0.32\linewidth}
  \centering
  \begin{tikzpicture}[xscale=1.1, yscale=1.0]
    \node at (0.0, 0.0) (A) {\small  $A$};
    \node at (0.0, -0.8) (C) {\small  $C$} edge[-] (A);
    \node at (0.0, -1.6) (D) {\small  $D$} edge[-] (C);
    \node at (0.0, -2.4) (B) {\small  $B$} edge[-] (D);

    \node at (1.2, -0.7) (E) {\small  $E$} edge[-] (A);
    \node at (-1, -3.2) (R) {\small  $R(A,B,C)$} edge[-] (B);
    \node at (0.7, -3.2) (S) {\small  $S(A,B,D)$} edge[-] (B);
    \node at (1.2, -1.5) (T) {\small  $T(A,E)$} edge[-] (E);
  \end{tikzpicture}
\end{minipage}
\hspace{3em}
\begin{minipage}[b]{0.32\linewidth}
  \centering
  \begin{tikzpicture}[xscale=1.1, yscale=1.0]
    \node at (0.0, 0.0) (E) {\small  $E$};
    \node at (0.0, -0.7) (C) {\small  $C$} edge[-] (E);
    \node at (0.0, -1.4) (D) {\small  $D$} edge[-] (C);
    \node at (0.0, -2.1) (A) {\small  $A$} edge[-] (D);
    \node at (-0.5, -2.8) (B) {\small  $B$} edge[-] (A);
    \node at (-1.5, -3.5) (R) {\small  $R(A,B,C)$} edge[-] (B);
    \node at (0.5, -3.5) (S) {\small  $S(A,B,D)$} edge[-] (B);
    \node at (1.0, -2.8) (T) {\small  $T(A,E)$} edge[-] (A);
  \end{tikzpicture}
\end{minipage}
\caption{Top row: Hypergraph of the query from Example~\ref{ex:access-top-large} and its canonical 
VO $\omega$. Bottom row: An intermediate and the final VO constructed by the 
function $\ftvo(\omega,  (\{C,D\}, \{E\}))$ in Figure~\ref{fig:canonical-to-free-top}.}
\label{fig:preprocessing_large_hypergraphs}
\end{figure}

The function $\ftvo$ in Figure~\ref{fig:canonical-to-free-top} turns canonical VOs into 
optimal VOs. The proof of the following proposition is given in Appendix~B.2 of the technical report~\cite{access_pattern_arxiv}.
\begin{prop}
	\label{prop:canonical-to-free-top}
  Given a CQAP $Q$, whose fracture $Q_\dagger(\calO|\calI)$ is hierarchical, and a canonical VO $\omega$ for 
  $Q_\dagger$,	$\ftvo(\omega, (\calI| \calO))$ constructs 	an access-top VO for $Q_\dagger$ with static width $\fw(Q)$ and dynamic width $\dfw(Q)$.
\end{prop}

%% file: results.tex
\section{Complexity of Dynamic CQAP Evaluation}
\label{sec:results}
%We  give an overview of our main results. 
In this work, we introduce a fully dynamic evaluation approach for arbitrary CQAPs whose complexity is stated in the following theorem.

\begin{thm}
  \label{thm:general}
  Given a CQAP with static width $\fw$ and dynamic width $\dfw$ and a database of size $N$,
  the query can be evaluated  with
    $\bigO{N^\fw}$ preprocessing time,  $\bigO{N^{\dfw}}$ update time under single-tuple updates, and 
    $\bigO{1}$ enumeration delay.
\end{thm}

Our approach has three stages: preprocessing, enumeration, and updates. They are explained in Sections
\ref{sec:preprocessing}, \ref{sec:enumeration}, and \ref{sec:updates}, respectively. 
  Given a CQAP with static width $\fw$ and dynamic width $\dfw$ and a database of size $N$,
we construct in the preprocessing stage a set of view trees in $\bigO{N^{\fw}}$ time 
that represent the result of the query (Proposition~\ref{prop:preprocessing}).
Using these view trees, we can enumerate with constant delay the tuples over the output variables, given any tuple over the input variables  (Proposition~\ref{prop:enumeration}). 
 The view trees can be maintained with $\bigO{N^{\dfw}}$ update time under single-tuple updates 
 to the base relations (Proposition~\ref{prop:updates}).
The full proof of Theorem~\ref{thm:general} is given in Appendix~F of the technical report~\cite{access_pattern_arxiv}.
%%%%%%%%%%%%%%

The following dichotomy states that the queries in CQAP$_0$ are precisely those CQAPs that can be evaluated with constant update time and enumeration delay. 

\begin{thm}
	\label{thm:dichotomy}
	Let any CQAP $Q$ and database of size $N$.
	\begin{itemize}
	\item If $Q$ is in $\text{CQAP}_0$, then it admits  
	$\bigO{N}$ preprocessing time,
	$\bigO{1}$ enumeration delay, and $\bigO{1}$  
	update time for single-tuple updates.
	\item If $Q$ is not in $\text{CQAP}_0$ and has no repeating relation symbols,
		then there is no  algorithm that computes $Q$ with arbitrary preprocessing  time, 
		$\bigO{N^{\frac{1}{2} - \gamma}}$ enumeration delay, and 
		$\bigO{N^{\frac{1}{2} - \gamma}}$ amortised update time,
		for any 
			$\gamma >0$, unless the \OMv conjecture fails. 
	\end{itemize}
\end{thm}
We prove 	Theorem~\ref{thm:dichotomy} in Section~\ref{sec:dichotomy}.
The hardness result in the theorem is based on the following \OMv problem:

 \begin{defi}[Online Matrix-Vector Multiplication (\OMv)~\cite{Henzinger:OMv:2015}]\label{def:OMv}
We are given an $n \times n$ Boolean matrix $\inst{M}$ and  receive $n$ Boolean column vectors $\inst{v}_1, \ldots, \inst{v}_n$ of size $n$, one by one; after seeing each vector $\inst{v}_i$, we output the product $\inst{M} 
\inst{v}_i$ before we see the next vector.
\end{defi}
 
It is strongly believed that the $\OMv$ problem cannot be solved in subcubic time.

\begin{conj}[\OMv Conjecture, Theorem 2.4~\cite{Henzinger:OMv:2015}]
\label{conj:omv}
For any $\gamma > 0$, there is no algorithm that solves the 
 \OMv problem in time $\bigO{n^{3-\gamma}}$.
\end{conj}

Queries in CQAP$_0$ have dynamic width $0$ and static 
width $1$ (Proposition~\ref{prop:cqap0_delta0}).
Our approach from   Theorem~\ref{thm:general} achieves 
linear preprocessing time, constant update time and enumeration delay
for such queries, so it is optimal for CQAP$_0$. 

The smallest queries not included in CQAP$_0$ are:
$Q_1(\calO |\cdot)= R(A), S(A,B), T(B)$  with  $\calO \subseteq \{A,B\}$; 
$Q_2(A |\cdot)= R(A,B),S( B)$; 
$Q_3( \cdot | A )=  R(A,B),S(B)$; and 
$Q_4(B |A)= R(A,B)$, $S(B)$.  
Each of these queries is equal to its fracture. 
Query $Q_1$ is not hierarchical. 
$Q_2$ is not free-dominant.  
$Q_3$ and $Q_4$ are not input-dominant.
Prior work showed that there is no algorithm that achieves constant
update time and enumeration delay for $Q_1$ and $Q_2$, unless the \OMv 
conjecture fails~\cite{BerkholzKS17}.
To prove the hardness statement in Theorem~\ref{thm:dichotomy}, we show in Section~\ref{sec:dichotomy}
that this negative result also holds for $Q_3$ and $Q_4$. Then, given an arbitrary CQAP $Q$ that is not in CQAP$_0$, we reduce the evaluation of one of the four queries above to the evaluation of $Q$.

%%%%%%%%%%%%%%%%%%
For CQAPs with hierarchical fractures, the complexities in Theorem~\ref{thm:general} can be parameterised to uncover trade-offs between preprocessing, update, and enumeration.

\begin{thm}
	\label{thm:main_dynamic}
	Let any CQAP $Q$ with static width $\fw$ and dynamic width $\dfw$, a database of size $N$, and $\eps \in [0,1]$. If $Q$'s fracture is hierarchical, then $Q$ admits $\bigO{N^{1 + (\fw -1)\eps}}$ preprocessing time, $\bigO{N^{1-\eps}}$ enumeration delay, and $\bigO{N^{\dfw\eps}}$ amortised  update time for single-tuple updates.
\end{thm}
We illustrate in Section~\ref{sec:trade-off} the core ideas of our algorithm achieving the trade-offs in Theorem~\ref{thm:main_dynamic}. 
The full proof of the theorem can be found in Appendix~F of the technical report~\cite{access_pattern_arxiv}.
\nop{the technical report~\cite{free-access-pat-arxiv}.}
The trade-off continuum in Theorem~\ref{thm:main_dynamic} can be obtained using one algorithm parameterised by 
$\epsilon$. In Section~\ref{sec:comparison}, we show that this algorithm either recovers or has lower complexity than prior approaches.
Using $\eps=1$, we recover the complexities in Theorem~\ref{thm:general} and therefore also the constant update time and delay for queries in $\text{CQAP}_0$ in Theorem~\ref{thm:dichotomy}.

Theorem~\ref{thm:main_dynamic} can be refined for $\text{CQAP}_1$, since $\dfw=1$ and $\fw\leq 2$ for queries in this class.

\begin{cor}[Theorem~\ref{thm:main_dynamic}]
\label{cor:cqap1}
	Let any query in $\text{CQAP}_1$, a database of size $N$, and $\eps \in [0,1]$. Then $Q$ admits $\bigO{N^{1 + \eps}}$ preprocessing time, $\bigO{N^{1-\eps}}$ enumeration delay, and $\bigO{N^{\eps}}$ amortised update time for single-tuple updates.
\end{cor}

The proof of the corollary is given in Appendix~F of the technical report~\cite{access_pattern_arxiv}.
For $\eps=0.5$, the amortised update time and the delay for queries in $\text{CQAP}_1$ match the lower bound in Theorem~\ref{thm:dichotomy} for all queries outside $\text{CQAP}_0$. This makes our approach weakly Pareto optimal for $\text{CQAP}_1$, as lowering both the amortised update time and the delay would violate the $\OMv$ conjecture. 

%% file: preprocessing.tex
\section{Preprocessing}
\label{sec:preprocessing}
In this section, we describe the preprocessing stage of our approach for the dynamic evaluation of 
arbitrary CQAPs. 
Consider in the following a CQAP $Q$, its fracture $Q_{\dagger}$, and a database of size $N$.

\begin{figure}[t]
	\centering
	\setlength{\tabcolsep}{3pt}
	\renewcommand{\arraystretch}{1.05}
	\renewcommand{\linenumber}{\makebox[2ex][r]{\rownumber\TAB}}
	\setcounter{magicrownumbers}{0}
	\begin{tabular}[t]{@{}c@{}c@{}l@{}}
		\toprule
		\multicolumn{3}{l}{$\vt$(\text{VO} $\nu$) : view tree}   \\
		\midrule
		\multicolumn{3}{l}{\MATCH $\nu$:}                                    \\
		\midrule
		\phantom{ab} & $R(\calY)$\hspace*{2.5em} & \linenumber \RETURN
		$R(\calY)$ \\[2pt]
		\cmidrule{2-3} \\[-6pt]
		             &
		\begin{minipage}[t]{0.15\linewidth}
			\vspace{-1.7em}
			\hspace*{-0.35cm}
			\begin{tikzpicture}[xscale=0.5, yscale=1]
				\node at (0,-2)  (n4) {$X$};
				\node at (-1.2,-3)  (n1) {$\nu_1$} edge[-] (n4);
				\node at (0,-3)  (n2) {$\ldots$};
				\node at (1.2,-3)  (n3) {$\nu_k$} edge[-] (n4);
			\end{tikzpicture}
		\end{minipage}
		             &
		\begin{minipage}[t]{0.75\linewidth}
			\vspace{-0.4cm}
      \linenumber \LET $T_i = \vt(\nu_i)$,  \ $\forall i\in[k]$ \\[0.5ex]
			\linenumber \LET $\calS = \{X\} \cup \dep_{\omega}(X)$ \\[0.5ex]
			\linenumber \LET $V_X(\calS) =$ join of roots of $T_1, {.}{.}{.}, T_k$ with variables not in $\calS$ \\
      \TAB\TAB\TAB\TAB\TAB\TAB\TAB\TAB\TAB\TAB\TAB\TAB\TAB\TAB\TAB\TAB\TAB\TAB\TAB  marginalised out \\[0.5ex]
			\linenumber \IF $X$ has no sibling \TAB\RETURN $
				\left\{
				\begin{array}{@{~~}c@{~~}}
					\tikz {
						\node at (1.2,-1)  (n4) {$V_X(\calS)$};
						\node at (0.6,-1.75)  (n1) {$T_1$} edge[-] (n4);
						\node at (1.2,-1.75)  (n2) {$\ldots$};
						\node at (1.8,-1.75)  (n3) {$T_k$} edge[-] (n4);
					}
				\end{array}  \right.$ \\[0.5ex]
			\linenumber \LET $V'_X(\calS\setminus \{X\}) = V_X(\calS)$ with variable $X$ marginalised out \\[0.5ex]
      \linenumber \RETURN $
				\left\{
				\begin{array}{@{~~}c@{~~}}
					\tikz {
						\node at (1.2,-0.15)  (n1) {$V_X'(\calS \setminus \{X\})$};
						\node at (1.2,-1)  (n4) {$V_X(\calS)$} edge[-] (n1);
						\node at (0.6,-1.75)  (n1) {$T_1$} edge[-] (n4);
						\node at (1.2,-1.75)  (n2) {$\ldots$};
						\node at (1.8,-1.75)  (n3) {$T_k$} edge[-] (n4);
					}
				\end{array}  \right.$ \\[0.5ex]
		\end{minipage}                                              \\[2.75ex]
		\bottomrule
	\end{tabular}\vspace{-0.1em} 
	\caption{The function $\tau$ constructs a view tree for a given VO $\omega$. It is defined using pattern matching on the structure of $\omega$, which can be a leaf or an inner node (cf.\@ left column under $\MATCH$).
	At each variable $X$ in $\omega$, the function defines a new view $V_X$ whose free variables $\calS$ are $X$ and the dependency set of $X$; its body is the join of the views defined at the variables that are roots of the child VOs of $X$. If $X$ has siblings, it defines a new view on top of $V_X$ so that $X$ becomes bound in $V_X$ (so it is marginalised).
    Note that when $X$ has an atom $R(\calS)$ as its only child in $\omega$, the new view $V_X(\calS)$ is redundant; for simplicity, we retain this view.
		}
	\label{fig:general_view_tree_construction}
\end{figure}
In the preprocessing stage, we construct a set of view trees that represent the result of  
$Q_\dagger$ over both its input and output variables. A view tree~\cite{Nikolic:SIGMOD:18} is a (rooted) tree with one view per node. 
It is a logical project-join plan in the classical database systems literature, but where each intermediate result is materialised.  The view at a node is defined as the join of the views at its children, possibly followed by a projection. The view trees are modelled following an access-top VO $\omega$ of $Q_\dagger$.
In the following, we discuss the case of $\omega$ consisting of a single tree; otherwise, we apply the preprocessing stage to each tree in $\omega$.

Given an access-top VO $\omega$ for $Q_{\dagger}$, the function
$\vt(\omega)$ in
Figure~\ref{fig:general_view_tree_construction} returns 
a view tree constructed from $\omega$.
The function recursively traverses $\omega$ bottom-up (Line~2) and creates
at each variable $X$, a view $V_X$
defined over the join 
of the views created for the children of $X$.
The schema of $V_X$
consists of $X$ and the dependency set of $X$
(Line~3). This view allows to efficiently enumerate
the $X$-values given a tuple of values for the variables in the dependency set.
If $X$ has siblings, the function
creates an additional view $V'_X$ on top of $V_X$ to aggregate away (or marginalise out) $X$ from $V_X$ (Line~6).
This view allows to efficiently maintain the ancestor views of $V_X$ under updates
to the views created for the siblings of $X$.

The next example demonstrates the construction of the view trees for a query in CQAP$_0$.
The construction time is linear in the database size.
  
  \begin{exa}
  \label{ex:general_preprocessing_CQAP0}
  Figure~\ref{fig:general_hypergraphs} shows the hypergraphs of the query  
  $Q(B,C,D,E | A) =$ $R(A,B,C),$ $S(A,B,D),$ $T(A,E)$ 
  and its fracture 
  $Q_\dagger(B,C,D,E | A_1,A_2) =$ $R(A_1,B,C),$ $S(A_1,B,D),$ $T(A_2,E)$.  
  The fracture has two connected components: 
  $Q_1(B,C,D|A_1) = R(A_1,B,C)$, $S(A_1,$ $B,D)$ and  $Q_2(E|A_2)=T(A_2,E)$.
  Figure~\ref{fig:general_CPAP_0} depicts an access-top  VO (left) for 
  $Q_1$ and its corresponding view tree (middle).
  The VO has static width $1$. 
  Each variable in the VO is mapped to a view in the view tree, e.g.,  
  $B$ is mapped to $V_B(A_1,B)$, where $\{B,A_1\}=\{B\}\cup\dep(B)$.
  The views $V_C'$ and $V_{D}'$  are auxiliary views. 
  %that allow for efficient maintenance under updates to $R$ and $S$.
  The views $V_C'$, $V_{D}'$, and $V_{A_1}$  marginalise out the 
  variables $C$, $D$ and respectively $B$ from their child views. 
  The view $V_B$ is the intersection of $V_C'$ and $V_{D}'$. 
  Hence, all views can be computed in $\bigO{N}$ time.
  Since the query fracture is acyclic, the view tree does not contain indicator projections.

  The only access-top VO for the connected component $Q_2$ of $Q_\dagger$ is the top-down path 
  $A_2- E-T(A_2,E)$. The views mapped to $A_2$ and $E$ are $V_{A_2}(A_2)$ and respectively 
  $V_{E}(A_2,E)$. They can obviously be computed in $\bigO{N}$ time.
  \qed
  \end{exa}

  \begin{figure}[t]
    \centering
    \begin{minipage}[b]{0.35\linewidth}
      \begin{tikzpicture}[xscale=0.96, yscale=0.8]
        \node at (-0.25, 0.0) (A) {\small  $A$};
        \node at (-1.0, -1.0) (B) {\small $B$};
        \node at (0.5, -1.0) (E) {\small $E$};
        \node at (-1.75, -2.0) (C) {\small $C$};
        \node at (-0.5, -2.0) (D) {\small $D$};
        \node at (0.7, -1.75) (T) {\small $T(A,E)$};
        \node at (-2.35, -2.75) (R) {\small  $R(A,B,C)$};
        \node at (-0.2, -2.75) (S) {\small  $S(A,B,D)$};
        \begin{pgfonlayer}{background}
          \draw[opacity=.5,fill opacity=.5,line cap=round, line join=round, line width=18pt,color=teal] (-0.25,0.0) -- (-1,-1) -- (-1.75, -2.0);
          \draw[opacity=.5,fill opacity=.5,line cap=round, line join=round, line width=15pt,color=yellow] (-0.25,0.0) -- (-1,-1) -- (-0.5, -2.0);
          \draw[opacity=.5,fill opacity=.5,line cap=round, line join=round, line width=15pt,color=orange] (-0.25,0.0) -- (0.5,-1.0);
        \end{pgfonlayer}
      \end{tikzpicture}
    \end{minipage}
  \hspace{0.25cm}
    \begin{minipage}[b]{0.35\linewidth}
      \begin{tikzpicture}[xscale=0.96, yscale=0.8, 
        he/.style={draw, rounded corners,inner sep=0pt},        % he = hyper edge
        ce/.style={draw,dashed, rounded corners=10pt} % ce = condition edge
        ]
        \node at (-1, 0.0) (A) {\small  $A_1$};
        \node at (-1.0, -1.0) (B) {\small $B$};
        \node at (-1.75, -2.0) (C) {\small $C$};
        \node at (-0.5, -2.0) (D) {\small $D$};
        \node at (-2.35, -2.75) (R) {\small  $R(A_1,B,C)$};
        \node at (-0.1, -2.75) (S) {\small  $S(A_1,B,D)$};
        \begin{pgfonlayer}{background}
          \draw[opacity=.5,fill opacity=.5,line cap=round, line join=round, line width=18pt,color=teal] (-1,0.0) -- (-1,-1) -- (-1.75, -2.0);
          \draw[opacity=.5,fill opacity=.5,line cap=round, line join=round, line width=15pt,color=yellow] (-1,0.0) -- (-1,-1) -- (-0.5, -2.0);
        \end{pgfonlayer}
      \end{tikzpicture}
    \end{minipage}
      \begin{minipage}[b]{0.15\linewidth}
      \begin{tikzpicture}[xscale=0.96, yscale=0.8, 
        he/.style={draw, rounded corners,inner sep=0pt},        % he = hyper edge
        ce/.style={draw,dashed, rounded corners=10pt} % ce = condition edge
        ]
        \node at (0.75, 0.0) (A2) {\small  $A_2$};
        \node at (0.75, -1.0) (E) {\small $E$};
        \node at (0.75, -1.75) (T) {\small $T(A_2,E)$};
        \begin{pgfonlayer}{background}
          \draw[opacity=.5,fill opacity=.5,line cap=round, line join=round, line width=15pt,color=orange] (0.75,0.0) -- (0.75,-1.0);
        \end{pgfonlayer}
      \end{tikzpicture}
    \end{minipage}
    \caption{(Left) Hypergraph of the two queries with the same body but different access patterns, as used in 
    Examples~\ref{ex:general_preprocessing_CQAP0} and \ref{ex:general_preprocessing_CQAP1}; 
    (middle and right) hypergraph of their fractures.}
    \label{fig:general_hypergraphs}
    \end{figure}

  \begin{figure}[t]
    \centering
    \hspace{-0.42cm}
  \begin{minipage}[b]{0.2\linewidth}
    \centering
    \begin{tikzpicture}[xscale=0.925, yscale=0.7]
  \begin{scope}[yshift=1.4cm, xshift=0.2cm]
  \node at (0.2, -1.2) (A) {\small  $\dep(A_1)=\emptyset$};	
  \node at (0.2, -1.8) (B) {\small  $\dep(B)=\{A_1\}$};	
  \node at (0.2, -2.4) (C) {\small  $\dep(C)=\{A_1,B\}$};	
  \node at (0.2, -3.0) (D) {\small  $\dep(D)=\{A_1,B\}$};
  \end{scope}
  \node at (0,  -4) (invisible) {};
    \end{tikzpicture}
  \end{minipage}
  \hspace{-1.4cm}
  \begin{minipage}[b]{0.28\linewidth}
    \centering
    \begin{tikzpicture}[xscale=0.925, yscale=0.7]
      \node at (-1, -1.0) (A) {\small  $A_1$};
      \node at (-1.0, -2.0) (B) {\small $B$} edge[-] (A);
      \node at (-1.95, -4.0) (C) {\small $C$} edge[-] (B);
      \node at (-0.05, -4.0) (D) {\small $D$} edge[-] (B);
      \node at (-2.15, -5.0) (R) {\small  $R(A_1,B,C)$} edge[-] (C);
      \node at (0.15, -5.0) (S) {\small  $S(A_1,B,D)$} edge[-] (D);	
      \node at (0,  -4) (invisible) {};
    \end{tikzpicture}
  \end{minipage}
  \hspace{-0.1cm}
  \begin{minipage}[b]{0.3\linewidth}
    \centering
    \begin{tikzpicture}[xscale=0.9, yscale=0.7]
      \node at (-0.9, -1.0) (B') {\small $V_{A_1}(A_1)$};
      \node at (-0.9, -2.0) (B) {\small $V_B(A_1,B)$} edge[-] (B');
      \node at (-2.15, -3.0) (C') {\small $V'_C(A_1,B)$} edge[-] (B);
      \node at (-2.15, -4.0) (C) {\small $V_C(A_1,B,C)$} edge[-] (C');
      \node at (-2.15, -5.0) (R) {\small  $R(A_1,B,C)$} edge[-] (C);
      \node at (0.4, -3.0) (D') {\small $V'_D(A_1,B)$} edge[-] (B);
      \node at (0.4, -4.0) (D) {\small $V_D(A_1,B,D)$} edge[-] (D');
      \node at (0.4, -5.0) (S) {\small  $S(A_1,B,D)$} edge[-] (D);
    \end{tikzpicture}
  \end{minipage}
  \hspace{-0.1cm}
  \begin{minipage}[b]{0.3\linewidth}
    \centering
    \begin{tikzpicture}[xscale=0.9, yscale=0.7]
      \node at (-0.9, -1.0) (B') {\small {$\delta V_{A_1}(a)$}};
      \node at (-0.9, -2.0) (B) {\small {$\delta V_B(a,b)$}} edge[-] (B');
      \node at (-2.05, -3.0) (C') {\small {$\delta V'_C(a,b)$}} edge[-] (B);
      \node at (-2.05, -4.0) (C) {\small {$\delta V_C(a,b,c)$}} edge[-] (C');
      \node at (-2.05, -5.0) (R) {\small  {$\delta R(a,b,c)$}} edge[-] (C);
      \node at (0.3, -3.0) (D') {\small $V'_D(A,B)$} edge[-] (B);
      \node at (0.3, -4.0) (D) {\small $V_D(A_1,B,D)$} edge[-] (D');
      \node at (0.3, -5.0) (S) {\small  $S(A_1,B,D)$} edge[-] (D);
    \end{tikzpicture}
  \end{minipage}
  \caption{(Left) Access-top VO for 
  $Q_1(B,C,D|A_1) = R(A_1,B,C), S(A_1,$ $B,D)$; (middle) the view tree constructed from the VO;
  (right) the delta view tree under a single-tuple update to $R$.}
  \label{fig:general_CPAP_0}
  \end{figure}

The next example considers a query in CQAP$_1$ where the view tree construction time 
is quadratic in the database size.   
  
\begin{exa}
\label{ex:general_preprocessing_CQAP1}
Consider the query 
$Q(E,D|A,C) = R(A,B,C), S(A,B,D), T(A,E)$ in CQAP$_1$
  and its fracture 
$Q_\dagger(E,D | A_1, A_2,C)$ $=$ $R(A_1,B,C), S(A_1,B,D), T(A_2,E).$
The fracture has the two connected components 
  $Q_1(B,D|A_1,C) = R(A_1,B,C), S(A_1,B,D)$ and  $Q_2(E|A_2)=T(A_2,E)$.
  The hypergraphs (Figure~\ref{fig:general_hypergraphs}) of $Q$ and its fracture are the same as for the query in 
  Example~\ref{ex:general_preprocessing_CQAP0}.
  Figure~\ref{fig:general_CPAP_1} depicts an access-top VO (left) for 
  $Q_1$ and its corresponding view tree (middle).
  The VO has static width $2$. 
The view $V_B$ joins the relations $R$ and $S$, which takes  
$\bigO{N^2}$ time.
The views $V_{D}$, $V_{C}$, and $V_{A}$ are constructed from $V_B$
by marginalising out one variable at a time. Hence, the view tree 
construction takes $\bigO{N^2}$ time.
The view tree for $Q_2$ is the same 
as in Example~\ref{ex:general_preprocessing_CQAP0} and can be constructed in linear time. 
\qed
\end{exa}

\begin{figure}[t]
  \centering  
  \hspace{-0.5cm}
\begin{minipage}[b]{0.25\linewidth}
  \centering
	\begin{tikzpicture}[xscale=0.925, yscale=0.7]
\begin{scope}[yshift=2.2cm, xshift=0.2cm]
\node at (0.2, -1.2) (A) {\small  $\dep(A_1)=\emptyset$};	
\node at (0.2, -1.8) (B) {\small  $\dep(C)=\{A_1\}$};	
\node at (0.2, -2.4) (C) {\small  $\dep(D)=\{A_1,C\}$};	
\node at (0.2, -3.0) (D) {\small  $\dep(B)=\{A_1,C,D\}$};
\end{scope}
\node at (0,  -4) (invisible) {};
	\end{tikzpicture}
\end{minipage}
\hspace{-1.6cm}
  \begin{minipage}[b]{0.29\linewidth}
    \centering
    \begin{tikzpicture}[xscale=0.96, yscale=0.7]
      \node at (-1, -1.0) (A) {\small  $A_1$};
      \node at (-1, -2.0) (C) {\small $C$} edge[-] (A);
      \node at (-1, -3.0) (D) {\small $D$} edge[-] (C);
      \node at (-1, -4.0) (B) {\small $B$} edge[-] (D);
      \node at (-2.2, -5) (R) {\small  $R(A_1,B,C)$} edge[-] (B);
      \node at (0.2, -5) (S) {\small  $S(A_1,B,D)$} edge[-] (B);
      \node at (0,  -6) (invisible) {};
    \end{tikzpicture}
  \end{minipage}
  % \hspace{-0.5cm}
  \begin{minipage}[b]{0.29\linewidth}
    \centering
    \begin{tikzpicture}[xscale=0.96, yscale=0.7]
      \node at (-0.9, -1.0) (A) {\small $V_{A_1}(A_1)$};
      \node at (-0.9, -2.0) (D) {\small $V_C(A_1,C)$} edge[-] (A);
      \node at (-0.9, -3.0) (C) {\small $V_D(A_1,C,D)$} edge[-] (D);
      % \node at (-0.9, -3.0) (B') {\small $V'_B(A_1,C,D)$} edge[-] (D);
      \node at (-0.9, -4.0) (B) {\small $V_B(A_1,B,C,D)$} edge[-] (C);
      \node at (-2.2, -5) (R) {\small  $R(A_1,B,C)$} edge[-] (B);
      \node at (0.2, -5) (S) {\small  $S(A_1,B,D)$} edge[-] (B);
      \node at (0,  -6) (invisible) {};
    \end{tikzpicture}
  \end{minipage}
   \begin{minipage}[b]{0.24\linewidth}
    \centering
    \begin{tikzpicture}[xscale=0.96, yscale=0.7]
      \node at (-0.9, -1.0) (A) {\small $\delta V_{A_1}(a)$};
      \node at (-0.9, -2.0) (D) {\small $\delta V_C(a,c)$} edge[-] (A);
      \node at (-0.9, -3.0) (C) {\small $\delta V_D(a,c,D)$} edge[-] (D);
      \node at (-0.9, -4.0) (B) {\small $\delta V_B(a,b,c,D)$} edge[-] (C);
      \node at (-2, -5) (R) {\small  ${\delta R (a, b,c)}$} edge[-] (B);
      \node at (0, -5) (S) {\small  $S(A,B,D)$} edge[-] (B);
      \node at (0,  -6) (invisible) {};
    \end{tikzpicture}
  \end{minipage}
\caption{(Left) Access-top VO for 
$Q_1(B,D|A_1,C) = R(A_1,B,C), S(A_1,$ $B,D)$; (middle) the view tree corresponding to the VO;
(right) the delta view tree under a single-tuple update to $R$.}
\label{fig:general_CPAP_1}
\end{figure}

  Finally, we exemplify the construction of a view tree for a cyclic query.

  \begin{exa}\label{ex:triangle-cqap-complexity}
  Figure~\ref{fig:general_triangle-prelim} depicts a VO and the view tree constructed from it for the triangle CQAP $Q(B,C|A) = R(A,B), S(B,C), T(C,A)$ from Example~\ref{ex:CQAP-triangle}.
  The view $V_C$ joins the relations $R$ and $S$ and the indicator projection $I_{A,B}R$, which can be computed in $\bigO{N^\frac{3}{2}}$ time using a worst-case optimal join algorithm. The view $V_B$ can be computed in linear time by looking up each tuple from $V'_C$ in $R$. The views $V'_C$ and $V_A$ are constructed by marginalising out one variable at a time in time $\bigO{N^\frac{3}{2}}$ and $\bigO{N}$ time, respectively. Hence, the view tree construction takes $\bigO{N^\frac{3}{2}}$ time.
  \qed
  \end{exa}

The time to construct the view tree $\vt(\omega)$
is dominated by the time to materialise the view $V_X$ for each variable $X$. 
The auxiliary view $V_X'$ above $V_X$ can be materialised by marginalising out $X$ in 
one scan over $V_X$. Each view $V_X$ can be materialised in
$\bigO{N^{\fw}}$ time, where $\fw = \rho^*_{Q_X}(\{X\} \cup \dep_{\omega}(X))$.
The definition of the static width of $\omega$ implies  
that the view tree $\vt(\omega)$ can be constructed in  
$\bigO{N^{\fw(\omega)}}$ time, as stated in the next proposition.
By choosing an access-top VO $\omega$ for $Q_{\dagger}$ 
with $\fw(\omega) = \fw(Q)$, we obtain the preprocessing time 
from Theorem~\ref{thm:general}.
%By choosing an access-top  VO whose static width is $\fw(Q)$,
%the preprocessing time of our approach becomes  
%$\bigO{N^{\fw(Q)}}$, as stated in Theorem~\ref{thm:general}.  

\begin{prop}
  \label{prop:preprocessing}
  Given a VO $\omega$ of static width $\fw$ and a  database of size $N$,
  the view tree $\vt(\omega)$
  %described in Figure~\ref{fig:general_view_tree_construction} 
  can be constructed in 
    $\bigO{N^\fw}$ time.
\end{prop}

\begin{proof}
Consider  a CQAP $Q$, a VO $\omega$ for $Q$ with static width $\fw(\omega) =\fw$, and a database of size $N$.
%We show that the view tree $T = \vt(\omega)$   
%can be constructed in  
%$\bigO{N^{\fw}}$ time.
%
Without loss of generality, assume that $\omega$ consists of a single tree.  
Otherwise, we do the analysis below 
for each of the constantly many trees in $\omega$.
%The preprocessing stage consists of materialising the view tree 
%$T = \vt(\omega)$ where $\vt$ is the function given  
%in Figure~\ref{fig:general_view_tree_construction}. 
We show by induction on the structure of $T = \vt(\omega)$   that every node in $T$
can be materialised in $\bigO{N^{\fw}}$ time, where $\vt$ is the procedure given 
in Figure~\ref{fig:general_view_tree_construction}. 

\smallskip
%\noindent
\textit{Base Case}:
Each leaf atom or indicator projection in  $T$ can be materialised in 
$\bigO{N}$ time. Since $\fw \geq 1$, the complexity bound holds in the base case.  

\smallskip
\textit{Induction Step}:
Consider an auxiliary view  $V_X'(\calS')$ in $T$ with $X \in \vars(\omega)$
and $\calS' = \dep_{\omega}(X)$.   
By construction, this view results from its single child 
view $V_X(\calS \cup \{X\})$ by marginalising out variable $X$.
By induction hypothesis,  the view $V_X$ can be computed in $\bigO{N^{\fw}}$ time, 
hence its size has the same asymptotic bound. 
We can compute $V_X'$ by scanning over
the tuples in $V_X$ and maintaining during the scan the count 
$|\sigma_{\calS' = \inst{s}}V_X|$ for each tuple $\inst{s}$ in $\pi_{\calS'}V_X$.
This can be done in $\bigO{N^{\fw}}$ overall time.

Consider now a view  $V_X(\calS)$ in $T$ with 
$X \in \vars(\omega)$ and $\calS = \{X\}\cup \dep_{\omega}(X)$.   
Let $V_{X_1}(\calS_1), \ldots,  V_{X_k}(\calS_k)$ be the child nodes of $V_X(\calS)$.
Each child node can be a view, an atom, or an indicator projection.
By induction hypothesis, the child nodes of $V_X(\calS)$ can be materialised in 
$\bigO{N^{\fw}}$ time. 

Consider any variable $Y$ that occurs in the schemas of at least two child nodes 
of $V_X(\calS)$. It follows from the  construction of view trees that 
$Y \in \calS = \{X\} \cup \dep_{\omega}(X)$:
Consider two child views $V_{X_i}(\calS_i)$ and $V_{X_j}(\calS_j)$ of $V_X(\calS)$
such that $Y \in \calS_i \cap \calS_j$ and the variables $X_i$ and $X_j$ are children of $X$ in $\omega$.
The two views $V_{X_i}$ and $V_{X_j}$ are siblings, so they are constructed as in Line 6 of $\vt$, which means that 
$\calS_i = \dep_{\omega}(X_i)$ and $\calS_j = \dep_{\omega}(X_j)$.
Thus $Y \in \dep_{\omega}(X_i) \cap \dep_{\omega}(X_j)$.
Since $Y$ must be a common ancestor of $X_i$ and $X_j$ in $\omega$, 
%by the definition of VOs,
$Y$ is either $X$ or an ancestor of $X$ that is
in the dependency set of $X$.
Hence, $Y$ is in $\calS = \{X\} \cup \dep_{\omega}(X)$.

Hence, any variable that does not occur in $\calS$ cannot be a join variable 
for the child views of $V_X$.
We first marginalise out the variables in the child views that do not occur in 
$\calS$. This can be done in $\bigO{N^{\fw}}$ time. 
Let $V_1'(\calS_1'), \ldots ,V_k'(\calS_k')$ be the resulting views. 
The view $V_X(\calS)$ can now be written as  
$V_X(\calS) = V_1'(\calS_1'), \ldots ,V_k'(\calS_k')$,
where $\bigcup_{i = 1}^k \calS_i' = \calS$.
We use a worst-case optimal join algorithm to compute the view $V_X$.
The size of $V_X$ is upper-bounded by $\bigO{N^{p}}$ where $p = \rho^*_{Q_X}(\calS)$
and
$Q_X$ is the query that joins 
all atoms and indicator projections in $\omega_X$~\cite{Ngo:SIGREC:2013}.
By definition of $\fw$, the quantity $p$ is upper-bounded by $\fw$. This means that 
the view $V_X$ can be computed in $\bigO{N^{\fw}}$ time~\cite{Ngo:JACM:18}.

Overall, we conclude that the desired complexity bound holds for the induction step. 
\end{proof}

%% file: enumeration.tex
\section{Enumeration}
\label{sec:enumeration}
We describe the enumeration procedure of our approach for the dynamic evaluation of arbitrary CQAPs.
Consider a CQAP $Q(\calO|\calI)$, its fracture $Q_\dagger(\calO|\calI_{\dagger})$, 
and an access-top VO $\omega$ for $Q_{\dagger}$. Recall 
from Section~\ref{sec:preprocessing} that in the preprocessing stage, our approach uses the procedure 
$\vt$ in Figure~\ref{fig:general_view_tree_construction} to 
construct view trees $T_j$ 
following $\omega$
for the connected components $Q_j(\calO_j|\calI_j)$ of the fracture $Q_{\dagger}$, 
as explained in Section~\ref{sec:preprocessing}.
These view trees are  maintained under updates (Section~\ref{sec:updates}). 
Consider an input tuple $\inst{i}$ over $\calI$ for $Q$.
We enumerate the output tuples of each $Q_j(\calO_j|\calI_j)$ and concatenate them to obtain the output tuples of $Q(\calO|\calI)$.

We first describe the enumeration for a single connected component $Q_j(\calO_j|\calI_j)$.
%Consider an input tuple $\inst{i}$ over $\calI$ for $Q$.
We enumerate tuples over $\calO_j$ for the input tuple $\inst{i}_j = \inst{i}[\calI_j]$ over $\calI_j$ from the view tree $T_j$.
We traverse in preorder the views in $T_j$ that are constructed for the free variables of $Q_j$.
At each view $V_X(\calS)$, we do the following:
If $X$ is an input variable, we check if $\inst{i}_j[\calS]$ is in $V_X$; 
if $\inst{i}_j[\calS]$ is not in $V_X$, this means that $Q_j(\calO_j|\inst{i}_j)$ is empty, so we stop. 
Otherwise, we continue with the traversal.
If $X$ is an output variable, we retrieve from $V_X$ an $X$-value that is paired the values in $\inst{i}_j$ and the values retrieved from 
the views above $V_X$.
Once all these views are visited, we report the tuple consisting of the retrieved values.
Reporting each tuple takes constant time, since lookup and retrieval of values are constant-time operations as discussed in Section~\ref{sec:preliminaries}.

The tuples in $Q(\calO|\inst{i})$ are the Cartesian product of  the tuples in $Q_1(\calO_1|\inst{i}_1), \ldots, Q_n(\calO_n|\inst{i}_n)$.
We enumerate the tuples in $Q(\calO|\inst{i})$ by interleaving the enumeration procedures for $Q_1(\calO_1|\inst{i}_1),$ $\ldots,$ $Q_n(\calO_n|\inst{i}_n)$.
This gives us a constant-delay enumeration procedure for $Q$.
We demonstrate the enumeration procedure in the following example.

\begin{exa}
  Consider the query 
  $Q(B,C,D,E | A)$ from 
  Example~\ref{ex:general_preprocessing_CQAP0} and the 
  two connected components $Q_1(B,C,D|A_1)$ and  $Q_2(E|A_2)$
  of its fracture. Figure~\ref{fig:general_CPAP_0} (middle) depicts the view tree for $Q_1$.
  Given an  $A_1$-value $a$, we can use this view tree to enumerate 
  the distinct tuples in $Q_1(B,C,D|a)$ with constant delay. 
 We first check if $a$ is included in 
 the view $V_{A_1}$. If not, $Q_1(B,C,D|a)$ must be empty and we stop.
 Otherwise, we retrieve the first $B$-value $b$ paired with $a$ in $V_B$,
the first $C$-value $c$ paired with $(a,b)$ in $V_C$, and 
the first $D$-value $d$ paired with $(a,b)$ in $V_D$.
Thus, we obtain in constant time the first output tuple $(b,c,d)$ in 
$Q_1(B,C,D|a)$ and report it. 
Then, we iterate over the remaining distinct $D$-values paired with
$(a,b)$ in $V_D$ and report for each
 such $D$-value $d'$, a new tuple $(b,c,d')$. 
 After  all $D$-values are exhausted,
 we retrieve the next distinct $C$-value paired with $(a,b)$ in $V_C$ 
 and restart the iteration over the distinct $D$-values paired with $(a,b)$ in $V_D$, and so on.
 Overall, we construct each distinct tuple in $Q_1(B,C,D| a)$ in constant time
 after the previous one is constructed.

  Assume now that we have constant-delay enumeration procedures 
  for the tuples in $Q_1(B,C,D|a)$ and  the tuples in $Q_2(E|a)$ 
  for any $A$-value $a$. We can enumerate with constant delay the tuples
  in  $Q(B,C,D,E | a)$ as follows. We ask for the first tuple $(b,c,d)$ in
$Q_1(B,C,D|a)$ and then iterate over the distinct $E$-values in  $Q_2(E|a)$.
For each such $E$-value $e$, we report the tuple $(b,c,d,e)$. Then, we ask for the 
next tuple in $Q_1(B,C,D|a)$ and restart the enumeration over the tuples in 
 $Q_2(E|a)$, and so on. 
 \qed
\end{exa}

The following proposition states that our approach achieves constant-delay enumeration of the 
tuples in the query output.  This matches the  enumeration delay stated in 
Theorem~\ref{thm:general}. 

\begin{prop}
\label{prop:enumeration}
Let  $Q(\calO | \calI)$
be a CQAP and 
$\omega$ an 
access-top VO $\omega$ for the fracture of $Q$ 
that consists of the trees $(\omega_j)_{j \in [n]}$.
Given any tuple $\inst{i}$ over $\calI$,  
the tuples in $Q(\calO | \inst{i})$ can be enumerated 
from the view trees $(\vt(\omega_j))_{j \in [n]}$ with constant delay. 
\end{prop}

\begin{proof}
Consider a CQAP $Q(\calO | \calI)$, its fracture  $Q_{\dagger}(\calO | \calI_{\dagger})$, and an 
access-top VO $\omega$ for $Q_\dagger$.
Assume that $\omega$ consists of the trees $\omega_1, \ldots , \omega_n$ and let 
$T_1 = \vt(\omega_1), \ldots , T_n = \vt(\omega_n)$ be the view trees 
constructed by the procedure $\vt$ in Figure~\ref{fig:general_view_tree_construction}.
We show that for any input tuple $\inst{i}$ over $\calI$, the tuples in   
$Q(\calO|\inst{i})$ can be enumerated with constant delay
using $T_1, \ldots ,T_n$.    

For $j \in [n]$, let 
$Q_j(\calO_j | \calI_j)$ with 
$\calO_j = \calO\cap \vars(\omega_j)$ and
$\calI_j = \calI_{\dagger}\cap \vars(\omega_j)$ 
be the CQAP that joins the atoms appearing at the leaves 
of $T_j$.
We first explain how for any $j \in [n]$ and $\inst{i}_j$ over $\calI_j$, the tuples 
in $Q_j(\calO_j | \inst{i}_j)$ can be enumerated with constant delay using the view tree $T_j$.
Since the view tree is constructed following an access-top
variable order, 
there is no view $V_Y$ with $Y$ being bound (output)
that is above a view $V_X$ with $X$ being free (input).
To construct the first output tuple in $Q_j(\calO_j | \inst{i}_j)$, we traverse $T_j$ in preorder and do the  following  
at each view $V_X$, where $X$ is free. 
If $X \in \calI_j$, i.e., it is an input variable, we check 
if the projection of $\inst{i}_j$ onto the schema of $V_X$ is included in $V_X$.
If not,  $Q_i(\calO_j | \inst{i}_j)$ is empty and we stop the traversal. 
Otherwise, we continue with the traversal.  
When we arrive at a view $V_X$ with $X \in \calO_j$,
we have already fixed a tuple $\inst{t}$ over the variables in the root
path of $X$. We retrieve in constant time a first X-value  in 
$\sigma_{\calS = \inst{t}'} V_{X}$, where $\calS$ is the schema 
of $V_X$ without $X$ and $\inst{t}' = \inst{t}[\calS]$.
After all views $V_X$ with free $X$ are visited, we have fixed all values over the variables in $\calO_i$, hence we report 
the tuple consisting of these values. 
Then, we iterate over the remaining distinct $Y$-values in the last visited view $V_Y$ with constant delay 
(given that the values over the root path of $Y$ are fixed).
For each distinct $Y$-value, we obtain a new tuple that we report. After all $Y$-values are exhausted,
 we backtrack.  
  
Assume that we can enumerate the tuples in   
$Q_j(\calO_j | \inst{i}_j)$ with constant delay for any $j \in [n]$
and tuple $\inst{i}_j$ over $\calI_j$.  
Consider a tuple $\inst{i}$ over $\calI$.
It holds   
$Q(\calO | \inst{i}) = \times_{j \in [n]} Q_j(\calO_j| \inst{i}_j)$
where $\inst{i}_j[X'] = \inst{i}[X]$ if $X=X'$ or $X$ is replaced by $X'$ when constructing the fracture of 
$Q$. We enumerate the tuples in $Q(\calO| \inst{i})$
by interleaving the enumeration procedures 
for $Q_1(\calO_1 | \inst{i}_1), \ldots, Q_n(\calO_n | \inst{i}_n)$, as follows.

\begin{center}
\renewcommand{\arraystretch}{1.15}
\renewcommand{\linenumber}{\makebox[2ex][r]{\rownumber\TAB}}
\setcounter{magicrownumbers}{0}
\begin{tabular}{@{\hskip 0.1in}l}
\midrule
\linenumber \FOREACH $\inst{o}_1 \in Q_1(\calO_1| \inst{i}_1)$ \\[0.2ex]
\linenumber \TAB ${\cdot}{\cdot}{\cdot}$ \\[0.2ex] 
\linenumber \TAB \TAB \FOREACH $\inst{o}_n \in Q_n(\calO_n| \inst{i}_n)$ \\[0.2ex] 
\linenumber \TAB\TAB\TAB \textbf{report} $\inst{o}_1{\cdot}{\cdot}{\cdot} \inst{o}_n$ \\[0.2ex]
\bottomrule
\end{tabular}
\end{center}
 
That is, we first retrieve the first complete tuple $\inst{o}_j$ from
 $Q_j(\calO_j | \inst{i}_j)$ for each  $j \in [n]$ and report 
 $\inst{o}_1 \cdots  \inst{o}_n$. Then, we iterate  over the remaining tuples 
 in  $Q_n(\calO_n | \inst{i}_n)$. For each such tuple $\inst{o}_n'$,
 we report $\inst{o}_1 \cdots  \inst{o}_n'$. After all tuples in 
 $Q_n(\calO_n | \inst{i}_n)$ are exhausted, we move to the next tuple
 in $Q_{n-1}(\calO_{n-1} | \inst{i}_{n-1})$ and restart the enumeration 
 for $Q_n(\calO_n | \inst{i}_n)$, and so on. 
 
We conclude that the time to report the first tuple 
in  $Q(\calO | \inst{i})$, the time to report a next tuple after the previous one is reported, 
and the time to signal the end of the enumeration after the last tuple is reported is constant. 
\end{proof}

%% file: updates.tex
\section{Updates}
\label{sec:updates}
In this section, we explain how our approach maintains the view trees constructed in the preprocessing stage 
 under single-tuple updates to the base relations. 
% describe the update procedure of our  dynamic evaluation approach for arbitrary CQAPs..
%Consider in the following a CQAP $Q(\calO|\calI)$, its fracture $Q_\dagger(\calO|\calI_{\dagger})$, and a database of size $N$.

Consider a CQAP $Q$, an access-top VO $\omega$ for the fracture $Q_{\dagger}$ and the view trees
$T_1, \ldots , T_n$ constructed from $\omega$ by the procedure $\vt$ in Figure~\ref{fig:general_view_tree_construction}. 
Let $\delta R = \{\inst{x} \rightarrow m\}$ be a single-tuple update 
%by the function $\tau$ in 
%Figure~\ref{fig:general_view_tree_construction} 
to an input relation $R$; $m$ is positive in case of insertion and negative in case of deletion. 
We first update each view tree $T_j$ that has an atom $R(\calX)$ at a leaf: We update each view on the path from that leaf to the root of the view tree using the standard delta rules~\cite{Chirkova:Views:2012:FTD}.
The update $\delta R$ may also trigger single-tuple updates to indicator projections $I_{\calZ}R$, as discussed in Section~\ref{sec:preliminaries}.
These updates to indicators are propagated up to the root of each view tree, like for $\delta R$.

\begin{exa}
Figure \ref{fig:general_CPAP_0} (right) shows the delta view tree for the 
view tree to the left  
under a single-tuple update $\delta R(a,b,c)$ to $R$. 
\nop{The delta view tree for an update to $S$ is analogous.} 
We update the relation $R(A,B,C)$ with $\delta R(a,b,c)$ in constant time.
The ancestor views of $\delta R$ are the deltas of the corresponding views, computed 
by propagating $\delta R$ from the leaf to the root. They can also be effected in constant time.
\nop{For instance, we update 
the view $V_C'(A_1,B)$ with $\delta V_C'(a,b)= \delta V_C(a,b,c)$ and
the view $V_B(A_1,B)$ with $\delta V_B(a,b)= \delta V_C'(a,b), V_D'(a,b)$.}
Overall, maintaining the view tree under a single-tuple update to any relation takes $O(1)$ time.
 
%%%%%%%%%%%%%%%%%%%%%%%%%%%%%%%%%% 
Consider now the delta view tree in Figure \ref{fig:general_CPAP_1} (right) obtained from the view tree
to its left under the single-tuple update $\delta R(a,b,c)$.  
We update $V_B(A_1,B,C,D)$ with $\delta V_B(a,b,c,D)= \delta R(a,b,c),S(a,b,D)$ 
in $\bigO{N}$ time, since there are at most $N$ $D$-values paired with $(a,b)$ in $S$. We then
update the views $V_D$, $V_C$, and $V_{A_1}$ in $\bigO{1}$ time.
Updates to $S$ are handled analogously. 
Overall, maintaining the view tree under a single-tuple update to any relation takes $O(N)$ time.
\qed
\end{exa}
%%%%%%%%%%%%%%%%%

The following proposition states the time to maintain the view trees under single-tuple updates to the base relations.
This matches the update time in Theorem~\ref{thm:general}.

\begin{prop}
  \label{prop:updates}
Given a VO $\omega$ consisting of the trees $(\omega_j)_{j \in [n]}$
and a database of size $N$, the view trees $(\vt(\omega_j))_{j \in [n]}$
can be maintained under single-tuple updates 
to the base relations with 
    $\bigO{N^{\dfw(\omega)}}$ update time.
\end{prop}

\begin{proof}
Consider a VO $\omega$ that consists of the trees $(\omega_j)_{j \in [n]}$ and 
a database of size $N$. Let 
$(T_j = \vt(\omega_j))_{j \in [n]}$ be the view trees 
constructed by the procedure $\vt$ in Figure~\ref{fig:general_view_tree_construction}.
We show that the view trees can be maintained 
with $\bigO{N^{\dfw(\omega)}}$ update time under single-tuple updates to the base relations.

Consider a single-tuple update 
to a base relation $R$.
We first update  each view tree $T_j$ referring to an atom of the form 
$R(\calX)$. Updating $T_j$ 
amounts to computing the deltas of the views on the path from 
 $R(\calX)$ to the root of the view tree.
 We have shown in the proof of Proposition~\ref{prop:preprocessing}
  that for each variable 
 $X$ in $\omega$, the views $V_X$ and $V_X'$ can be materialised in 
 $\bigO{N^p}$ time, where $p = \rho_{Q_X}^*(\{X\} \cup \dep_{\omega}(X))$.
 Since the update fixes the values in $\calX$, 
 the time to compute the delta of these views under 
 the update becomes 
 $\bigO{N^d}$,
where 
$d = \rho_{Q_X}^*((\{X\} \cup \dep_{\omega}(X)) \setminus \calX)$.
A single-tuple update to $R$ can trigger a single-tuple update 
to each indicator view of the form $I_{\calZ}(R(\calZ))$.
Following a similar argument as above, we conclude 
that the time to compute the deltas of the views under such updates 
is $\bigO{N^{d'}}$,
where 
$d' = \rho_{Q_X}^*((\{X\} \cup \dep_{\omega}(X)) \setminus \calZ)$.
It follows from the definition of the dynamic width of VOs that
the exponents $d$ and $d'$ are upper-bounded by $\dfw(\omega)$. 
This implies that the overall update time is 
$\bigO{N^{\dfw(\omega)}}$.
\end{proof}

%% file: discussion.tex
\section{Discussion of our Approach}
\label{sec:discussion}
Sections \ref{sec:preprocessing}-\ref{sec:updates} explain our approach to evaluating arbitrary CQAPs. We next discuss key decisions behind our approach.

\paragraph{1. Variable orders} Our approach can be rephrased to use hypertree decompositions~\cite{Gottlob99} instead of VOs, since they are different syntaxes for the same query decomposition class~\cite{OlteanuZ15}. Indeed, the set consisting of a variable and its dependency set in a VO can be interpreted as a bag of a hypertree decomposition whose edges between bags reflect those between the variables in the VO. Variable orders are more natural for our algorithms for constructing view trees and for enumeration, as well as worst-case optimal join algorithms such as the LeapFrog TrieJoin~\cite{LFTJ:ICDT:2014} and their use for constructing factorised representations of query results~\cite{OlteanuZ15}: These algorithms  proceed one variable at a time and not one bag of variables at a time. VO-based algorithms express more naturally computation by variable elimination.

\paragraph{2. Access-top VOs} Access-top VOs can have higher static and dynamic widths than arbitrary VOs. However, they are needed to attain the constant-delay enumeration in Theorem~\ref{thm:general}, as explained next. The maintenance procedure for view trees ensures that each view is calibrated\footnote{A relation $R$ is calibrated with respect to other relations in a query $Q$ if each tuple in $R$ participates to at least one tuple in the output of $Q$.} with respect to all of its descendant views and relations, since the updates are propagated bottom-up from the relations to the top view. Since the views constructed for the input variables are above all other views in a view tree constructed from an access-top VO, these views are calibrated. For a given tuple of values over the input variables, the calibration of these views guarantees that if they do not agree with this tuple, then there is no output tuple associated with the input tuple. For constant-delay enumeration, we follow a top-down traversal of the view tree and use the constant-time lookup of the hash maps implementing the views. 
Furthermore, since the output variables are above the bound variables in the VO, tuples of values over the output variables can be retrieved from views whose schemas do not contain bound variables. Hence, we can enumerate the \emph{distinct} tuples over the output variables for a given tuple over the input variables.

In case we would have used an arbitrary (and not access-top) VO, then the input variables may be anywhere in the VO; in particular, there may be views above the relations with the input variables that do not have input variables. On an enumeration request, the values given to the input variables act as selection conditions on the relations and may require the calibration of the views on top before the enumeration starts; this calibration may be as expensive as computing the query. Otherwise, we incur a non-constant cost for the enumeration of each output tuple. Either way, the enumeration delay may not be constant.

\paragraph{3. Lazy approach using residual queries} A simple CQAP evaluation approach is the lazy approach. On  updates, the lazy approach just updates the input relations. On enumeration, where each input variable is given a value, it computes the residual query obtained by setting the input variables to the given values. The enumeration of the tuples in the output of a residual query cannot guarantee constant delay, since the parts of the input relations, which satisfy the selection conditions on the input variables, are not necessarily calibrated, and the calibration may take as much time as computing the residual query.

\paragraph{4. Replacing each occurrence of an input variable by a fresh variable} Although this query rewriting removes the joins on the input variables, it does not affect the correctness of query evaluation. For enumeration, all fresh variables are fixed to given values. In access-top VOs, these variables are above the other variables and are in views that are calibrated with respect to the relations in their respective connected component of the rewritten query. We can then check whether all view trees satisfy the assignment of values to the input values. If a view tree fails, then the query output is empty for the  values given to the input variables.

\paragraph{5. Query fractures} The query rewriting in the previous discussion point is only the first step of query fracturing. The second step merges all fresh variables for an input variable into one variable in case they are in the same connected component. This does not affect correctness but may affect the complexity, as exemplified next.
Consider the triangle query in Example~\ref{ex:triangle-cqap-complexity}: $Q(B,C|A) = R(A,B), S(B,C), T(C,A)$. If we were to replace $A$ by two fresh variables $A_1$ and $A_2$, then the rewritten query would be: $Q'(B,C|A_1,A_2) = R(A_1,B), S(B,C), T(C,A_2)$. It still has one connected component. An access-top VO for $Q'$ is $A_1-A_2-B-C$ ($A_1$ and $A_2$ may be swapped, same for $B$ and $C$). The static width of $Q'$ is 2. Yet by merging back $A_1$ and $A_2$, we obtain $Q$, which admits the access-top VO $A-B-C$ and static width $3/2$ (same width can be obtained if $B$ and $C$ are swapped), as in Example~\ref{ex:triangle-cqap-complexity}.

%% file: dichotomy.tex
\section{A Dichotomy for CQAPs}
\label{sec:dichotomy}
In this section, we prove our dichotomy result in Theorem~\ref{thm:dichotomy}, which states that  
the queries in the class CQAP$_0$ are precisely those queries that can be evaluated with constant update time and enumeration delay:

\medskip
\noindent
\textbf{Theorem~\ref{thm:dichotomy}.}
\textit{
	Let any CQAP $Q$ and database of size $N$.
	\begin{itemize}
	\item If $Q$ is in $\text{CQAP}_0$, then it admits  
	$\bigO{N}$ preprocessing time,
	$\bigO{1}$ enumeration delay, and $\bigO{1}$  
	update time for single-tuple updates.
	\item If $Q$ is not in $\text{CQAP}_0$ and has no repeating relation symbols,
		then there is no  algorithm that computes $Q$ with arbitrary preprocessing  time, 
		$\bigO{N^{\frac{1}{2} - \gamma}}$ enumeration delay, and 
		$\bigO{N^{\frac{1}{2} - \gamma}}$ amortised update time,
		for any 
			$\gamma >0$, unless the \OMv conjecture fails. 
	\end{itemize}
}

\medskip
The \OMv conjecture was introduced in Section~\ref{sec:results} (Conjecture~\ref{conj:omv}).
Before proving Theorem~\ref{thm:dichotomy}, we  
introduce an auxiliary lemma and a proposition. 

The next lemma states that the evaluation complexity of the fracture $Q_{\dagger}$ of 
a CQAP $Q$ is upper-bounded by the evaluation complexity of $Q$.  

\begin{lem}
\label{lem:fracture_equivalence}
If a CQAP $Q$ can be evaluated with 
$\bigO{f_p(N)}$ preprocessing time,
$\bigO{f_e(N)}$ enumeration delay, and $\bigO{f_u(N)}$ amortised update time
for databases of size $N$ and some functions $f_p$, $f_e$, and $f_u$, 
then the fracture $Q_{\dagger}$ can be evaluated with the same asymptotic complexities. 
\end{lem}

\begin{proof}
Consider a CQAP $Q(\calO|\calI)$,
its fracture $Q_{\dagger}(\calO|\calI_{\dagger})$, and  
a database $\calD$  for $Q_{\dagger}$ of size $N$.
We call a fresh variable $A$ in $Q_{\dagger}$ that replaces 
a variable $A'$ in $Q$ a {\em representative} of $A'$. 
Let $Q_1, \ldots, Q_n$ be the connected components of 
$Q_{\dagger}$ and  $C_1, \ldots, C_n$ sets of database relations such that 
each $C_i$ consists of the relations that are referenced in $Q_i$.
We create from $\calD$ the databases $\calD_1, \ldots , \calD_n$,
where each $\calD_i$ is constructed as follows.
The database $\calD_i$ contains each relation $R$ in $\calD$,
modified as follows: (1) If  $R \in C_i$ and $R$ has a variable $A$ in its schema
that is a representative of a variable $A'$, then the variable $A$ is replaced 
by $A'$; (2) the values in all relations not contained in
$C_i$ are replaced  
by a single dummy value $d_i$.
The overall size of the databases $\calD_1, \ldots , \calD_n$ is $\bigO{N}$.
Given an input tuple $\inst{t}$ over $\calI$, we denote by 
$(Q(\calO|\inst{t}),\calD_i)$ the result of $Q$ for input $\inst{t}$
over  $\calD_i$. The result consists of the tuples over the 
output variables in $C_i$ for the given input tuple $\inst{t}$, paired with the dummy value $d_i$ over the output variables not in $C_i$.
Intuitively, the result of $Q_{\dagger}$
on $\calD$ can be obtained from the Cartesian product 
 of the results of $Q$ on $\calD_1, \ldots, \calD_n$.
 To be more precise, consider a tuple $\inst{t}_{\dagger}$
 over $\calI_{\dagger}$. 
 We define for each $i \in [n]$, a tuple $\inst{t}_i$ over $\calI$ such that
$\inst{t}_{i}[A] = \inst{t}_{\dagger}[A']$ if  $A'$ is a representative of  $A$.
The result of $Q_{\dagger}(\calO|\inst{t}_\dagger)$ on $\calD$ 
is equal to the Cartesian product 
$\times_{i \in [n]}\pi_{\calO_i}(Q(\calO|\inst{i}),\calD_i)$,
where $\calO_i$ is the set of output variables of $Q$ contained 
in $C_i$. 
Now, assume that we want to enumerate the tuples in 
$(Q_{\dagger}(\calO|\inst{t}_\dagger), \calD)$. 
 We start the enumeration procedure for each
$Q(\calO|\inst{i}),\calD_i)$ with $i \in [n]$.
For each 
$\inst{t}_1'\in Q(\calO|\inst{t}_1),\calD_1),$ $\ldots,$
$\inst{t}_n'\in Q(\calO|\inst{t}_n),\calD_n)$, we return 
the tuple $\pi_{\calO_1} \inst{t}_1'$ $\circ \ldots \circ$
$\pi_{\calO_n} \inst{t}_n'$. Hence, 
the tuples in $(Q_{\dagger}(\calO|\inst{t}_\dagger), \calD)$ can be enumerated 
with $\bigO{f_e(N)}$ delay if  $Q$ admits 
$\bigO{f_e(N)}$ enumeration delay.
%, we can use the enumeration procedure for $Q$
%to obtain the result tuples for $Q_{\dagger}$ with the same asymptotic delay.
%
We execute the preprocessing procedure 
for $Q$ on each of the databases $\calD_1, \ldots, \calD_n$,
which takes $\bigO{f_p(N)}$ overall time. 
Consider an update $\{\inst{t} \mapsto m \}$ to a relation 
$R$ that is contained in the connected component $C_i$ with
$i \in [n]$. We apply the update $\{\inst{t}_{\calI} \mapsto m \}$
to relation $R$ in $\calD_i$, where 
$\inst{t}_{\calI}$ is the tuple over $\calI$ defined as: 
$$\inst{t}_{\calI}[A] = \left\{
\begin{array}{ll}
\inst{t}[A'] & \text{ if } A' \text{ is a representative of } A \\
\inst{t}[A] & \, \textrm{otherwise} \\
\end{array}\right. $$ 
The update takes $\bigO{f_u(N)}$ amortised update time.

Overall, we obtain an evaluation procedure 
for $Q_{\dagger}$ with 
$\bigO{f_p(N)}$ preprocessing time,
$\bigO{f_e(N)}$ enumeration delay, and $\bigO{f_u(N)}$ amortised update time.
\end{proof}    

The next proposition is essential for the complexity upper bound in 
Theorem~\ref{thm:dichotomy}.

\begin{prop}
\label{prop:cqap0_delta0}
Every query in CQAP$_0$ has dynamic width $0$ and static width $1$.
\end{prop}
    
\begin{proof}     
Consider a query $Q$ in CQAP$_0$ and its 
fracture $Q_{\dagger}$. We first show that the dynamic width 
of $Q$ is $0$.  By definition, $Q_{\dagger}$
is hierarchical, free-dominant, and input-dominant. 
Hierarchical queries admit canonical VOs. 
In canonical VOs, it holds:
If a variable $A$ dominates a variable $B$,
then $A$ is on top of $B$.  
Hence, $Q_{\dagger}$ admits a canonical VO 
that is access-top. 
%We show that $\omega$ has dynamic width $0$.
Consider a variable $X$ in $\omega$ and 
an atom $R(\calY)$ in the subtree 
$\omega_X$ rooted at $X$. By the definition 
of canonical VOs, it holds: the dependency set of $X$ consists of the ancestor variables 
of $X$; and $\calY$ contains $X$ and all ancestor variables of $X$.
Hence, we have 
$\rho_{Q_X}^*((\{X\} \cup \dep_{\omega}(X))\setminus\calY) = 
\rho_{Q_X}^*((\{X\} \cup \anc_{\omega}(X))\setminus\calY)=
\rho_{Q_X}^*(\emptyset)=0$. This implies that the 
dynamic width of $\omega$ is $0$. This means 
that the dynamic width of $Q_{\dagger}$, hence,
the dynamic width of $Q$ is $0$.

It follows from Proposition \ref{prop:width_delta_inequal}
that the static width of $Q$ 
is $1$\footnote{To simplify the  presentation, we assume that $Q$ contains at least one variable, so it has 
the static width at least $1$. Otherwise, it can be trivially evaluated with constant preprocessing time, update time, 
and enumeration delay.}. 
\end{proof}     

In the following, we first prove the complexity upper bound and then the complexity lower bound 
stated in Theorem~\ref{thm:dichotomy}.

\subsection{Complexity Upper Bound}	
We prove the first statement in Theorem~\ref{thm:dichotomy}.
Assume that $Q$ is in  $\text{CQAP}_0$.
 By Proposition~\ref{prop:cqap0_delta0},
$Q$'s dynamic width is $0$.
By the definition of CQAP$_0$, the fracture $Q_{\dagger}$ must be hierarchical.
From Proposition~\ref{prop:width_delta_inequal}, 
the static width of $Q_{\dagger}$, hence the static width of 
$Q$, is at most $1$.
%This means that the static width of $Q$ is at most $1$ and its dynamic width is $0$.
%\footnote{In case $Q$ does not contain any atom with non-empty 
%schema, its static width is $0$. We ignore here this trivial case.}.
%By Definition \ref{def:fac_width}, the fracture $Q'$ of $Q$ has static width $1$ and 
%dynamic width $0$. 
It follows from Theorem~\ref{thm:general} that $Q$ 
can be evaluated with  
$\bigO{N}$ preprocessing time, 
$\bigO{1}$ update time, and
$\bigO{1}$ enumeration delay. 

%%%%%%%%%%%%%%%%
%Lower bound
%%%%%%%%%%%%%%%%
\subsection{Complexity Lower Bound}
We prove the second statement in Theorem~\ref{thm:dichotomy}.
The proof is based on a  reduction of the 
Online Matrix-Vector Multiplication (\OMv)
problem (Definition~\ref{def:OMv}) to the 
evaluation of  CQPAs that are not in $\text{CQAP}_0$.

We start with the high-level proof idea. 
Consider the following simple CQAPs, which are 
not in $\text{CQAP}_0$. 
\begin{align*}
Q_1(\calO |\cdot)= &R(A), S(A,B), T(B)  \TAB  \calO \subseteq \{A,B\}  \\ 
Q_2(A |\cdot)=& R(A,B),S( B)  \\
Q_3( \cdot | A )= & R(A,B),S(B) \\
Q_4(B |A)=& R(A,B),S(B)  
\end{align*}

Each query is equal to its
fracture.
Query $Q_1$ is not hierarchical;
$Q_2$ is not free-dominant; 
$Q_3$ and $Q_4$ are not input-dominant.
It is known that queries that are not hierarchical or 
free-dominant do not admit 
constant update time and enumeration delay,
unless the \OMv conjecture 
fails \cite{BerkholzKS17}. 
We show that the $\OMv$ problem can also be reduced 
to the evaluation of each of the queries $Q_3$ and $Q_4$.
Our reduction implies that any algorithm 
that evaluates $Q_3$ or $Q_4$ 
with arbitrary preprocessing time, 
$\bigO{N^{\frac{1}{2}-\gamma}}$
amortised update time, and 
$\bigO{N^{\frac{1}{2}-\gamma}}$
enumeration delay, for any $\gamma >0$, can be used to solve the 
\OMv problem in subcubic time, which 
rejects the \OMv conjecture. 
We then show that, given any CQAP $Q$
that is not in 
$\text{CQAP}_0$ and does not have repeating relation symbols, we can   
reduce the evaluation of one of the queries $Q_1$, $Q_2$, $Q_3$ and $Q_4$
to the evaluation of $Q$.

%%%%%%%%%%%%%%%%%%%%%
%proof starts
%%%%%%%%%%%%%%%%%%%%%%
In each of the following two reductions, 
our starting assumption is that there is 
an algorithm $\calA$ that evaluates the given query 
with arbitrary 
preprocessing  time, $\bigO{N^{\frac{1}{2} - \gamma}}$ amortised update time, 
and $\bigO{N^{\frac{1}{2} - \gamma}}$ enumeration delay
for some $\gamma >0$. We then show that $\calA$ can be used 
to design an algorithm $\calB$ 
that solves the \OMv problem
in subcubic time.

\paragraph*{Hardness for $Q_3$}
Given $n \geq 1$, let $\inst{M}$, $\inst{v}_1,$ $\ldots,$ $\inst{v}_n$
be an input to the \OMv problem, where $\inst{M}$ is an $n \times n$ Boolean
Matrix and $\inst{v}_1, \ldots, \inst{v}_n$ are Boolean 
column vectors  of size $n$. Algorithm $\calB$ uses relation  
$R$ to encode matrix $\inst{M}$ and relation $S$ to encode the 
incoming vectors $\inst{v}_1, \ldots, \inst{v}_n$. 
The database domain is $[n]$. 
Algorithm $\calB$ first executes the preprocessing stage
 on the empty database. Since the database is empty, 
 the preprocessing stage must end after constant time. 
 Then, it executes 
 at most $n^2$ updates to relation $R$ such that $R(i,j) = 1$
 if and only if $\inst{M}(i,j) = 1$. Afterwards, it performs a round
 of operations for each incoming vector $\inst{v}_r$
 with $r \in [n]$.
 In the first part of each round, it executes at most $n$ updates
 to relation $S$ such that $S(j) = 1$ if and only if
 $\inst{v}_r(j) = 1$. 
 Observe that $Q_3(\cdot| i)$ is true for some $i\in [n]$ 
 if and only if $(\inst{M} \inst{v}_r)(i) =1$.
Algorithm $\calB$ constructs 
the result vector $\inst{u}_r = \inst{M} \inst{v}_r$ as follows. 
It asks for each 
$i \in [n]$, whether  
 $Q_3(\cdot| i)$ is true, i.e.,
 $i$ is in the result of $Q_3$.
 If yes, the $i$-th entry of the result of $\inst{u}_r$
 is set to $1$, otherwise, it is set to $0$.

 \smallskip
 \textit{Time Analysis.}
The size of the database remains $\bigO{n^2}$ during the whole procedure. Algorithm $\calB$
needs at most $n^2$ updates to encode $\inst{M}$ by relation 
$R$. Each update can be processed in $\bigO{(n^2)^{\frac{1}{2}-\gamma}}$
amortised update time. 
Hence, the overall time to execute these updates is 
$\bigO{n^2 (n^2)^{\frac{1}{2}-\gamma}} = \bigO{n^{3 - 2\gamma}}$.
In each round $r$ with $r \in [n]$, algorithm $\calB$ executes 
$n$ updates to encode vector $\inst{v}_r$ into relation $S$ and
asks for the result of $Q_3(\cdot | i)$ for every $i \in [n]$.  
The $n$ updates and requests 
need  
$\bigO{n (n^2)^{\frac{1}{2}-\gamma}} = \bigO{n^{2 - 2\gamma}}$
time.
Hence, the overall time for a single round is $\bigO{n^{2 - 2\gamma}}$.
Consequently,
the time for $n$ rounds is 
$\bigO{nn^{2 - 2\gamma}}= \bigO{n^{3 - 2\gamma}}$.
This means that the overall time of the reduction is 
 $\bigO{n^{3 - 2\gamma}}$ in worst-case, which is subcubic.

\paragraph*{Hardness for $Q_4$}
The reduction differs slightly from the case for $Q_3$ in the way 
algorithm $\calB$ constructs 
the result vector  $\inst{u}_r = \inst{M} \inst{v}_r$ in each round $r$.
For each 
$i \in [n]$, it starts the enumeration process for 
 $Q_4(B|i)$. If one tuple is returned, it stops
 the enumeration process and sets the $i$-th  
 entry of $\inst{u}_r$ to be $1$. 
 If no tuple is returned, 
 the $i$-th  entry is set to $0$. 
 Thus, the time to decide the $i$-th entry 
 of the result of $\inst{u}_r$
 is the same as in case of $Q_3$.
Hence, the overall time of the reduction stays 
 subcubic.

\paragraph*{Hardness in the General Case}
Consider now an arbitrary CQAP $Q$ 
that is not in CQAP$_0$ and does not have 
repeating relation symbols. 
%Let $Q'$ be the fracture of $Q$.
Since $Q$ is not in 
CQAP$_0$, this means
that its fracture $Q_{\dagger}$ is either not hierarchical, not free-dominant, or 
not input-dominant. 
If $Q_\dagger$ is not hierarchical or it is not free-dominant and all free variables are output, it 
follows from prior work that there is no algorithm that evaluates $Q_\dagger$ with 
 $\bigO{N^{\frac{1}{2} - \gamma}}$ enumeration delay, and
  $\bigO{N^{\frac{1}{2} - \gamma}}$ amortised update time for any $\gamma > 0$,
unless the \OMv conjecture fails \cite{BerkholzKS17}. 
By Lemma~\ref{lem:fracture_equivalence}, 
no such algorithm can exist for $Q$.
Hence, we assume that $Q_\dagger$ is hierarchical and consider two cases:  
\begin{enumerate}
\item $Q_\dagger$ is not free-dominant and all free variables are input 
\item  $Q_\dagger$ is free-dominant but not input-dominant
\end{enumerate}

 \smallskip
 \textit{Case (1).}
The query must contain 
an input  variable $A$ and a bound 
variable $B$
such that $\atoms(A) \subset \atoms(B)$.
This mean that there are
two atoms $R(\calX)$ and $S(\calY)$ 
with $\calY \cap \{A,B\} = \{B\}$
and $A, B \in \calX$.   
Assume that there is an algorithm $\calA$ that evaluates 
 $Q_{\dagger}$ with arbitrary preprocessing time, 
 $\bigO{N^{\frac{1}{2} - \gamma}}$ enumeration delay, and
  $\bigO{N^{\frac{1}{2} - \gamma}}$ amortised update time, for some $\gamma > 0$.
  We will design an algorithm $\calB$ that evaluates 
  $Q_3$ with the same complexities. 
  This rejects the \OMv conjecture.
  Hence, by Lemma~\ref{lem:fracture_equivalence},
  $Q$ cannot be evaluated with these complexities, unless 
  the \OMv conjecture fails. 
  
We define $\calR_{(A,B)}$ to be the set of atoms that 
contain both $A$ and $B$ in their schemas
and $\calS_{(\neg A,B)}$ to be the set of atoms 
that contain $B$ but not $A$.
Note that there cannot be any atom containing $A$
but not $B$, since this would imply that the query is not hierarchical, 
contradicting our assumption. 
We use  each atom $R'(\calX') \in \calR_{(A,B)}$
to encode atom $R(A,B)$ and
each atom $S'(\calY') \in \calS_{(\neg A,B)}$
to  encode atom $S(B)$ in $Q_3$.
Consider a database $\calD$ of size $N$ for $Q_3$
and a dummy value $d$ that is not included in 
the domain of $\calD$. 
We write $(\calS, A= a, B= b, d)$ to denote a tuple
over schema $\calS$ that 
 assigns the values $a$ and $b$ to the variables $A$ and respectively $B$
 and all other variables in $\calS$ to $d$.
Likewise, $(\calS, B= b, d)$ denotes a tuple
that assigns value $b$ to $B$ 
and all other variables in $\calS$ to $d$.
Algorithm $\calB$ first constructs from $\calD$ a database $\calD'$ for $Q_\dagger$
as follows. For each tuple $(a,b)$ in relation $R$ and each atom 
$R'(\calX')$ in $\calR_{A,B}$, it assigns the tuple $(\calX', A=a, B=b, d)$
to relation $R'$. Likewise,   
for each value $b$ in relation $S$ and each atom 
$S'(\calY')$ in $\calS_{(\neg A, B)}$, it assigns the tuple $(\calY', B=b, d)$
to relation $S'$. The size of $\calD'$ is linear in $N$.
Then, algorithm $\calB$ executes the preprocessing for $Q_{\dagger}$ on $\calD'$.
Each single-tuple update $\{(a,b) \mapsto m\}$ 
to relation $R$ is translated 
to a sequence of single-tuple updates 
$\{(\calX', A=a, B=b, d) \mapsto m\}$ to all relations 
 referenced by atoms in $\calR_{(A,B)}$.
Analogously, updates $\{b \mapsto m\}$ 
to $S$ are translated to updates 
$\{(\calS', B=b, d) \mapsto m\}$ to all relations 
$\calS'$ with $\calS'(\calY')\in \calS_{(\neg A,B)}$.  
Hence, the amortised update time is 
$\bigO{N^{0.5-\gamma}}$. 
Each input tuple  $(a)$ for $Q_3$ is translated into an input tuple  
$(\calI_{\dagger}, A=a, d)$ for $Q_\dagger$
where $\calI_{\dagger}$ is the set of input variables for $Q_{\dagger}$. 
Recall that all free variables of $Q_{\dagger}$ are input. 
The answer of $Q_3(\cdot | a)$ is true if and only if 
the answer of  $Q_\dagger(\cdot |(\calI_\dagger, A=a, d))$ is true. 
The answer time is $\bigO{N^{0.5-\gamma}}$. 
We conclude that $Q_3$ can be evaluated with 
$\bigO{N^{0.5-\gamma}}$ enumeration delay and 
 $\bigO{N^{0.5-\gamma}}$ amortised update time,
 a contradiction due to the  \OMv conjecture.  
  
 %%%%%%%%%%%%%%
 %General case where second condition is not satisfied
 %%%%%%%%%%%%%
 \smallskip
 \textit{Case (2).}
We now consider the case that the query $Q_{\dagger}$ 
is free-dominant but not input-dominant.
  In this case, we reduce the evaluation of 
  $Q_4$ to the evaluation of $Q_{\dagger}$.
  The reduction is analogous to Case (1).
  The way we encode the atoms $R(A,B)$
  and $S(B)$, do preprocessing, and translate the updates 
  is exactly the same as in Case (1).
  The only difference is the way we retrieve 
  the $B$-values in $Q_4(B|a)$ for an input value $a$.
   We translate $a$ into an input tuple to $Q_{\dagger}$ where all input variables 
   besides $A$ are assigned to $d$.
   Recall that $Q_{\dagger}$ might have several output variables besides 
   $B$. By construction, they can be assigned only to $d$. 
   Hence, all output tuples returned by $Q_\dagger$
   have distinct  $B$-values. These $B$-values constitute the result of 
   $Q_4(B|a)$. We conclude that $Q_4$ can be evaluated with 
$\bigO{N^{0.5-\gamma}}$ enumeration delay and 
 $\bigO{N^{0.5-\gamma}}$ amortised update time,
 which contradicts the  \OMv conjecture.  
 
 \medskip
 
 Overall, we obtain that CQAPs that are not in CQAP$_0$ and do not have 
repeating relation symbols cannot be evaluated 
with $\bigO{N^{0.5-\gamma}}$ enumeration delay and 
 $\bigO{N^{0.5-\gamma}}$ amortised update time for any $\gamma >0$, 
unless the \OMv conjecture fails. This concludes the proof of the lower bound statement in 
Theorem~\ref{thm:dichotomy}.

%% file: tradeoff.tex
\section{Trade-Offs for CQAPs with Hierarchical Fractures}
\label{sec:trade-off}
For CQAPs with hierarchical fractures, we can parameterise the complexities in Theorem~\ref{thm:general} 
to obtain trade-offs between preprocessing time, update time, and enumeration delay. We first restate our main result
on such trade-offs from Section~\ref{sec:results}:

\medskip
\textbf{Theorem~\ref{thm:main_dynamic}}
\textit{Let any CQAP $Q$ with static width $\fw$ and dynamic width $\dfw$, a database of size $N$, and $\eps \in [0,1]$. If $Q$'s fracture is hierarchical, then $Q$ admits $\bigO{N^{1 + (\fw -1)\eps}}$ preprocessing time, $\bigO{N^{1-\eps}}$ enumeration delay, and $\bigO{N^{\dfw\eps}}$ amortised  update time for single-tuple updates.
}

\medskip

We achieve these trade-offs by following two core ideas from prior work~\cite{Trade_Offs_LMCS23}.
First, we partition the input relations into heavy and light parts based on the degrees of the values. This transforms a query over the input relations into a union of queries over heavy and light relation parts.
Second, we employ different evaluation strategies for different heavy-light combinations of parts of the input relations. This allows us to confine the worst-case behaviour during query evaluation, caused by high-degree values in the database. 

We construct a set of VOs for the hierarchical fracture of a given CQAP. Each VO represents a different evaluation strategy over heavy and light relation parts. 
For VOs over light relation parts, we follow the general approach from Section~\ref{sec:preprocessing} and construct view trees from access-top VOs.  
For VOs involving heavy relation parts, we construct view trees from VOs that are not access-top, thus yielding non-constant enumeration delay but better preprocessing and update times. This trade-off is controlled by the parameter $\epsilon$.

The enumeration faces a new challenge: the tuples encoded in the constructed view trees may overlap, yet we need to enumerate only distinct tuples, i.e., tuples that have not been reported before. To address this challenge, we adapt the union algorithm from prior work~\cite{Durand:CSL:11}, which is originally designed to enumerate distinct elements from a union of sets. We modify this algorithm to enumerate distinct tuples from multiple view trees.

Handling updates also faces a new challenge: 
although propagating updates in the constructed view trees follows the procedure from Section~\ref{sec:updates},
updates may change the degrees of values, causing previously light tuples to become heavy and vice versa. 
In such cases, we need to rebalance the data partitioning and possibly recompute some views. While such rebalancing steps may take longer than a single-tuple update, they happen periodically, and their amortised cost remains the same as that of a single-tuple update.

Sections~\ref{sec:partitioning}-\ref{sec:trade-off-updates} elaborate our technique and algorithmic ideas
that achieve the trade-offs in Theorem~\ref{thm:main_dynamic}. 
The full details of our approach are given in Appendices~C--E
of the technical report~\cite{access_pattern_arxiv}.
Section~\ref{sec:comparison} compares our maintenance strategy achieving these trade-offs with typical eager and lazy approaches. 
%The full proof of Theorem~\ref{thm:main_dynamic}  is given in Appendix~\ref{appendix:results}.

\subsection{Data Partitioning}
\label{sec:partitioning}
We partition relations based on the frequencies of their values. 
For a database $\calD$, relation $R \in \calD$ over schema $\mathcal{X}$, schema $\calS \subset\mathcal{X}$, and threshold $\theta$, the pair $(R^{\calS\veryshortarrow H}, R^{\calS\veryshortarrow L})$ is a {\em partition} of $R$ on $\calS$ with threshold $\theta$
	if it satisfies the  conditions:\\[6pt]
	\begin{tabular}{rl}
		{(union)}            & $R(\tup{x}) =R^{\calS\veryshortarrow H}(\tup{x}) + R^{\calS\veryshortarrow L}(\tup{x})$ for $\tup{x} \in \Dom(\calX)$ \\[4pt]
		{(domain partition)} &
		$\pi_{\calS} R^{\calS\veryshortarrow H} \cap \pi_{\calS} 
		R^{\calS\veryshortarrow L} = \emptyset$          \\[4pt]
		(heavy part)         & $\forall \tup{t} \in \pi_{\calS} R^{\calS\veryshortarrow H},$ $\exists K \in \calD$:
		$|\sigma_{\calS= \tup{t}} K| \geq \frac{1}{2}\theta$   \\[4pt]
		(light part)         & $\forall \tup{t} \in \pi_{\calS} R^{\calS\veryshortarrow L}$ and $\forall K \in \calD$:
		$|\sigma_{\calS=\tup{t}} K| < \frac{3}{2}\theta$
	\end{tabular}\\[6pt]
	We call $(R^{\calS\veryshortarrow H}, R^{\calS\veryshortarrow L})$ a {\em strict} partition of $R$ on $\calS$ with threshold
	$\theta$ if it satisfies the union and
	domain partition conditions and the strict versions
	of the heavy and light part conditions:
	\\[6pt]
	\begin{tabular}{rl}
		(strict heavy part)         & $\forall \tup{t} \in \pi_{\calS} R^{\calS\veryshortarrow H}$,  
		$\exists K \in \calD$:
		$|\sigma_{\calS= \tup{t}} K| \geq \theta$   \\[4pt]
		(strict light part)         & $\forall\tup{t} \in \pi_{\calS} R^{\calS\veryshortarrow L}$ and $\forall K \in \calD$:
		$|\sigma_{\calS=\tup{t}} K| < \theta$
	\end{tabular}\\[6pt]
The relation $R^{\calS\veryshortarrow H}$ is called {\em heavy}, and the relation $R^{\calS\veryshortarrow L}$ is called {\em light} on the partition key $\calS$, as they consist of all $\calS$-tuples 
in $R$ that are heavy and respectively light.
Due to the domain partition, the relations $R^{\calS\veryshortarrow H}$ and 
$R^{\calS\veryshortarrow L}$ are disjoint. 
For $|\calD|=N$ and a strict partition $(R^{\calS\veryshortarrow H}, R^{\calS\veryshortarrow L})$ of $R$ on $\mathcal{S}$ with threshold 
$\theta=N^\eps$ for $\eps\in[0,1]$, we have two bounds:
\[
\text{(1) } \forall\tup{t} \in \pi_{\calS} R^{\calS\veryshortarrow L}: |\sigma_{\calS=\tup{t}} R^{\calS \veryshortarrow L}| < \theta = N^\eps,
\quad\text{ and }\quad
\text{(2) } |\pi_{\calS} R^{\calS\veryshortarrow H}| \leq \frac{N}{\theta} = N^{1-\eps}.
\]
% (1) $\forall\tup{t} \in \pi_{\calS} R^{\calS\veryshortarrow L}: |\sigma_{\calS=\tup{t}} R^{\calS \veryshortarrow L}| < \theta = N^\eps$;
% and 
% (2) $|\pi_{\calS} R^{\calS\veryshortarrow H}| \leq \frac{N}{\theta} = N^{1-\eps}$.
The first bound follows directly from the strict light part condition. 
The second bound follows from the strict heavy part condition, which says that for each tuple $\tup{t} \in |\pi_\calS R^{\calS\veryshortarrow H}|$, there exists a relation $K$ such that $|\sigma_{\calS=\tup{t}}K| \geq N^\eps$. Assume now that there exists more than $N^{1-\eps}$ such tuples. Then, the database contains more than $N^{1-\eps} N^\eps = N$ tuples, which contradicts our assumption that the database is of size $N$.

Disjoint relation parts can be further partitioned independently of each other on different partition keys.
We write $R^{\calS_1 \veryshortarrow s_1, \ldots, \calS_n \veryshortarrow s_n}$ to denote the relation part obtained after partitioning $R^{\calS_1 \veryshortarrow s_1, \ldots, \calS_{n-1} \veryshortarrow s_{n-1}}$ on $\calS_n$, where $s_i \in \{H,L\}$ for $i \in [n]$. 
The domain of $R^{\calS_1 \veryshortarrow s_1, \ldots, \calS_n \veryshortarrow s_n}$ is the intersection of the domains of $R^{\calS_i \veryshortarrow s_i}$, for $i \in [n]$.
We refer to $\calS_1 \veryshortarrow s_1, \ldots, \calS_n \veryshortarrow s_n$ as a heavy-light signature for $R$.
Consider for instance a relation $R$ with schema $(A,B,C)$.
One possible partition of $R$ consists of the relation parts 
$R^{A\veryshortarrow L}$, $R^{A\veryshortarrow H, AB \veryshortarrow L}$,
and $R^{A\veryshortarrow H, AB \veryshortarrow H}$.
The union of these relation parts constitutes the relation $R$.
In our approach described in Sections \ref{sec:trade-off-preprocessing}-\ref{sec:trade-off-updates}, the partition keys
$\calS_1, \ldots, \calS_n$ in a signature  
$\calS_1 \shortarrow s_1, \ldots, \calS_n \shortarrow s_n$ form a strict  inclusion chain, i.e., 
$\calS_1 \subset \ldots \subset \calS_n$. In general, partition keys can be disjoint.  
\subsection{Preprocessing}
\label{sec:trade-off-preprocessing}
The preprocessing has two steps. 
First, we construct a set of VOs corresponding to the different evaluation strategies over the heavy and light relation parts. \nop{Each such VO is constructed from the {\em canonical} VO of $Q$ by turning some of its subtrees into access-top VOs.} 
Second, we build a view tree from each such VO using the function $\vt$ from Figure~\ref{fig:general_view_tree_construction}.
We illustrate the idea in the following example.

\begin{figure}[t]
	\centering
	\footnotesize
	\begin{minipage}[b]{0.49\linewidth}
	  \centering
	  \begin{tikzpicture}[xscale=0.96, yscale=0.7]
		\node at (-0.9, -1.0) (A) {\small $V_{A_1}(A_1)$};
		\node at (-0.9, -2.0) (D) {\small $V_C(A_1,C)$} edge[-] (A);
		\node at (-0.9, -3.0) (C) {\small $V_D(A_1,C,D)$} edge[-] (D);
		\node at (-0.9, -4.0) (B) {\small $V_B(A_1,B,C,D)$} edge[-] (C);
		\node at (-2.5, -5.0) (R) {\small  $R^{A_1B \veryshortarrow L}(A_1,B,C)$} edge[-] (B);
		\node at (0.7, -5.0) (S) {\small  $S^{A_1B \veryshortarrow L}(A_1,B,D)$} edge[-] (B);
	  \end{tikzpicture}
	\end{minipage}
  \begin{minipage}[b]{0.49\linewidth}
	\centering
	  \begin{tikzpicture}[xscale=0.96, yscale=0.7]
		  \node at (-0.9, -1.0) (B') {\small $V_{A_1}(A_1)$};
		  \node at (-0.9, -2.0) (B) {\small $V_B(A_1,B)$} edge[-] (B');
		  \node at (-2.5, -3.0) (C') {\small $V'_C(A_1,B)$} edge[-] (B);
		  \node at (-2.5, -4.0) (C) {\small $V_C(A_1,B,C)$} edge[-] (C');
		  \node at (-2.5, -5.0) (R) {\small  $R^{A_1B \veryshortarrow H}(A_1,B,C)$} edge[-] (C);
		  \node at (0.7, -3.0) (D') {\small $V'_D(A_1,B)$} edge[-] (B);
		  \node at (0.7, -4.0) (D) {\small $V_D(A_1,B,D)$} edge[-] (D');
		  \node at (0.7, -5.0) (S) {\small  $S^{A_1B \veryshortarrow H}(A_1,B,D)$} edge[-] (D);
	  \end{tikzpicture}
  \end{minipage}
  \caption{
  	View trees constructed for $Q_1(D|A_1,C)=R(A_1,B,C), S(A_1,B,$ $D)$ from Example~\ref{ex:CQAP_1} 
	using the VOs:  $A_1-C-D-B-\{R^{A_1B \veryshortarrow L}(A_1,B,C), S^{A_1B \veryshortarrow L}(A_1,B,D)\}$ (left) and $A_1-B-\{ C - R^{A_1B \veryshortarrow H}(A_1,B,C), D-S^{A_1B \veryshortarrow H}(A_1,B,D)\}$ (right).}
  \label{fig:preprocessing_CQAP_1}
\end{figure} 

\nop{
% We next exemplify our adaptive maintenance strategy on a $\text{CQAP}_1$ query.
}

\begin{exa}
\label{ex:CQAP_1}
We explain the construction of the view trees for the connected component from Figure~\ref{fig:general_hypergraphs} (middle) corresponding to the query $Q_1(D|A_1,C)=R(A_1,B,C),S(A_1,B,D)$. 
In the canonical VO of this query, shown in Figure~\ref{fig:general_CPAP_0} (left), the bound variable $B$ dominates the free variables $C$ and $D$. We create a strict partition of the relations $R$ and $S$ on $(A_1,B)$ with threshold $N^{\eps}$, where $N$ is the database size.

To evaluate the join over the light relation parts, we turn the subtree in the canonical VO rooted at $B$ into an access-top VO and construct a view tree following this new VO, see Figure~\ref{fig:preprocessing_CQAP_1} (left). 
We compute the view $V_B(A_1,B,C,D)$ in time $\bigO{N^{1+\eps}}$: For each $(a,b,c)$ in the light part $R^{A_1B\veryshortarrow L}(A_1,B,C)$ of $R$, we fetch the $D$-values in $S^{A_1B\veryshortarrow L}(A_1,B,D)$ that are paired with $(a,b)$. The iteration in $R^{A_1B\veryshortarrow L}(A_1,B,C)$ takes $\bigO{N}$ time and for each $(a,b)$, there are at most $N^\eps$ $D$-values in $S^{A_1B\veryshortarrow L}(A_1,B,D)$. 
The views $V_{D}$, $V_{C}$, and $V_{A}$ result from $V_B$
by marginalising out one variable at a time. Overall, this takes $\bigO{N^{1+\eps}}$ time.  
 
 To evaluate the join over the heavy parts of $R$ and $S$, we construct a view tree following the canonical VO (Figure~\ref{fig:preprocessing_CQAP_1} right). The VO and view tree are the same as in Figure~\ref{fig:general_hypergraphs}, except that the leaves are the heavy parts of $R$ and $S$. We can materialise this view tree in $\bigO{N}$ time, cf.\@ Example~\ref{ex:general_preprocessing_CQAP0}.

 Overall, we can compute the two view trees in $\bigO{N^{1+\eps}}$ time. 
\qed
\end{exa}

\begin{figure}[t]
	\centering
	\setlength{\tabcolsep}{3pt}
	\renewcommand{\arraystretch}{1.1}
	\setcounter{magicrownumbers}{0}
	\begin{tabular}[t]{@{}l@{}l@{}l@{}}
		\toprule
		\multicolumn{3}{l}{$\vos(\text{VO } \nu, \text{access pattern } (\calO | \calI))$ : set of VOs} \\
		\midrule
		\multicolumn{3}{l}{\MATCH $\nu$:} \\
		\midrule
		\phantom{ab} & $R^{\signs}(\calY)$\hspace*{2.5em} &
		\hspace{-0.6cm}
		\linenumber \RETURN $\{R^{\signs}(\calY)\}$\\[2pt]
		\cmidrule{2-3} \\[-6pt]
		             &
		\begin{minipage}[t]{2.75cm}
			\hspace*{-0.35cm}
			\begin{tikzpicture}[xscale=0.45, yscale=1]
				\node at (0,-2)  (n4) {$X$};
				\node at (-1.2,-3)  (n1) {$\nu_1$} edge[-] (n4);
				\node at (0,-3)  (n2) {$\ldots$};
				\node at (1.2,-3)  (n3) {$\nu_k$} edge[-] (n4);
			\end{tikzpicture}
		\end{minipage}
		             &
					 \hspace{-0.6cm}
		\begin{minipage}[t]{11.3cm}
			% \hspace{-1.65cm}
			\vspace{-1.55cm}
			\linenumber \LET $key = \anc_{\omega}(X) \cup \{X\}$\\[0.5ex]
			\linenumber \LET $\calI_X = (\calI \cap \vars(\nu)) \cup \anc_{\omega}(X)$ \\[0.5ex]
			\linenumber \LET $\calO_X = \calO \cap \vars(\nu)$ \\[0.5ex]
			\linenumber \LET $Q_X(\calO_X | \calI_X) = \text{join of } \atoms(\nu)$ \\[0.5ex]
			\linenumber \IF $Q_X(\calO_X|\calI_X)$ is $\text{CQAP}_0$ \\[0.5ex]
			\linenumber \TAB \RETURN $\{\, \ftvo(\nu, (\calO_X|\calI_X)) \,\}$ \\[0.5ex]
			\linenumber \IF $X \in \calI$ \OR ($X \in \calO$ \AND $\calI \cap \vars(\nu) = \emptyset$) \\[0.5ex]
			\linenumber \TAB \RETURN $
				\Bigg\{
				\begin{array}{@{~~}c@{~~}}
					\tikz {
						\node at (0,-1)  (n4) {$X$};
						\node at (-0.6,-1.75)  (n1) {$\hat{\nu}_1$} edge[-] (n4);
						\node at (0,-1.75)  (n2) {$\ldots$};
						\node at (0.6,-1.75)  (n3) {$\hat{\nu}_k$} edge[-] (n4);
					}
				\end{array} \Bigg|\
				\hat{\nu}_i \in \vos(\nu_i, (\calO|\calI)), \, \forall i \in [k] \Bigg\}$ \\[0.5ex]
			\linenumber \LET $htrees = 
      \Bigg\{
				\begin{array}{@{~~}c@{~~}}
					\tikz {
						\node at (0,-1)  (n4) {$X$};
						\node at (-0.6,-1.75)  (n1) {$\hat{\nu}_1$} edge[-] (n4);
						\node at (0,-1.75)  (n2) {$\ldots$};
						\node at (0.6,-1.75)  (n3) {$\hat{\nu}_k$} edge[-] (n4);
					}
				\end{array} \Bigg|\ \hat{\nu}_i \in \vos(\nu_i^{key \shortarrow H}, (\calO|\calI)), \, \forall i \in [k] \Bigg\}$ \\[0.5ex]
			\linenumber \LET $ltree = \ftvo(\nu^{key \shortarrow L}, (\calO_X|\calI_X))$ \\[0.5ex]
			\linenumber \RETURN $htrees \cup \{\, ltree \,\}$\\[-1.5ex]
		\end{minipage} \\
		\bottomrule
	\end{tabular}\vspace{-0.1em}
	\caption{Construction of a set of VOs
		from a canonical VO $\omega$ of a hierarchical CQAP with access pattern $(\calO|\calI)$.
		Each constructed VO corresponds to an evaluation strategy of some part of the
		query result. The VO $\nu^{key \shortarrow s}$ for $s \in \{H,L\}$ has the structure of $\nu$ but the heavy-light signature of each atom is extended by $key \rightarrow s$. 
		The procedure $\ftvo$ is given in  Figure~\ref{fig:canonical-to-free-top}.
	}
	\label{fig:evaluation_strategies}
\end{figure}

We next describe  the construction of a set of VOs
from a canonical VO $\omega$ of a hierarchical CQAP $Q(\calO|\calI)$.
Without loss of generality, we assume that $\omega$ is a tree; in case $\omega$ is a forest, the reasoning below applies independently to each tree in the forest.
Figure~\ref{fig:evaluation_strategies} shows the construction procedure for a canonical VO $\omega$ and an access pattern $(\calO | \calI)$.
The construction proceeds recursively on the structure of $\omega$ and forms the query $Q_X(\calO_X|\calI_X)$ at each variable $X$ (Line 5).
The query $Q_X$ is the join of the atoms in $\omega_X$,
the set $\calO_X$ consists of the output variables in $\omega_X$, and
the set $\calI_X$ consists of the input variables in $\omega_X$ and all ancestor variables along the path from $X$ to the root of $\omega$.
The next step analyses the query $Q_X$.

If $Q_X$ is in $\text{CQAP}_0$, we turn $\omega_X$ into an access-top VO for $Q_X$ using 
the procedure $\ftvo$ in Figure~\ref{fig:canonical-to-free-top} (Lines 6-7).
Queries in $\text{CQAP}_0$ admit a canonical access-top VO. Hence, 
for such queries, this restructuring does not increase the static width
of $\omega_X$.

If $Q_X$ is not in $\text{CQAP}_0$, 
then $\omega_X$ contains a problematic variable, which is either a
bound variable that dominates a free variable or an output variable that dominates an input variable. 
If $X$ is {\em not} a problematic variable,
we recur on each subtree and combine the constructed VOs (Line 9). 
Otherwise, we form evaluation strategies that compute different parts of the result of $Q_X$ over its input relations partitioned on $\mathit{key}$, which is the set of variables on the path from $X$ to the root of the canonical VO for $Q$, including $X$.
We create two sets of VOs: $htrees$ and $ltree$.
For the former, 
for each subtree $\nu_i$ of $\omega_X$, we construct a VO $\nu_i^{key \shortarrow H}$ by extending the heavy-light signature of each atom in $\nu_i$ with $key \rightarrow H$,
and we recur on $\nu_i^{key \shortarrow H}$ (Line 10).
This ensures that the evaluation of $Q_X$ is over relation parts that are heavy on $key$.
The VO $ltree$ is obtained by extending the heavy-light signature of each atom in $\omega_X$ with $\{key \shortarrow L\}$ and turning $\omega_X^{key \shortarrow L}$ into an access-top VO (Line 11);
% The VO $ltree$ is obtained by turning $\omega_X$ into an 
% access-top VO using 
% the procedure $\ftvo$ in Figure~\ref{fig:canonical-to-free-top}
% and extending the heavy-light signature of each atom
% in the VO with \rev{$\{key \shortarrow K\}$} (Line 11);
this restructuring of the VO may increase its static width.

We construct a view tree for each VO formed in the previous step. 
For each view tree, we create a strict partition of the input relations based on their heavy-light signature and compute the queries defining the views. We refer to this step as view tree materialisation.

We next discuss the complexity of view tree materialisation.
The view trees constructed for the evaluation of queries in $\text{CQAP}_0$ or over heavy relation parts follow canonical VOs, meaning that they can be materialised in linear time. 
The view trees constructed for the evaluation of queries over light relation parts follow access-top VOs. 
Using the degree constraints in the input relations, each such view tree can be materialised in $\bigO{N^{1+(\fw-1)\eps}}$ time, where $\fw$ is the static width of the query.
We give the intuition for this complexity.
Consider a view $V(\calS)$ in a view tree over light relation parts and the set 
$\calA$ of leaf atoms in the subtree rooted at $V(\calS)$.
%The view $V(\calS)$ can be computed by joining the 
%the atoms in $\calA$ and aggregating away the variables that 
%that do not appear in $\calS$. Doing the aggregation after the join can blowup the overall complexity.
%Hence, we push the aggregation past the join as follows.
Since for each hierarchical query, the integral edge cover number is the same as the fractional edge cover number 
(Lemma~\ref{lem:rho_rhostar}), $\fw$ must be a natural number.  
By definition of $\fw$, we can cover all variables in $\calS$ using at most $\fw$ atoms from $\calA$.
Let $\calA_1 \subseteq \calA$ be a set of at most $\fw$ atoms that cover the variables in $\calS$.
First, we compute the join of the atoms in $\calA_1$ as follows:
We choose one atom $R(\calX)$ in $\calA_1$, iterate over the tuples in $R$, and for each such tuple, we iterate over the matching tuples in the relations of the other at most $\fw-1$ atoms in $\calA_1$.
Since all leaf relations are light, there are at most $N^\eps$ matching tuples in each relation.
Thus, the join can be computed in $\bigO{N^{1+(\fw-1)\eps}}$ time. Let $T$
be the resulting relation and $\calY$ the set of all variables of the atoms in $\calA_1$.
We can rewrite the view $V(\calS)$ as:
$$V(\calS) = T(\calY), R_1(\calX_1), \ldots , R_n(\calX_n),$$ 
where $R_1(\calX_1), \ldots , R_n(\calX_n)$ are the atoms in $\calA \setminus \calA_1$. 
The above query is free-connex $\alpha$-acylic \cite{BaganDG07}, since its body is $\alpha$-acyclic and the free variables 
$\calS$ are included in the schema $\calY$ of one atom. 
Hence, using Yannakakis' algorithm, $V$ can be computed 
in time linear in the input plus the output size of the query~\cite{BeeriFMY83}. 
The input size is upper-bounded by the worst-case size of
$T(\calY)$, which is  $\bigO{N^{1+(\fw-1)\eps}}$. The output is a subset of       
$T$ projected onto $\calS$, hence, its size is also $\bigO{N^{1+(\fw-1)\eps}}$. 
Thus, $V(\calS)$ can be computed in $\bigO{N^{1+(\fw-1)\eps}}$ time.
Overall, the view tree materialisation takes $\bigO{N^{1+(\fw-1)\eps}}$ time, as stated in Theorem~\ref{thm:main_dynamic}.

\subsection{Enumeration}
\label{sec:trade-off-enumeration}
We next discuss how to enumerate output tuples from the view trees constructed for a CQAP with hierarchical fractures.
Our approach builds upon the enumeration procedure for hierarchical queries from prior work~\cite{Trade_Offs_LMCS23}.
For queries in $\text{CQAP}_0$, the preprocessing stage constructs view trees from access-top VOs. Such view trees admit constant enumeration delay, as discussed in Section~\ref{sec:enumeration}. 

For queries not in $\text{CQAP}_0$, the preprocessing stage constructs view trees from VOs that are not access-top. 
Enumerating distinct tuples from these view trees poses two challenges:
(1) for view trees built from VOs that are not access-top,
the enumeration approach from Section~\ref{sec:enumeration} would report  
the values of bound variables before the values of free variables or
the values of output variables before setting the values of input variables; and
(2) the tuples encoded in the constructed view trees may overlap, while we need to enumerate distinct tuples.
We rely on the union algorithm~\cite{Durand:CSL:11} to handle these challenges.

%%%%%%%%%%%%%%%%
% UNION
%%%%%%%%%%%%%%%%
\begin{figure}[t]
%\small
\begin{center}
\renewcommand{\arraystretch}{1.2}
\setcounter{magicrownumbers}{0}
\begin{tabular}{l@{}}
	\toprule
	\textsc{Union}$(\text{view trees } T_1, \ldots, T_n):$ tuple \\
	\midrule
	\linenumber \IF ($n=1$) \RETURN $T_n.next(\,)$ \\
	\linenumber \IF ($t_{\hspace{0.06em}[n-1]} := $ \textsc{Union}($T_1,\ldots, T_{n-1}$)) $\neq$ \EOF \\
	\linenumber \TAB \IF ($T_n.\mathit{lookup}\hspace{0.06em}(t_{\hspace{0.06em}[n-1]}) = True$) \\
	\linenumber \TAB\TAB $t_n := T_n.next(\,)$ \\
	\linenumber \TAB\TAB \RETURN $t_n$ \\  
	\linenumber \TAB \RETURN $t_{\hspace{0.06em}[n-1]}$ \\
	\linenumber \IF $(t_n := T_n.next(\,)) \neq $ \EOF \\
	% \linenumber \TAB $m_{[n]} = m_n + \sum_{i\in[n-1]} T_i.\mathit{lookup}\hspace{0.06em}(t_n)$\\
	\linenumber \TAB \RETURN $t_n$ \\
	\linenumber \RETURN \EOF \\
	\bottomrule
\end{tabular}
\end{center}
\caption{Report the next tuple in a union of view trees.}
\label{fig:enumeration-union}
\end{figure}

The \textsc{Union} algorithm is given in Figure~\ref{fig:enumeration-union}. 
It takes as input $n$ view trees that represent possibly overlapping sets of tuples and returns a tuple that is distinct from all tuples returned before. 
Each view tree supports two operations: $\mathit{next}(\,)$ returns the next tuple in the view tree or $\tup{EOF}$ if the view tree is exhausted, and $\mathit{lookup}(t)$ checks whether the tuple $t$ is present in the view tree.

We first explain the algorithm on two view trees $T_1$ and $T_2$ that represent possibly overlapping sets of tuples.
Each call returns one tuple or $\tup{EOF}$. 
The algorithm returns the next tuple $t_1$ in $T_1$ only if it is  not present in $T_2$; otherwise, it returns the next tuple in $T_2$ (Lines~2-6). In case $T_1$ is exhausted, the algorithm returns the next tuple in $T_2$, or $\tup{EOF}$ in case $T_2$ is also exhausted.

In the case of more than two view trees ($n > 2$), we consider the union of the first $n-1$ view trees as one view tree and $T_n$ as another view tree. This reduces the general case to the previous case of two view trees.

The \textsc{Union} algorithm performs $\bigO{n}$ $\mathit{lookup}$ and $\mathit{next}$ operations over $n$ view trees before reporting a tuple. Thus, its runtime is $\bigO{n(delay + lookup)}$, where $delay$ is the time to retrieve the next tuple in a view tree and $lookup$ is the cost of a lookup into a view tree.

We use the \textsc{Union} algorithm to address the two aforementioned challenges in enumerating from the view trees constructed from VOs that are not access-top. 
For the first challenge, let $A$ be a variable that violates the free-dominance or input-dominance condition.
The constructed non-access-top view trees are over relation parts where $A$-values are heavy.
We instantiate a view tree for each $A$-value and use the union algorithm to report only the distinct tuples.
The number of instantiated view trees is upper-bounded by the number of heavy $A$-values, i.e., $n = \bigO{N^{1-\eps}}$.
Since $A$ is fixed in each instantiated view tree, $A$ is effectively an input variable and the view tree is as if constructed from an access-top VO,
and thus supports constant-delay enumeration using the enumeration approach from Section~\ref{sec:enumeration}.
The lookup operation can be performed using the enumeration procedure, where all free variables are considered as input variables and set to the values of the tuple to be looked up, which thus takes constant time.
Hence, the delay of the union algorithm is $\bigO{N^{1-\eps}}$.

For the second challenge, we use the union algorithm to report only distinct tuples from the set of view trees.
As explained in the first challenge, the view trees admit enumeration delay $\bigO{N^{1-\eps}}$ and thus $\bigO{N^{1-\eps}}$ lookup time.
The number of constructed view trees is constant.
Overall, the delay of the union algorithm is $\bigO{N^{1-\eps}}$.
%%%%%%%%%%%%%%%%
% END: UNION
%%%%%%%%%%%%%%%%

\begin{exa}
We explain the enumeration procedure for the view trees from Figure~\ref{fig:preprocessing_CQAP_1}, constructed for the query $Q_1(D|A_1,C)=R(A_1,B,C), S(A_1,B,$ $D)$.
The view tree on the left, constructed over the light parts of $R$ and $S$, corresponds to an access-top VO. For a fixed $(A_1,C)$-value, enumerating the matching $D$-values from $V_D$ takes constant time per output value.
The view tree on the right, however, corresponds to a VO where $B$ violates the free-dominance and input-dominance conditions. 
This view tree comprises the heavy parts of $R$ and $S$ partitioned on $(A_1, B)$. The number of distinct $(A_1, B)$-values in each part is at most $N^{1-\epsilon}$, meaning that the size of the view $V_B(A_1, B)$ built on top of these parts is also at most $N^{1-\epsilon}$.
To resolve the issue with violating $B$, 
we find the $B$-values that are paired with the input $A_1$-value in $V_B$ and also paired with the input $(A_1,C)$-value in $R^{A_1B\veryshortarrow H}$, and instantiate for each such $B$-value a view tree.
The number of such view trees is at most $N^{1-\epsilon}$, and each view tree supports constant-time lookup and constant-delay enumeration of the $D$-values in $S^{A_1B\veryshortarrow H}$.
To report only distinct $D$-values, we employ the union algorithm over the iterators instantiated from the right tree and the iterator over the left tree. The cost of de-duplication using this algorithm is proportional to the number of instantiated view tree iterators; thus, the enumeration delay is $\bigO{N^{1-\epsilon}}$.
 \qed
\end{exa}

Appendix~D of the technical report provides a complete proof of the enumeration delay $\bigO{N^{1-\epsilon}}$ from Theorem~\ref{thm:main_dynamic}, along with our enumeration procedure, which uses the standard iterator interface with \textit{open} and \textit{next} methods~\cite{access_pattern_arxiv}.

\subsection{Updates}
\label{sec:trade-off-updates}
A single-tuple update to an input relation may cause changes in several view trees constructed for a given hierarchical CQAP. If the input relation is partitioned, we first identify which part of the relation is affected by the update. We then propagate the update in each view tree containing the affected relation part, as discussed in Section~\ref{sec:updates}.  

\begin{exa}
We consider the maintenance of the view trees from Figure~\ref{fig:preprocessing_CQAP_1} under a single-tuple update $\delta R(a,b,c)$ to $R$. The update affects the heavy part $R^{A_1B\veryshortarrow H}$ if $(a,b) \in \pi_{A_1,B}R^{A_1B\veryshortarrow H}$; otherwise, it affects the light part $R^{A_1B\veryshortarrow L}$.
For the former, we propagate the update from $R^{A_1B\veryshortarrow H}$ to the root. For each view on this path, we compute its delta query and update the view in constant time for fixed $(a,b,c)$. 
For the latter, we compute the delta $\delta V_B(a,b,c,D)$ $= \delta R^{A_1B\veryshortarrow L}(a,b,c), S^{A_1B\veryshortarrow L}(a,b,D)$ in $\bigO{N^\eps}$ time because there are at most $N^\eps$ $D$-values paired with $(a,b)$ in $S^{A_1B\veryshortarrow L}$.
We then update $V_D(a,c,D)$ with $\delta V_D(a,c,D) = \delta V_B(a,b,c,D)$ in $\bigO{N^\eps}$ time and update the views $V_C(A_1,C)$ and $V_{A_1}(A_1)$ in constant time. The case of single-tuple updates to $S$ is analogous.
Overall, maintaining the two view trees under a single-tuple update to any input relation takes $\bigO{N^\eps}$ time. 
\qed
\end{exa}

As the database evolves under updates, we periodically rebalance the relation partitions and views to account for a new database size and updated degrees of data values. The cost of rebalancing is amortised over a sequence of updates.
We give the intuition behind the amortised cost of rebalancing. The full proof is in Appendix~E of the technical 
report~\cite{access_pattern_arxiv}.

\paragraph{Major Rebalancing.}
We loosen the partition threshold to amortise the cost of rebalancing over multiple updates. Instead of the actual database size $N$, the threshold now depends on
a number $M$ for which the invariant $\floor{\frac{1}{4}M} \leq N \leq M$ always holds. If the database size falls below $\floor{\frac{1}{4}M}$ or reaches $M$, we perform major rebalancing, where we halve or respectively double $M$, followed by recreating a strict partition of the input relations with the new threshold $M^\eps$ and recomputing the views.

A major rebalancing requires $\bigO{N}$ time to repartition the relations and $\bigO{N^{1+(\fw-1)\eps}}$ time to recompute the view trees using the procedure from Section~\ref{sec:trade-off-preprocessing}.
This cost is amortised over $\Omega(M)$ updates. After a major rebalancing step, it holds that $N = \frac{1}{2}M$ (after doubling), or $N = \frac{1}{2}M - 1$ (after halving). To violate the size invariant $\floor{\frac{1}{4}M} \leq N \leq M$ and trigger another major rebalancing, the number of required updates is at least $\frac{1}{4}M$. 
The amortised time of major rebalancing is thus $\bigO{N^{(\fw-1)\eps}}$.
By Proposition~\ref{prop:width_delta_inequal}, we have $\dfw = \fw$ or $\dfw = \fw - 1$; hence, the amortised major rebalancing cost is $\bigO{N^{\dfw\eps}}$.

\paragraph*{Minor Rebalancing.}
After an update $\delta R = \{\tup{x} \rightarrow m\}$ to relation $R$, we check the degrees of the values in $\tup{x}$. 
Consider a partition key $key$ that is included in the schema of $\tup{x}$ and the projection $\tup{v}$ of $\tup{x}$ onto $key$.
If $\tup{v}$ is included in the light part of the partition of $R$ on $key$ but the degree of $\tup{v}$ is not below $\frac{3}{2}M^\eps$ in 
at least one input relations, all tuples containing  $\tup{v}$ are moved to the relation parts that are heavy on $\tup{v}$. 
Likewise, if $\tup{v}$ is in a relation part that is heavy on $key$ but the degree of $\tup{v}$ is below $\frac{1}{2}M^\eps$ in all input relations, all tuples containing  $\tup{v}$ are moved to the relation parts that are light on $\tup{v}$.

\nop{light part and heavy part conditions of each partition of $R$. Consider the light part $R^{\calS \veryshortarrow L}$ of $R$ partitioned on schema $\calS$. If the number of tuples in $R^{\calS \veryshortarrow L}$ that agree with $\tup{x}$ on $\calS$ exceeds $\frac{3}{2}M^\eps$, then we perform a minor rebalancing step that moves those tuples from $R^{\calS \veryshortarrow L}$ to $R^{\calS \veryshortarrow H}$.
Additionally, to maintain the light part condition, we also rebalance all other relations $K \in\calD \setminus \{R\}$ by moving the tuples that agree with $\tup{x}$ on $\calS$ from $K^{\calS\veryshortarrow L}$ to $K^{\calS\veryshortarrow H}$.
Similarly, we maintain the heavy part condition: If after the update every relation $K\in\calD$ has fewer than $\frac{1}{2}M^\eps$ tuples that agree with $\tup{x}$ on $\calS$, then we move those tuples from $K^{\calS \veryshortarrow H}$ to $K^{\calS \veryshortarrow L}$.
}

A minor rebalancing step requires $\bigO{N^{(\dfw + 1)\eps}}$ time:
It either moves $\bigO{\frac{3}{2}M^\eps}$ tuples that contain $\tup{v}$ to relations parts that are heavy on $\tup{v}$ (light to heavy) or $\bigO{\frac{1}{2}M^\eps}$ tuples that contain $\tup{v}$ to relation parts that are light on $\tup{v}$ (heavy to light).
Each move is by an insert and a delete operation, which takes $\bigO{N^{\dfw\eps}}$ time.
The total cost $\bigO{N^{(\dfw + 1)\eps}}$ of minor rebalancing is amortised over $\Omega(M^\eps)$ updates. This lower bound on the number of updates comes from the gap between the two thresholds in the heavy and light part conditions.
The amortised cost of minor rebalancing is $\bigO{N^{\dfw\eps}}$.

Overall, even though the cost of rebalancing steps take time more than $\bigO{N^{\dfw\eps}}$, they happen periodically, and their amortised cost remains the same as for a single-tuple update.

\nop{
% An update may change the degree of values over a partition key from light to heavy or vice versa. 
% In such cases, we need to rebalance the partitioning and possibly recompute some  views. Although such rebalancing steps may take time more than $\bigO{N^{\dfw\eps}}$, they happen periodically and their amortised cost remains the same as for a single-tuple update.
}

\subsection{Comparison with Prior Approaches}
\label{sec:comparison}
We compare our adaptive maintenance strategy with the mainstream eager and lazy approaches in an IVM scenario where either all or a fraction of the output tuples are reported after a batch of updates.
We show in the following examples that our approach has at most the same overall time complexity as these mainstream approaches.

\begin{exa}\label{ex:compare-eager-lazy-running}

Let us consider the running example with the query from Example~\ref{ex:CQAP_1}:
$$Q_1(D\mid A_1, C) = R(A_1,B,C), S(A_1,B,D)$$
Assume the relations have size $O(N)$. The query result has size $\bigO{N^2}$ for all pairs of input $(A_1, C)$-values and $\bigO{N}$ for one such pair.

We can recover the complexities for typical eager and lazy maintenance approaches using our approach by setting $\epsilon=1$ and respectively $\epsilon=0$ (except for the complexity of the preprocessing in the lazy approach):

\vspace{9pt}
\begin{center}	
	\begin{tabular}{l| l l l }
	Approach  & Preprocessing & Update     & Delay \\\toprule 
		Eager & $\bigO{N^2}$  & $\bigO{N}$ & $\bigO{1}$\\
		Lazy  & $\bigO{1}$    & $\bigO{1}$ & $\bigO{N}$\\
		Ours  & $\bigO{N^{1+\epsilon}}$    & $\bigO{N^{\epsilon}}$ & $\bigO{N^{1-\epsilon}}$
%		Ours ($\epsilon=0.5$) & $\bigO{N^{1.5}}$    & $\bigO{N^{0.5}}$ & $\bigO{N^{0.5}}$\\
	\end{tabular} 
\end{center}
\vspace{9pt}

\begin{figure}
\centering
\begin{minipage}{0.28\textwidth}
\begin{tikzpicture}[xscale=0.9, yscale=0.7]

% Define a custom decoration for crosses
\tikzset{crosses/.style={
	decoration={markings, mark=between positions 0 and 1 step 1.8mm with {
		\draw[thick] (-0.05, -0.05) -- (0.05, 0.05) (-0.05, 0.05) -- (0.05, -0.05);
	}},
	postaction={decorate}
}}

% Define the axes
\draw[->] (0,0) -- (3.5,0) node[right] {$m$};
\draw[->] (0,0) -- (0,4.5) node[above] {$\log_N$cost};

% x-axis ticks and labels
\foreach \x in {0,1,2,3}
	\draw (\x,0.05) -- (\x,-0.05) node[below] {\x};

% y-axis ticks and labels
\draw (0.05,0.5) -- (-0.05,0.5) node[left] {\small $0.5$};
\draw (0.05,1) -- (-0.05,1) node[left] {\small $1$};
\draw (0.05,2) -- (-0.05,2) node[left] {\small $2$};
\draw (0.05,3) -- (-0.05,3) node[left] {\small $3$};
\draw (0.05,4) -- (-0.05,4) node[left] {\small $4$};

% % Define dashed line points 
% \def\points{ {3/3}, {3/4} }

\draw[dashed, gray] (1,1) -- (1,0);
% Loop to create points and dashed lines
\foreach \x/\y in {1/2, 2/3, 3/4, 2/2} {
	\draw[dashed, gray] (\x,\y) -- (\x,0);
	\draw[dashed, gray] (\x,\y) -- (0,\y);
}
% complexities
\draw[thick] (0,1) -- (3.5,4.5) ;
\draw[very thick, dashed, dash pattern=on 2mm off 1mm] (0,1) -- (1,1) -- (3.5,3.5) ;
\draw[draw=none, crosses] (0,0.5) -- (1,0.92) -- (3.5,3.42) ;
\end{tikzpicture}
\vspace*{-0.5cm}
{\small $$k=0$$}
\end{minipage}
\begin{minipage}{0.28\textwidth}
\begin{tikzpicture}[xscale=0.9, yscale=0.7]

% Define a custom decoration for crosses
\tikzset{crosses/.style={
	decoration={markings, mark=between positions 0 and 1 step 1.8mm with {
		\draw[thick] (-0.05, -0.05) -- (0.05, 0.05) (-0.05, 0.05) -- (0.05, -0.05);
	}},
	postaction={decorate}
}}

% Define the axes
\draw[->] (0,0) -- (3.5,0) node[right] {$m$};
\draw[->] (0,0) -- (0,4.5) node[above] {$\log_N$cost};

% x-axis ticks and labels
\foreach \x in {0,1,2,3}
	\draw (\x,0.05) -- (\x,-0.05) node[below] {\x};

% y-axis ticks and labels
\draw (0.05,1) -- (-0.05,1) node[left] {\small $1$};
\draw (0.05,1.5) -- (-0.05,1.5) node[left] {\small $1.5$};
\draw (0.05,2) -- (-0.05,2) node[left] {\small $2$};
\draw (0.05,3) -- (-0.05,3) node[left] {\small $3$};
\draw (0.05,4) -- (-0.05,4) node[left] {\small $4$};

% % Define dashed line points 
% \def\points{ {3/3}, {3/4} }

\draw[dashed, gray] (0,1.5) -- (1,1.5);
\draw[dashed, gray] (1,0) -- (1,2);
% Loop to create points and dashed lines
\foreach \x/\y in {2/3, 3/4} {
	\draw[dashed, gray] (\x,\y) -- (\x,0);
	\draw[dashed, gray] (\x,\y) -- (0,\y);
}
% complexities
\draw[thick] (0,1) -- (3.5,4.52) ;
\draw[very thick, dashed, dash pattern=on 2mm off 1mm] (0,2) -- (2,2) -- (3.5,3.5) ;
\draw[draw=none, crosses] (0,1) -- (2,1.92) -- (3.5,3.42) ;
\end{tikzpicture}
\vspace*{-0.5cm}
{\small $$k=1$$}
\end{minipage}
\begin{minipage}{0.28\textwidth}
\begin{tikzpicture}[xscale=0.9, yscale=0.7]

% Define a custom decoration for crosses
\tikzset{crosses/.style={
	decoration={markings, mark=between positions 0 and 1 step 1.8mm with {
		\draw[thick] (-0.05, -0.05) -- (0.05, 0.05) (-0.05, 0.05) -- (0.05, -0.05);
	}},
	postaction={decorate}
}}

% Define the axes
\draw[->] (0,0) -- (3.5,0) node[right] {$m$};
\draw[->] (0,0) -- (0,4.5) node[above] {$\log_N$cost};

% x-axis ticks and labels
\foreach \x in {0,1,2,3}
	\draw (\x,0.05) -- (\x,-0.05) node[below] {\x};

% y-axis ticks and labels
\draw (0.05,1) -- (-0.05,1) node[left] {\small $1$};
\draw (0.05,1.5) -- (-0.05,1.5) node[left] {\small $1.5$};
\draw (0.05,2) -- (-0.05,2) node[left] {\small $2$};
\draw (0.05,2.5) -- (-0.05,2.5) node[left] {\small $2.5$};
\draw (0.05,3) -- (-0.05,3) node[left] {\small $3$};
\draw (0.05,4) -- (-0.05,4) node[left] {\small $4$};

\draw[dashed, gray] (0,2.5) -- (2,2.5);
\draw[dashed, gray] (1,0) -- (1,2);
\draw[dashed, gray] (2,0) -- (2,3);
% Loop to create points and dashed lines
\foreach \x/\y in {3/4} {
	\draw[dashed, gray] (\x,\y) -- (\x,0);
	\draw[dashed, gray] (\x,\y) -- (0,\y);
}
% complexities
\draw[thick] (0,2) -- (1,2) -- (3.5,4.5) ;
\draw[very thick, dashed, dash pattern=on 2mm off 1mm] (0,3) -- (3,3) -- (3.5,3.5) ;
\draw[draw=none, crosses] (0,1.5) -- (3,2.92) -- (3.5,3.42) ;
\end{tikzpicture}
\vspace*{-0.5cm}
{\small $$k=2$$}
\end{minipage}
\hfill
\begin{minipage}{0.12\textwidth}
% legends
\begin{tikzpicture}
% Define a custom decoration for crosses
\tikzset{crosses/.style={
	decoration={markings, mark=between positions 0 and 1 step 1.8mm with {
		\draw[thick] (-0.05, -0.05) -- (0.05, 0.05) (-0.05, 0.05) -- (0.05, -0.05);
	}},
	postaction={decorate}
}}

\draw[thick] (0,0) -- (0.75,0) node[right] {Eager};
\draw[very thick, dashed, dash pattern=on 2mm off 1mm] (0,-0.5) -- (0.75,-0.5) node[right] {Lazy};
\draw[draw=none, crosses] (0,-1) -- (0.75,-1) node[right] {Ours};
\end{tikzpicture}
\end{minipage}
\vspace{6pt}
\caption{
Plotting the exponents in the complexities of three maintenance approaches (ours, eager, and lazy) as piecewise linear functions in the parameters $m$ and $k$, for processing a batch of $\bigO{N^m}$ single-tuple updates followed by the enumeration of $\bigO{N^k}$ output tuples. Our approach is asymptotically faster or the same as the best of the eager and lazy approaches.
}
\label{fig:tradeoff-example}
\end{figure}

The eager approach precomputes the initial output in $\bigO{N^2}$ time. On a single-tuple update, it eagerly computes the delta query obtained by fixing the variables of one relation to constants; this delta query can be done in linear time. It can then enumerate the $D$-values for any input pair of $\{A_1,C\}$-values with constant delay.

The lazy approach has no precomputation and only updates each relation, without propagating the changes to the query output. For the first enumeration request for a pair $(a_1,c)$ of input values, it needs to calibrate the relations in the residual query $Q_1(D) = R(a_1, B,c), S(a_1,B,D)$. This takes linear time. After that, it can enumerate the $D$-values for that input pair with constant delay. For another pair of input values, it needs again to recalibrate the relations, which also takes linear time. Its delay is linear in worst-case.

Consider now an IVM scenario where, after every $\bigO{N^m}$ single-tuple updates, we request the enumeration of $\bigO{N^k}$ output tuples (in the lexicographic order $A_1,C,D$, for some pairs of input values), for $m \geq 0$ and $0\leq k\leq 2$.
After the initial preprocessing, our approach then takes $\bigO{N^{m+\epsilon}+N^{k+1-\epsilon}} = \bigO{N^{\max\{m+\epsilon,k+1-\epsilon\}}}$ overall time to accommodate the batch of updates followed by the enumeration requests. In contrast, the eager and lazy approaches need time $\bigO{N^{m+1}+N^{k}} = \bigO{N^{\max\{m+1,k\}}}$ and $\bigO{N^{m}+N^{k+1}}=\bigO{N^{\max\{m,k+1\}}}$ respectively. 
For any value for $m\geq 0$, the time complexity of our approach is at most that of the other approaches and it can be asymptotically better. Figure~\ref{fig:tradeoff-example} shows the complexity of the three approaches for different values of $m$ and $k$.

For $k=0$, our approach has the time complexity $\bigO{N^{\frac{m+1}{2}}}$ for $m \leq 1$ and $\epsilon=\frac{1-m}{2}$, and the complexity $\bigO{N^{m}}$ for $m>1$ and $\epsilon=0$.
In contrast, the eager and lazy approaches take time $\bigO{N^{m+1}}$ and $\bigO{N^{\max\{m,1\}}}$ respectively.

For $k=1$, our approach has the time complexity $\bigO{N^{1+\frac{m}{2}}}$ for $m \leq 2$ and $\eps = \frac{2-m}{2}$, and the complexity $\bigO{N^{m}}$ time for $m>2$ and $\eps = 0$.
In contrast, the eager and lazy approaches take time $\bigO{N^{m+1}}$ and $\bigO{N^{\max\{m,2\}}}$ respectively.

For $k=2$, our approach has the time complexity $\bigO{N^{\frac{m+3}{2}}}$ for $m \leq 3$ and $\eps = \frac{3 - m}{2}$, and the complexity $\bigO{N^{m}}$ for $m>3$ and $\eps = 0$. 
In contrast, the eager and lazy approaches take time $\bigO{N^{\max\{m+1, 2\}}}$ and $\bigO{N^{\max\{m,3\}}}$ respectively.

The 4-cycle query from Example~\ref{ex:illustrate_cqap_0}:
$$Q_2(A,C \mid B,D) = R(A,B), S(B,C), T(C,D), U(A,D)$$
exhibits the same trade-offs and complexities as the above acyclic query $Q_1$ for all three approaches: ours, eager, and lazy. Therefore, the analysis and conclusion are the same.
\qed
\end{exa}

%% file: probsemantics.tex
\section{Semantics for Updates in Probabilistic Databases}
\label{sec:probabilistic-semantics}

In this section and the following section, we extend our dynamic evaluation approach to probabilistic databases. Here, we discuss possible semantics of updates in probabilistic databases: Given a single-tuple insertion or a deletion, how to update the probabilistic database to incorporate this update? In the next section, we show how to maintain the result of any query in CQAP$_0$ over probabilistic databases under updates.

We first recall the notion of tuple-independent probabilistic databases and the semantics of query evaluation over such databases. We then contrast several update semantics.

%%%%%%%%%%%%%%%%%%%%%%
\subsection{Tractable Query Evaluation over Probabilistic Databases}
\label{sec:pdb-intro}

A probabilistic database is a relational database in which the tuples are pairwise independent probabilistic events~\cite{Suciu:PDB:11}. We interpret a tuple $t$ as being in the database with probability $p(t)$ and out of the database with probability $1-p(t)$. 
Since each tuple can be in or out of the database, a probabilistic database of $n$ such tuples represents $2^n$ possible worlds, one world for each relational database representing a subset of the set of tuples in the database. 
Let $S$ be the set of tuples in a probabilistic database $\mathcal{W}$ and $W\in\mathcal{W}$ be one of its possible worlds, i.e., $W\subseteq S$. The probability $P(W)$ of $W$ is the product of (1) the probability of each tuple in $W$, and (2) one minus the probability of each tuple in $S$ and not in $W$:
\[
    P(W) = \prod_{t\in W} p(t) \cdot \prod_{t\in S\setminus W} (1-p(t))
\]

Given a Boolean conjunctive query $Q$ and a probabilistic database $\mathcal{W}$, the semantics of $Q$ is to compute $Q$ in each possible world $W$ of $\mathcal{W}$ and sum up the  probabilities of those possible worlds where its answer is true:
\[
    P(Q) = \sum_{W\in\mathcal{W}: Q(W) = \text{true}} P(W) 
\]
For a non-Boolean conjunctive query $Q(\calF)$ with free variables $\calF$ and any tuple of values $\inst{f}\in\Dom(\calF)$ in the active domain of the tuple of free variables $\calF$, we define the residual Boolean query $Q_{\inst{f}}$ where we set the variables in $\calF$ to their respective values in $\inst{f}$. Then, the probability for $\inst{f}$ to be in the output of $Q$ is $P(Q_{\inst{f}})$.

For a CQAP $Q(\calO | \calI)$ and a given tuple $\inst{in}\in\Dom(\calI)$ of values for the input variables, the query $Q(\calO | \inst{in})$ is a (possibly non-Boolean) conjunctive query. Therefore, the probability for a tuple $\inst{out}\in\Dom(\calO)$ in the active domain of the output variables $\calO$ is given by $P(Q_{\inst{out}\circ\inst{in}})$, where $Q_{\inst{out}\circ\inst{in}}$ is the residual Boolean query obtained by setting the free output and input variables to their respective values in the tuple of values $\inst{out}\circ\inst{in}$.

\nop{
Similarly, given a CQAP $Q$ and a probabilistic database $\mathcal{W}$, the semantics of $Q$ is to compute $Q$ in each possible world $W$ of $\mathcal{W}$ and sum up the  probabilities of those possible worlds where its answer is true.
}

The query semantics does not lead to a practical query evaluation, as it requires to iterate over all possible worlds. Instead, state-of-the-art query evaluation techniques (1) exploit the query structure to compute directly on the probabilistic database, without the need to iterate over possible worlds, or (2) derive the so-called query lineage, which is a Boolean function tracing the possible derivations of the query answer from the input tuples, and then use knowledge compilation techniques to compile the lineage into a tractable form that allows efficient probability computation~\cite{Suciu:PDB:11}.

A remarkable result is the following computational dichotomy~\cite{Dalvi:VLDB:04}: Let $Q$ be a Boolean conjunctive query without repeating relation symbols and $\mathcal{W}$ any probabilistic database. If $Q$ is hierarchical, then its data complexity is polynomial time. If $Q$ is non-hierarchical, then its data complexity is hard for \#P. An immediate generalization holds for non-Boolean conjunctive queries, by checking whether their residual Boolean queries (obtained by fixing the free variables to constants) are hierarchical~\cite{OlteanuHK09}.

An implication of this dichotomy is that CQAPs with hierarchical fracture can be computed in polynomial time data complexity over probabilistic databases. A natural question is whether they can be also maintained with constant update time and constant enumeration delay. As shown in Section~\ref{sec:maintenance-CQAP0}, to achieve these maintenance desiderata, we need the three properties from Definition~\ref{def:CQAP_zero}: the query fracture is hierarchical, free-dominant, and input-dominant. Tractability in the static setting only requires the hierarchical property.

%%%%%%%%%%%%%%%%
\subsection{Probabilistic Update Semantics} 
\label{sec:pdb-update-prob-semantics}

Prior work~\cite{DBLP:conf/pods/BerkholzM21} considers an update semantics for probabilistic databases that is deterministic in case of deletions and probabilistic in case of insertions. Given an insertion of a tuple $t$ with probability $p(t)$ in a probabilistic relation $R$, the tuple is inserted into $R$ as an independent event $t$ with probability $p(t)$. Given a deletion of a tuple $t$ from a probabilistic relation $R$, the tuple is removed from $R$ if it exists in $R$, regardless of its probability; if $t$ does not exist in $R$, no action is taken.

A natural interpretation of single-tuple updates, which agrees with the possible worlds semantics of probabilistic databases, is that of independent probabilistic events: Given an insertion (deletion) of a tuple $t$ with probability $p$, we insert $t$ in (delete $t$ from) the database with probability $p$ and ignore the update with probability $1-p$. 

\begin{exa}\label{ex:prob-set-semantics}
    Consider a probabilistic database consisting of a tuple $t$ with probability $1/2$. We consider two scenarios: We either insert or delete $t$ with probability $1/4$.
    
    {\em The insertion case.} We have four possible worlds, depending on whether each of the two events holds: (1) $t$ is in the database and we ignore the insert; this world consists of the tuple $t$ and has the probability $1/2\cdot (1-1/4) = 3/8$; (2) $t$ is in the database and the insertion is triggered; this world consists of $t$ and has probability $1/2\cdot 1/4 = 1/8$; (3) $t$ is not in the database and the insertion is ignored; this world is empty and has probability $(1-1/2)\cdot(1 - 1/4) = 3/8$; (4) $t$ is not in the database and the insertion is triggered; this world consists of $t$ and has probability $(1-1/2)\cdot 1/4 = 1/8$. As expected, the sum of the probabilities of all four worlds is 1. The third world is empty, all other worlds consist of $t$. The probability of $t$ is the sum of the probabilities of all worlds except the third world, or equivalently 1 minus the probability of the third world: $1 - 3/8 = 5/8$. 
    
    {\em The deletion case.} We again have four possible worlds, depending on whether each of the two events holds: (1) $t$ is in the database and we ignore the delete; this world consists of the tuple $t$ and has the probability $1/2\cdot (1-1/4) = 3/8$; (2)  $t$ is in the database and the deletion is triggered; this world is empty and has probability $1/2\cdot 1/4 = 1/8$; (3) $t$ is not in the database and the deletion is ignored; this world is empty and has probability $(1-1/2)\cdot(1 - 1/4) = 3/8$; (4) $t$ is not in the database and the deletion is triggered, albeit with no effect; this world has probability $(1-1/2)\cdot 1/4 = 1/8$. As expected, the sum of the probabilities of all four worlds is 1. Out of them, only the first world has $t$, so the probability that $t$ is in the database after the update is the probability of this world, which is $3/8$.
    \qed
\end{exa}

We can generalise Example~\ref{ex:prob-set-semantics}. Given a single-tuple update $t\mapsto p$, we update the probabilistic database as follows. If the update is an insertion and $t$ is already in the database with probability $p'$, then the updated database contains $t$ with probability $p + p' - p\cdot p' = 1 - (1 - p)(1-p')$; this corresponds to the probability of those worlds where at least one of the two holds: (i) the tuple is inserted and (ii) the tuple is in the database. If the database has no event $t\mapsto p'$ before the update, or equivalently $p'=0$, then after the insertion the database contains $t$ with probability $p$.  If the update is a deletion and $t$ is already in the database with probability $p'$, then the updated database contains $t$ with probability $p'\cdot (1-p)$; this corresponds to the probability of the world where $t$ is in the database and the deletion is not triggered. If the database has no event $t\mapsto p'$ before the update, or equivalently $p'=0$, then before and after the deletion the database does not contain $t$ (and the deletion has no effect). The above behaviour holds regardless of other possible tuples in the database, since they are independent of both $t$ and the update.

We call this semantics the {\em probabilistic set semantics}: It interprets each update as an independent probabilistic event and uses set semantics (no duplicates) within each world. The query maintenance mechanism put forward in Section~\ref{sec:probabilistic} can propagate updates from the input relations up the view trees constructed for any query in CQAP$_0$ using this probabilistic set semantics for updates in constant time, while allowing for constant-delay enumeration of the query result after each update. 

A shortcoming of the probabilistic set semantics, as already apparent in Example~\ref{ex:prob-set-semantics}, is that the order of updates matters: Given two updates, one deleting $t$ and one inserting $t$, then the two possible orders of updates yield different databases. Furthermore, the semantics ignores the multiplicity of a tuple in a possible world, so this semantics does not generalise the relational case discussed in the previous sections, where we maintain tuple multiplicities to ensure correct maintenance and accommodate out-of-order updates, to the probabilistic setting. In particular, this means that a possible world, where we trigger several insertions of the same tuple $t$ followed by one deletion of $t$, is empty. Also, if we were to first delete and then insert, then the deletion is lost and therefore has no effect.

The probabilistic set semantics can be generalised to avoid the two aforementioned pitfalls: Instead of maintaining the probability of a tuple being in (and missing from) the database, we maintain the discrete probability distribution over its possible multiplicities: $\{(i,p_i) \mid i\in\mathbb{Z}\}$, where $p_i$ is the probability that the tuple has multiplicity $i$, $p_i\neq 0$ for finitely many $i$ values, and  $\sum_i p_i = 1$.
Like in the relational case in the previous sections, the multiplicity is an integer and  captures the number of insertions and deletions of a tuple in the database; for derived tuples in views defined over the database, it captures the number of derivations from the input tuples. This generalisation is the {\em probabilistic bag semantics}.

\begin{exa}\label{ex:prob-bag-semantics}
    Consider now a probabilistic database that contains a tuple $t$ whose probability distribution over its multiplicities is: $\{ (2,p_2),(1,p_1),(0,p_0),(-1,p_{-1})\}$. That is, tuple $t$ has multiplicity $i$ with probability $p_i$, for $-1\leq i\leq 2$. We again consider two scenarios: We either insert or delete $t$ with probability $p$.
    
    {\em The insertion case. }
    The new probability distribution over the multiplicities of $t$ becomes: 
    $\{(3,p_2\cdot p), (2, p_2\cdot(1-p)+p_1\cdot p), (1, p_1\cdot(1-p)+p_0\cdot p), (0,p_0\cdot(1-p)+p_{-1}\cdot p), (-1, p_{-1}\cdot(1-p))\}$. 
    The tuple has multiplicity $i$ after the insertion if either (1) it had multiplicity $i$ before the insertion and the insertion is not triggered, or (2) it had multiplicity $i-1$ before the insertion and the insertion is triggered. In the first case, the probability is the product of the probability $p_i$ that the tuple is in the database with multiplicity $i$ and of the probability $1-p$ that the insertion is not triggered. The product here is correct since the two events are independent, and they must both occur. Similarly, in the second case, the two cases are mutually exclusive events, so their joint probability is the sum of their probabilities.  

    {\em The deletion case. }
    If we delete $t$ with probability $p$, then the new probability distribution over the multiplicities of $t$ becomes: $\{(2,p_2\cdot(1-p)), (1, p_2\cdot p + p_1\cdot(1-p)), (0,p_1\cdot p + p_0\cdot(1-p)), (-1, p_0\cdot p + p_{-1}\cdot (1-p)), (-2,p_{-1}\cdot p)\}$. The reasoning is similar to that of insertion.
    \qed
\end{exa}

We can generalise Example~\ref{ex:prob-bag-semantics} to formally define probability distributions over multiplicities and the operations on them, as detailed in Appendix~G of the technical report~\cite{access_pattern_arxiv}. 
A drawback of the probabilistic bag semantics is that the probability distribution associated with each tuple in an input relation can grow linearly with the number of updates; for tuples in views, their probability distributions can grow polynomially (in data complexity) with the number of updates. As a result, both the update and the enumeration steps are expensive. 

To alleviate the computational complexity brought by the probabilistic bag semantics for updates, we can maintain the expectation and variance of the probability distributions over tuple multiplicities, instead of storing and maintaining the full distributions. We then associate each tuple with a pair of the expected value and the variance of its multiplicity. We refer to this refinement as the {\em expectation-variance update semantics}.

Under the expectation-variance update semantics, 
an insertion of a tuple $t$ with probability $p$ is a random event, where the multiplicity of $t$ is a random variable $X$ with the probability distribution $\{ (0, 1-p), (1,p) \}$.
% $X$ with the probability distribution over its multiplicities: $\{ (0, 1-p), (1,p) \}$. 
By definition, the expectation of $X$ is given by
$\text{E}[X] = 0\cdot(1-p) + 1\cdot p=p$ and 
the variance of $X$ is given by
$\text{Var}[X] = \text{E}[X^2] - \text{E}[X]^2 = 1^2 \cdot p - p^2 = p(1-p)$. 
Similarly, for a deletion of a tuple $t$ with probability $p$, the multiplicity of $t$ is a random variable $Y$ with the probability distribution $\{ (0, 1-p), (-1,p) \}$. Then, by definition, $\text{E}[Y]=-1 \cdot p = -p$ and $\text{Var}[Y] = \text{E}[Y^2] - \text{E}[Y]^2 = (-1)^2 \cdot p - (-p)^2 = p(1-p)$.

To compute the expectation and variance of the multiplicity of tuple in a view, we exploit properties of the sum and product of two independent random variables $X$ and $Y$: 
% For two independent random variables $X$ and $Y$, the following properties hold:
\begin{align*}
    \text{E}[X + Y] &= \text{E}[X] + \text{E}[Y] \\
    \text{Var}[X + Y] &= \text{Var}[X] + \text{Var}[Y] \\
    \text{E}[X Y] &= \text{E}[X] \, \text{E}[Y] \\
    \text{Var}[X Y] &= \text{Var}[X]\,\text{Var}[Y] + \text{Var}[X]\,\text{E}[Y]^2 + \text{Var}[Y]\,\text{E}[X]^2
\end{align*}

Inserting (deleting) a tuple with some probability increases (decreases) the expected value of the tuple's multiplicity. Note that the expected multiplicity can be negative.

\begin{exa}\label{ex:prob-expectation-semantics}
    Consider a probabilistic database under the expectation-variance update semantics, where each tuple is paired with the expected value and variance of its multiplicity. 
    Since updates are independent probabilistic events, the expected value and variance after an update can be computed using the above properties of the sum of two independent random variables.
    When a tuple is updated with probability $p$, its expected multiplicity increases by $p$ if the update is an insertion and decreases by $p$ if the update is a deletion, while the variance of its multiplicity increases by $p(1-p)$ in both cases.
    \qed
\end{exa}

We can define two binary operations, $\oplus$ and $\odot$, to compute the expected value and variance of the sum and respectively product of two independent random variables, given the (expected value, variance) pairs of the two variables.

\begin{defi}\label{def:expprob}
Define binary operations $\oplus: \mathbb{R}^2 \times \mathbb{R}^2 \to \mathbb{R}^2$ and $\odot: \mathbb{R}^2 \times \mathbb{R}^2 \to \mathbb{R}^2$ as:
\begin{align*}
    (a,b) \oplus (c,d) &= (a+c, b+d) \\
    (a,b) \odot (c,d) &= (ac, bd + a^2d + bc^2)    
\end{align*}
\end{defi}
Both operations execute in constant time. 

\begin{prop}\label{prop:expprobmonoid}
The structures $(\mathbb{R}^2,\oplus,(0,0))$ and $(\mathbb{R}^2,\odot,(1,0))$ are commutative monoids.
\end{prop}

The query maintenance mechanism in Section~\ref{sec:probabilistic} can propagate updates from the input relations to the result of any query in CQAP$_0$ using the $\oplus$ and $\odot$ operations from Definition~\ref{def:expprob} under the probabilistic expectation-variance semantics for updates.

%%%%%%%%%%%%%%%%

%% file: probabilistic.tex
\section{Dynamic Evaluation for CQAP$_0$ over Probabilistic Databases}
\label{sec:probabilistic}
\label{sec:maintenance-CQAP0}

    We here show that any query in CQAP$_0$ without repeating relation symbols can be maintained over a probabilistic database with constant update time and enumeration delay under both the set semantics and the expectation-variance semantics for updates in probabilistic databases. We conclude with a discussion on the maintenance under the probabilistic bag semantics.

Our main insight is that the exact same maintenance approach used for queries in CQAP$_0$ over relational databases is also applicable for such queries over probabilistic databases. There are two main reasons for this. (1) The probability of any Boolean hierarchical conjunctive query without repeating relation symbols can be computed using two operators, independent-project and independent-join~\cite{Suciu:PDB:11}. 
Independent-project computes the disjunction of independent events,
as discussed in Section~\ref{sec:pdb-update-prob-semantics} for the various update semantics considered.
Independent-join computes the conjunction of independent events, which is possible if we only join distinct probabilistic relations (no self-joins). 
(2) For queries in CQAP$_0$ over probabilistic databases, we can construct trees of views that are tuple-independent relations and constructed only using projections, which employ the independent-project operator, and one-to-one joins, which employ the independent-join operator.

\begin{exa}
    Recall the query $Q_1(B,C,D|A_1) = R(A_1,B,C), S(A_1,B,D)$ in CQAP$_0$, whose hypergraph is depicted in Figure~\ref{fig:general_hypergraphs} (middle), its canonical access-top variable order is in Figure~\ref{fig:general_CPAP_0} (left), and its view tree is in Figure~\ref{fig:general_CPAP_0} (middle). The probabilistic relations $R$ and $S$ consist of pairwise independent tuples. The view $V'_C(A_1,B)$ is created by projecting away $C$ from the relation $R$. The projection may create duplicates, i.e., tuples $(a_1,b)\mapsto p_i$ with the same pair of values $(a_1,b)$ for the variables $(A_1,B)$ in $V'_C$ and $i\in[n]$. Since these tuples are pairwise independent, we can replace them by one tuple $t\mapsto 1 - \prod_{i\in[n]}(1-p_i)$ under the probabilistic set semantics. The tuples in $V'_C$ remain pairwise independent, even after removing the duplicates as explained above. The same treatment applies to the view $V'_D(A_1,B)$, which is created by projecting away $D$ from the relation $S$. The view $V_B(A_1,B)$ is the intersection of the views $V'_C$ and $V'_D$. Each resulting tuple appears in both views. Its probability is the product of its probabilities in the two views.
    Since the tuples in $V_B$ result from distinct tuples in the child views, they are pairwise independent. Finally, the view $V_{A_1}(A_1)$ is created by projecting away $B$ from $V_B$. The duplicates are merged into a single tuple whose probability is the probability of the disjunction of its duplicates. Again, the distinct tuples in $V_B$ are pairwise independent, since they originate from disjoint sets of tuples in the child view $V_B$.
\qed
\end{exa}

The following statement captures the property that the view trees for queries in CQAP$_0$ have independent tuples.

\begin{prop}\label{prop:pdb-cqap0-independence}
    Given a query $Q$ in CQAP$_0$  without repeating relation symbols, a canonical access-top variable order $\omega$ for $Q$, and the view tree $\tau(\omega)$ for $Q$ over a probabilistic database $D$. Then the tuples in each view of $\tau(\omega)$ are pairwise independent.
\end{prop}

\begin{proof}
    We show this using induction on the structure of the view tree built for $Q$.

    {\bf Base case:} By definition, the tuples in the input relations are pairwise independent. Updates to the input relations also preserve the independence of the tuples under both probabilistic set and bag semantics, as discussed in Section~\ref{sec:pdb-update-prob-semantics}.
    
    {\bf Inductive step:} We assume the tuples in the child views are pairwise independent and show this to be the case also for parent views. 
    
    For a query in CQAP$_0$, the view tree construction in Figure\@~\ref{fig:general_view_tree_construction} has three cases of inner views: (1) It either constructs a parent view that is a copy of the child view (this is not needed but keeps the algorithm and its analysis simpler); (2) Alternatively, it creates a parent view that is a projection of the child view; (3) The last case is a parent view that is the intersection of several child views. We analyse each of these cases next. 
    
    (1) This holds trivially by the induction hypothesis.

    (2) A projection may create duplicates, which can be merged into one common tuple whose probability is that of the disjunction of the duplicates. Since the tuples in the view can be partitioned into disjoint sets of duplicates, such that within each set and across sets the tuples are independent, the output tuples for the different sets are pairwise independent.

    (3) An intersection of several child views is a 1-1 join. Since the input relations are distinct (no self-join), the tuples across all the child views are pairwise independent. The intersection then yields the subset of these tuples that appear in all child views, so these tuples remain pairwise independent. The probability of each such tuple is the product of the probabilities of the tuple in each of the child views.
\end{proof}

Proposition~\ref{prop:pdb-cqap0-independence} essentially states that the view trees we create in case of queries in CQAP$_0$ correspond to so-called safe plans used for efficient probability computation~\cite{Dalvi:VLDB:04,Suciu:PDB:11}. The safe plans are however not enough for efficient maintenance, i.e., for constant update time and constant enumeration delay. Indeed, CQAP$_0$ further constraints the queries to be free-dominant and input-dominant.

\begin{thm}\label{th:pdb-maintenance}
    Given a  query $Q$ in CQAP$_0$ without repeating relation symbols and a probabilistic database $D$ of size $N$, then $Q$ can be maintained with $\bigO{N}$ preprocessing time, $\bigO{1}$ time for single-tuple updates, and $\bigO{1}$ enumeration delay, using the probabilistic set semantics 
    or the probabilistic expectation-variance semantics for updates. 
\end{thm}

\begin{proof}
We first prove the theorem under the probabilistic set semantics and then extend the proof to the probabilistic expectation-variance semantics.

Consider a  query $Q(\calO|\calI)$ in CQAP$_0$ without repeating relation symbols, its fracture $Q_{\dagger}(\calO|\calI_{\dagger})$, and a database $D$  of size $N$. Theorem~\ref{thm:dichotomy} states that $Q$ can be maintained with $\bigO{N}$ preprocessing time, 
$\bigO{1}$ update time, and  $\bigO{1}$ enumeration delay, in case the database $D$ consists of relations that map tuples to 
 multiplicities, as defined in our data model in Section \ref{sec:preliminaries}.
These complexities are achieved by our approach described in 
Sections~\ref{sec:preprocessing}-\ref{sec:updates}.   
In case $D$ is a probabilistic database, we can  achieve the same complexities using the same maintenance approach with two
twists, which we  explain next.

We distinguish the following cases when handling a single-tuple update $t \mapsto p$:
\begin{enumerate}
\item {\em Updating a base relation}. Assume that the probability of tuple $t$ being in the database is $p^{old}$, then we compute the probability $p^{new}$ of tuple $t$ being in the database after the update as discussed in Section~\ref{sec:pdb-update-prob-semantics}: 
\begin{align*}
p^{new} = 
\begin{cases}
    1 - (1 - p^{old}) \cdot (1 - p), & \text{ if the update is an insertion} \\
    p^{old} \cdot (1 -p), & \text { if the update is a deletion.}
\end{cases}
\end{align*}
Computing $p^{new}$ takes constant time, assuming that the basic arithmetic operations can be performed in constant time. We then propagate the information that the probability of tuple $t$ has changed from $p^{old}$ to $p^{new}$ further up in each affected view tree.

\item {\em Updating a view $V$}. 
Since the query $Q$ is in CQAP$_0$, each view represents the intersection of its child views, possibly followed by an aggregation that projects away variables. 
A single-tuple update coming from one child yields a change containing at most one tuple whose probability is the product of the probabilities of the joined tuples. 

While in the relational case, the aggregation amounts to summing up the multiplicities of duplicates, in the probabilistic case the aggregation amounts to computing the probability of the disjunction of independent events corresponding to $k$ duplicate tuples, as per Proposition~\ref{prop:pdb-cqap0-independence}. 

Consider $k$ such independent events with probabilities $p_1, \ldots p_k$. Their joint probability is $p=1-\prod_{i\in[k]}(1-p_i)$. 
Our goal is to maintain $p$ whenever any $p_i$ changes. If any of the duplicates is certain (i.e., $p_i=1$), then the joint probability becomes 1, effectively disregarding all other probabilities. This does not encode how many of these duplicates are certain. To ensure we maintain the correct joint probability $p$ under changes to any of the probabilities $p_1,\ldots,p_k$, we associate the tuple $t$ in the view $V$ resulting from the aggregation of the $k$ duplicates with a pair $(q, m)$, where $q$ is the product term $\prod_{i\in[k]: p_i < 1} (1-p_i)$, which involves only the probabilities of the uncertain duplicates, and $m$ is the number of the certain duplicates. Then, the joint probability $p$ is $1$ when $m > 0$, i.e., when there is at least one certain duplicate, and $1-q$ when $m=0$, i.e., when there is no certain duplicate and then $p = 1 - q = 1 - \prod_{i\in[k]} (1-p_i)$.

Assume now that the probability $p_i$ changes from $p^{old}$ to $p^{new}$. We compute for the tuple $t$ the new pair $(q^{new}, m^{new})$ from the current pair $(q^{old}, m^{old})$ as follows:
$$
(q^{new}, m^{new}) = 
\begin{cases}
    (q^{old}, m^{old}) & \text{if } p^{old} = 1 \land p^{new} = 1 \\[3pt]
    (q^{old}\cdot(1-p^{new}), m^{old} - 1) & \text{if } p^{old} = 1 \land p^{new} < 1 \\[3pt]
    (\frac{q^{old}}{1-p^{old}}, m^{old} + 1) & \text{if } p^{old} < 1 \land p^{new} = 1 \\[3pt]
    (\frac{q^{old}}{1-p^{old}}\cdot(1-p^{new}), m^{old}) & \text{if } p^{old} < 1 \land p^{new} < 1
\end{cases}
$$
In each case computing $(q^{new}, m^{new})$ takes constant time. From $(q^{old}, m^{old})$ and $(q^{new}, m^{new})$, computing the aggregated probability before and after the update, needed for subsequent propagation in the view tree, also takes constant time. 
Overall, the overhead added to the maintenance cost by the probability computation is constant. 
The correctness of the probabilities assigned to tuples in views follows from Proposition~\ref{prop:pdb-cqap0-independence}.
\end{enumerate}

Our enumeration procedure from Section~\ref{sec:enumeration} reports for any input tuple,
 all tuples in the query result with constant delay. 
In case of probabilistic databases, we  need to report for each output tuple also its probability. 
We explain in the following how to compute these probabilities. 
Assume that our view trees constructed in the preprocessing stage follow an access-top VO $\omega$ for 
 $Q_{\dagger}(\calO|\calI_{\dagger})$ that consists of the trees $\omega_1, \ldots , \omega_n$. 
 Let $T_1 = \tau(\omega_1), \ldots, T_n = \tau(\omega_n)$ be the view trees constructed using the 
 procedure $\tau$ in Figure~\ref{fig:general_view_tree_construction}. 
 For each $j \in [n]$, let $Q_j(\calO_j \mid \calI_j)$ with 
 $\calO_j = \calO \cap \vars(\omega_j)$ and $\calI_j = \calI_{\dagger} \cap \vars(\omega_j)$ be the CQAP that joins 
 the atoms appearing at the leaves of $T_j$. In Section~\ref{sec:enumeration}, we explain how for any $j \in [n]$
 and $\inst{i}_j$ over $\calI_j$, the tuples in $Q_j(\calO_j | \inst{i}_j)$ can be enumerated with constant 
 delay using the view tree $T_j$. For each such tuple $\inst{t}_j \in Q_j(\calO_j | \inst{i}_j)$, we can traverse $T_j$ to compute its probability as follows.     
We first check whether the schema of $\inst{t}_j$ is equal to the schema of the root view $V$. If yes,
 the probability of $\inst{t}_j$ is $V(\inst{t}_j)$. Otherwise, let $\hat{T}_1, \ldots , \hat{T}_k$ be the child trees
 of the root view $V$, and let $\hat{\inst{t}}_{j}^1, \ldots, \hat{\inst{t}}_{j}^k$ be the restrictions of $\inst{t}_j$ onto the variables of $\hat{T}_1, \ldots , 
 \hat{T}_k$, respectively. 
    We recursively compute the probabilities $p_1, \ldots, p_k$ of the tuples 
    $\hat{\inst{t}}^1_j, \ldots , \hat{\inst{t}}^k_j$, and set the probability 
    of $\inst{t}_j$ to $\Pi_{i \in [k]} p_i$.

Consider now a tuple $\inst{i}$ over $\calI$. The set of tuples in $Q(\calO | \inst{i})$ is equal to the Cartesian product  
$\times_{j \in [n]} Q_j(\calO_j | \inst{i}_j)$,
where  $\inst{i}_j[X'] = \inst{i}[X]$ if $X = X'$ or $X$ is replaced by $X'$ when constructing the fracture of $Q$.  Section~\ref{sec:enumeration} explains how to enumerate the tuples in this Cartesian product with constant delay, given that the tuples in each $Q_j(\calO_j | \inst{i}_j)$ can be enumerated with constant delay. 
It remains to explain how to obtain the probability of a tuple $\inst{t}$ that is the concatenation of tuples 
$\inst{t}_1 \in Q_1(\calO_j | \inst{i}_j), \ldots ,\inst{t}_n \in Q_n(\calO_n | \inst{i}_n)$.
% Consider a tuple $\inst{t}$ that is the concatenation of tuples 
% $\inst{t}_1 \in Q_1(\calO_j | \inst{i}_j), \ldots ,\inst{t}_n \in Q_n(\calO_n | \inst{i}_n)$.
% It remains to explain how to obtain the probability of $\inst{t}$.
Since the query is without self-joins, the constructed view trees are over disjoint sets of relations, so the tuples in the views of one view tree are independent of the tuples in the views of another view tree. Thus, the probability of $\inst{t}$ is the product of the probabilities of $\inst{t}_1, \ldots , \inst{t}_n$.

The maintenance approach for queries in CQAP$_0$ extends to updates under the probabilistic expectation-variance semantics.
Each tuple $t$ in the database is associated with a pair representing the expectation and variance of the tuple's multiplicity.
An update of a tuple $t$ with probability $p$ changes the expected multiplicity by $p$ for an insertion and by $-p$ for a deletion, and increases the variance by $p(1-p)$.
Updating a base relation involves summing up the expectation-variance pairs for $t$ using the $\oplus$ operation from Definition~\ref{def:expprob}.
The correctness of this approach follows from the linearity of expectation and the linearity of variance given that all updates are independent probabilistic events, see Section~\ref{sec:probabilistic-semantics}.

We then propagate the new expectation-variance pair further up in each affected view tree. A single-tuple update from a child results in at most one affected tuple, whose expectation-variance pair is computed using the sum $\oplus$ and product $\otimes$ operations from Definition~\ref{def:expprob} in case of projection and respectively join. These operations take constant time as they only combine two expectation-variance pairs. Thus, updates are propagated in each affected view tree in constant time.
The enumeration procedure under the expectation-variance semantics follows a similar approach to that used in set semantics.
\nop{
    We then propagate the new expectation-variance pair further up in each affected view tree. Following similar reasoning as above, a single-tuple update from a child results in at most one affected tuple, whose expectation-variance pair is computed as the product of the expectation-variance pairs of the joined tuples using the $\odot$ operation from Definition~\ref{def:expprob}. This product computation takes constant time. Thus, updates are propagated in each affected view tree in constant time.
The enumeration procedure under the expectation-variance semantics follows a similar approach to that used in set semantics.
}
\end{proof}

The same maintenance approach can be applied to queries in CQAP$_0$ over probabilistic databases under the probabilistic bag semantics for updates. Each tuple is mapped to a probability distribution from a set $S$ of probability distributions, and the constructed views are evaluated using two binary operation, addition and multiplication, defined over $S$; Appendix~G of the technical report provides further details of these operations~\cite{access_pattern_arxiv}. 

Theorem~\ref{th:pdb-maintenance} can be extended to the probabilistic bag semantics, albeit at a higher cost for maintenance and enumeration. When a tuple is updated, the support of its associated probability distribution grows, making the distribution's size dependent on the database size. Consequently, both addition and multiplication operations over distributions require non-constant time, leading to non-constant time for update propagation and enumeration.

%% file: related.tex
\section{Related Work}
\label{sec:relatedwork}

Our work is the first to investigate the dynamic evaluation for queries with access patterns.

\paragraph{Free Access Patterns} Our notion of queries with free access patterns corresponds to parameterized queries~\cite{AbiteboulHV95}. These queries have selection conditions that set variables to parameter values to be supplied at query time.
Prior work closest in spirit to ours investigates the space-delay trade-off for the static evaluation of full conjunctive queries with free access patterns~\cite{Deep:2018}. It constructs a succinct representation of the query output, from which the tuples that conform with value bindings of the input variables can be enumerated. \nop{It relies on a hypertree decomposition of the query where the input variables form a connected subtree.} It does not support queries with projection nor dynamic evaluation. Follow-up work considers the static evaluation for Boolean conjunctive queries with access patterns\nop{, where every free variable is fixed to a constant at query time}~\cite{Deep:arxiv:21}. 
Further works on queries with access patterns~\cite{Florescu:SIGMOD:99, Yerneni:ICDT:99, Deutsch:TCS:2007, Benedikt:VLDB:15, Benedikt:PODS:14} consider the setting where {\em input} relations have input and output variables and there is no restriction on whether they are bound or free; also, a variable may be input in a relation and output in another. This poses the challenge of whether the query can be answered under specific access restrictions~\cite{Nash:FOAccess:04, Nash:UCQAccess:04, Li:ICDT:2001}.

\paragraph{Dynamic evaluation}
Our work generalises the dichotomy for $q$-hierarchical queries under updates~\cite{BerkholzKS17} and the complexity trade-offs for  queries under updates~\cite{Kara:ICDT:19,Kara:TODS:2020,KaraNOZ2020}. The IVM approaches Dynamic Yannakakis~\cite{Idris:dynamic:SIGMOD:2017} and F-IVM~\cite{Nikolic:SIGMOD:18}, which is implemented on top of DBToaster~\cite{DBT:VLDBJ:2014}, achieve (i) linear-time  preprocessing, linear-time single-tuple updates, and constant enumeration delay for free-connex acyclic queries; and (ii) linear-time preprocessing, constant-time single-tuple updates, and constant enumeration delay for $q$-hierarchical queries. Theorem~\ref{thm:general} recovers these results by noting that the static and dynamic widths are: $1$ and respectively in $\{0,1\}$ for free-connex acyclic queries and $1$ and respectively $0$ for $q$-hierarchical queries.
We refer the reader to a comprehensive comparison~\cite{Trade_Offs_LMCS23} of dynamic query evaluation techniques and how they are recovered by the trade-off~~\cite{KaraNOZ2020}  extended in our work. 

Our CQAP$_0$ dichotomy strictly generalises the one for $q$-hierarchical queries~\cite{BerkholzKS17}: The set of $q$-hierarchical queries is a strict subset of CQAP$_0$, while there are hard patterns of non-CQAP$_0$ beyond those for non-$q$-hierarchical queries. 

There are key technical differences between the prior framework for dynamic evaluation trade-off~\cite{KaraNOZ2020} and ours: different data partitioning; new modular construction of view trees; access-top variable orders; new iterators for view trees modelled on any variable order. We create a set of variable orders that represent heavy/light evaluation strategies and then map them to view trees. One advantage is a simpler complexity analysis for the views, since the variables orders and their view trees share the same width measures.

\paragraph{Cutset optimisations} Cutset conditioning~\cite{DBLP:books/daglib/0066829} and cutset sampling~\cite{DBLP:journals/jair/BidyukD07} are used for efficient exact and approximate inference in Bayesian networks. The idea is to \emph{choose} a cutset, which is a subset of variables, such that conditioning on the variables in the cutset, i.e., instantiating them with possible values, yields a network with a small treewidth that allows exact inference. The set of input variables of a CQAP can be seen as a \emph{given} cutset, while fixing the input variables to given values is conditioning. Query fracturing, as introduced in our work, is a query rewriting technique that does not have a counterpart in cutset optimisations in AI.

%% file: conclusion.tex
%%%%%%%%%%%%%%%%
%Conclusion
%%%%%%%%%%%%%%%%
\section{Conclusion and Outlook}

This paper introduces a fully dynamic evaluation approach for conjunctive queries with free access patterns. It gives a syntactic characterisation of those queries that admit constant-time update and delay and further investigates the trade-off between preprocessing time, update time, and enumeration delay for such queries.

The work presented in this article can be extended naturally in a number of ways.

\paragraph{Adaptive maintenance.} The computational complexity of static query evaluation can be asymptotically improved by combining several execution strategies (query plans, hypertree decompositions, variable orders, or view trees) for the same query, where each strategy is adapted to a different part of the data. Such adaptive strategies can also benefit dynamic query evaluation. Our optimality results for CQAP$_0$ and CQAP$_1$ rely in fact on maintaining several view trees for one query, each view tree adapted to a different (light or heavy) part of the data. Extending our optimality results to queries beyond CQAP$_0$ and CQAP$_1$ is likely to require data partitioning and adaptive maintenance. Yet the view trees used by our approach to maintain one query can be derived from {\em one} variable order. Using several variable orders may lead to improved complexity for queries beyond CQAP$_0$ and CQAP$_1$. There are several technical challenges to overcome when translating existing adaptive approaches for static query evaluation, e.g.,~\cite{Khamis:PODS:17,ZhaoDK23,Xiao:OutputSensitiveYannakakis:2024}, to dynamic query evaluation, including: Can the cost of regular rebalancing of heavy-light partitions be amortised over the sequence of updates so that it does not increase the single-tuple update time? Can the multiplicities of each tuple in the views and input relations be maintained as efficiently as the maintenance of the tuple itself? To appreciate the difficulty of addressing the latter question, note the blow-up in the number of variable orders needed by the PANDA adaptive strategy for static query evaluation~\cite{Khamis:PODS:17}: This is exponential in the query size and poly-logarithmic in the data size, so the poly-logarithmic factor in the data size carries over to the enumeration delay. A possible solution is to consider a restriction of PANDA~\cite{FAQAI:TODS:2020}, which ensures that different variable orders yield disjoint sets of query output tuples, albeit with a higher complexity than PANDA.

\paragraph{Beyond hierarchical queries.} An open research question is the generalisation of our maintenance trade-off for {\em all} CQAPs as well as of the optimality for {\em all} CQAPs. 
\nop{beyond the CQAP$_0$ and CQAP$_1$ classes. }
The recent trade-off between preprocessing time and enumeration delay for $\alpha$-acyclic conjunctive queries in the static setting~\cite{Tradeoff:CSL:2023} can be extended to also consider the update time and also to apply to arbitrary conjunctive queries and CQAPs.

\paragraph{Support for aggregates.} Our approach requires the maintenance of the multiplicities (the number of derivations or counts) of tuples in each view and input relation. Section~\ref{sec:probabilistic} also shows how to maintain tuple probabilities in case of queries in the CQAP$_0$ class. More generally, our approach can support any aggregate expressible using the sum and product operations of a ring, as detailed in the F-IVM system~\cite{FIVM:VLDBJ:2024}.

\paragraph{Beyond probabilistic databases.} Our maintenance approach introduced in Section~\ref{sec:probabilistic} can be extended beyond probabilistic databases. A special case of probabilistic databases is when all probabilities are $\frac{1}{2}$, so the probability distributions are uniform. This corresponds to the {\em model counting} problem~\cite{DBLP:series/faia/GomesSS21}: Given a Boolean query, in case of arbitrary probability distributions we compute the probability of the query to be true, whereas in case of uniform probability distributions we compute the fraction of those possible worlds where the query is true.
    An immediate corollary of the dichotomy for conjunctive queries without repeating relation symbols in probabilistic databases~\cite{Dalvi:VLDB:04,Suciu:PDB:11} is that model counting can be computed efficiently for hierarchical queries. Our work complements this result in the static setting with a similar result in the dynamic setting: A corollary of Theorem~\ref{th:pdb-maintenance} is that model counting for CQAP$_0$ can be maintained with linear preprocessing time, constant update time, and constant enumeration delay. This also immediately implies that further tasks, which can be expressed using model counting, can immediately benefit from the efficient maintenance approach put forward in our work. Prime examples are the computation of the Shapley and Banzhaf values of database tuples, whose computation in relational databases is polynomial-time equivalent to model counting~\cite{KOS:PODS:2024} and is in particular tractable for hierarchical queries~\cite{LivshitsBKS21,DeutchFKM22,DeutchFKM22,DBLP:journals/corr/abs-2308-05588}.